\documentclass[11pt]{amsart}
\usepackage[letterpaper, margin=1in]{geometry}
\usepackage{hyperref}
\usepackage{cleveref}
\usepackage[disable,colorinlistoftodos,bordercolor=orange,backgroundcolor=orange!20,linecolor=orange,textsize=scriptsize]{todonotes}
\usepackage{graphicx, url,enumerate}
\usepackage{enumitem}
\usepackage{amssymb,mathtools}
 \renewcommand{\geq}{\geqslant}
 \renewcommand{\leq}{\leqslant}
 \renewcommand{\ge}{\geqslant}
 \renewcommand{\le}{\leqslant}

\usepackage{algorithm}
\usepackage{algpseudocode}
\usepackage{color}
\usepackage{tikz}
\usetikzlibrary{3d,calc}
\usepackage{eulervm} 

\numberwithin{equation}{section}
\numberwithin{algorithm}{section}

\theoremstyle{plain}
\newtheorem{theorem}{Theorem}[section]
\newtheorem{proposition}[theorem]{Proposition}
\newtheorem{lemma}[theorem]{Lemma}
\newtheorem{corollary}[theorem]{Corollary}

\newtheorem{definition}[theorem]{Definition}

\theoremstyle{remark}

\let\H\undefined
\let\S\undefined

\def\<#1>{\langle #1\rangle}

\DeclareMathOperator{\tr}{Tr}
\DeclareMathOperator{\diag}{diag}
\DeclareMathOperator{\rank}{rank}
\DeclareMathOperator{\hrank}{hrank}

\DeclareMathOperator{\Gr}{Gr}

\DeclareMathOperator{\H}{H}
\DeclareMathOperator{\S}{S}

\newcommand{\bm}{\mathbf{m}}
\newcommand{\bn}{\mathbf{n}}
\newcommand{\C}{\mathbb{C}}
\newcommand{\F}{\mathbb{F}}
\newcommand{\N}{\mathbb{N}}
\newcommand{\Q}{\mathbb{Q}}
\newcommand{\R}{\mathbb{R}}
\newcommand{\Z}{\mathbb{Z}}

\newcommand{\e}{\mathbf{e}}
\newcommand{\f}{\mathbf{f}}

\newcommand{\x}{\mathbf{x}}
\newcommand{\X}{\mathbf{X}}
\newcommand{\y}{\mathbf{y}}
\newcommand{\z}{\mathbf{z}}

\newcommand{\w}{\mathbf{w}}
\newcommand{\W}{\mathbf{W}}

\newcommand{\cB}{\mathcal{B}}
\newcommand{\cC}{\mathcal{C}}
\newcommand{\cD}{\mathcal{D}}

\newcommand{\cF}{\mathcal{F}}
\newcommand{\cG}{\mathcal{G}}
\newcommand{\cH}{\mathcal{H}}
\newcommand{\cI}{\mathcal{I}}

\newcommand{\cM}{\mathcal{M}}

\newcommand{\cQ}{\mathcal{Q}}
\newcommand{\cR}{\mathcal{R}}
\newcommand{\cS}{\mathcal{S}}
\newcommand{\cT}{\mathcal{T}}

\newcommand{\cW}{\mathcal{W}}
\newcommand{\cX}{\mathcal{X}}
\newcommand{\cY}{\mathcal{Y}}
\newcommand{\cZ}{\mathcal{Z}}
\newcommand{\ba}{\mathbf{a}}

\newcommand{\bb}{\mathbf{b}}
\newcommand{\bc}{\mathbf{c}}

\newcommand{\be}{\mathbf{e}}
\newcommand{\bbf}{\mathbf{f}}
\newcommand{\bF}{\mathbf{F}}
\newcommand{\bg}{\mathbf{g}}

\newcommand{\bi}{\mathbf{i}}
\newcommand{\bj}{\mathbf{j}}
\newcommand{\bk}{\mathbf{k}}
\newcommand{\bl}{\mathbf{l}}

\newcommand{\bu}{\mathbf{u}}

\newcommand{\bv}{\mathbf{v}}
\newcommand{\bw}{\mathbf{w}}
\newcommand{\bx}{\mathbf{x}}

\newcommand{\bz}{\mathbf{z}}
\newcommand{\rA}{\mathrm{A}}
\newcommand{\rB}{\mathrm{B}}
\newcommand{\rC}{\mathrm{C}}

\newcommand{\rE}{\mathrm{E}}
\newcommand{\rG}{\mathrm{G}}
\newcommand{\rH}{\mathrm{H}}

\newcommand{\rK}{\mathrm{K}}

\newcommand{\rO}{\mathrm{O}}
\newcommand{\rP}{\mathrm{P}}

\newcommand{\rS}{\mathrm{S}}
\newcommand{\rT}{\mathrm{T}}
\newcommand{\rU}{\mathrm{U}}
\newcommand{\rV}{\mathrm{V}}
\newcommand{\rZ}{\mathrm{Z}}
\newcommand{\0}{\mathbf{0}}
\newcommand{\1}{\mathbf{1}}
\newcommand{\trans}{^\top}

\begin{document}
\title[Tensors,  entanglement,  separability, and complexity]{Tensors, entanglement,\\ separability, and their complexity}
\author{Shmuel Friedland\\
\texttt{friedlan@uic.edu}}
\address{Department of Mathematics, Statistics and Computer Science, University of Illinois, Chicago, IL 60607-7045, USA}

\subjclass[2010]{
05C50,15A69,15A75,68Q04,68Q17,68W25,81P16,81P40,90C05,90C08}

 \date{October 8, 2026}
\begin{abstract}\bf{
The aim of this paper is to show how to characterize the entanglement and separability of d-partite states, and to obtain both known and new results using the modern theory of tensors.
The geometric measure of entanglement of a pure state is one of most natural ways to quantify the entanglement, which is simply related to the spectral norm of a tensor state.   On the other hand, the logarithm of the nuclear norm of the state and density tensors can be considered as its ``energy''. We first show that the most geometric measure entangled $d$-partite state
has the minimum spectral norm and  maximum nuclear norm.   Second,  we introduce the notion of Hermitian and density tensors,  and the subspaces of bi-symmetric and bi-skew-symmetric Hermitian tensors, which correspond to Bosons and Fermions respectively.  We show that separable density tensors,  and strongly separable bi-symmetric density tensors are characterized by the value (equal to one) of their corresponding nuclear norms.   In general,  these characterizations are NP-hard to verify.  Third,  the main result of this paper to show that the above quantities are computed in polynomial time when we restrict our attention to Bosons:  symmetric $d$-qubits, or more generally to symmetric $d$-qunits in $\C^n$,  and the corresponding bi-symmetric Hermtian density tensors, for a fixed value of $n$.
}
\end{abstract}
\maketitle
\tableofcontents
\section{Introduction}\label{sec:intro}
\subsection{Statement of the main results}\label{subsec:majrt}
The aim of this paper is to show how to characterize the entanglement and separability of d-partite states, and to obtain both known and new results using the modern theory of tensors.  We mostly avoid Dirac's notation which is compact but is not always suitable for our purposes.  Our notations is closer to \cite{Lan12}.Observe that in this paper $d$ is the number of parties and $n$ is a local dimension, which is opposite to physics convention.

Let $\cH$ be an $n$-dimensional Hilbert space.  Denote $[n]:=\{1,\ldots,n\}$. We identify $\cH$ with $\C^n$ with the inner product $\langle \x,\y\rangle=\x^\dagger \y$,  where $\x^\dagger=\bar \x^\top$.
A $d$-partite space is given by $\C^{\bn}:=\otimes_{j=1}^d \C^{n_j}, \bn=(n_1,\ldots,n_d)$.  
A $d$-partite state is a tensor $\cT=[t_{i_1,\ldots,i_d}]\in\C^{\bn}, \bi=(i_1,\ldots,i_d)\in[\bn]=[n_1]\times\cdots\times[n_d]$, where $\|\cT\|=\sqrt{\sum_{\bi\in [\bn]}|t_{\bi}|^2}=1$.   A bi-partite state is given by a matrix $T=[t_{i,j}]\in \C^{n_1\times n_2}$.
Denote by $\Pi_{\bn}=\{\otimes_{j=1}^d\x_j: \x_j\in\C^{n_j}, \|\x_j\|=1, j\in[d]\}$, the product states in $\C^{\bn}$.  
The spectral norm of $\cT\in \C^{\bn}$ is defined by $\|\cT\|_{\infty}=\max_{\cX\in\Pi_{\bn}}|\langle \cX,\cT\rangle|$.  
For $d=2$, $\|T\|_{\infty}$ is the spectral norm of $T$, which is equal to $\sigma_1(T)$ the first Schmidt value: the first singular value of $T$.  The geometric measure of entanglement (GME) is the distance of a state $\cT$ to the prodiuct states $\Pi_{\bn}$ and is equal to $\sqrt{2(1-\|\cT\|_{\infty}})$.  Thus, the most entangled states with respect to GME are correspond to the minimum $\alpha_{\infty,\bn}:=\min_{\cT\in \C^{\bn}, \|\cT\|=1} \|\cT\|_{\infty}$.
For bi-partite states the most entangled states are well known.   For $n_1=n_2=n$ they are $T=\frac{1}{\sqrt{n}} U$, where $U$ is a unitary matrix.  For $d\ge 3$ we know only the following cases of most GME $d$-qubits ($n=2$): $d=3$ and $d=4$.  For $d=3$ this is the $W$-state \eqref{defWstate} \cite{TWP09}.  In subsection \ref{subsec:quadrit} we give a complete proof of the result in \cite{DFLW17} that the Higuchi-Sudbery $4$-qubit \cite{HS00} is the most GME $4$-qubit.   Let $n^{\times d}=n(1,\ldots,1)\in\N^d$. For $d$-qubit for $d\gg 1$ we show in subsection \ref{subsec:bigd} that $-2\log_2\alpha_{\infty,2^{\times d}}=d+O(\log_2d)$.  

Next we consider the GME for Bosons, which correspond to symmetric tensors states in the subspace of symmetric tensors $\rS^d\C^n\subset \C^{n^{\times d}}$.   Most GME Bosons correspond to the minimum $\alpha_{\infty, n^{\times d},s}=\min_{\cX\in\rS^d\C^n, \|\cX\|=1} \|\cX\|_{\infty}$.
Note that  $\dim\rS^d\C^n={n+d-1 \choose d}$  which is usually much smaller than $\dim \C^{n^{\times d}}=n^d$.   Fix $n\ge 2$ and consider $d\ge 1$.  Observe that $\dim\rS^d\C^n=O(d^{n-1})$.    We show that for $d\gg 1$ we have $-2\log_2\alpha_{\infty,n^{\times d},s}=O(\log_2{n+d-1\choose d})$.
We also consider the Fermions, which correspond to the skew-symmetric states $\rA^d\C^n$, which live in the exterior algebra  of wedge products.

Another measure of entanglement of states is the nuclear norm, which is the  dual norm of the spectral norm, denoted as $\|\cdot\|_{1}$ on $\C^{\bn}$.  For $d=2$ $\|T\|_1$ is the sum of Schmidt's numbers: the sum of singular values of the matrix $T$.
It can be characterized as a minimum problem \eqref{nucchar}.  The quantity $2\log \|\cT\|_1$ can be viewed as an energy of to create a state $\cT$ from the product states.
We show a known inequality $2\log\|\cT\|_1\ge -2\log\|\cT\|_{\infty}$ for a state $\cT\in\C^{\bn}$.  The most entangled state with respect to GME are also most entangled state with respect to the nuclear norm. 

Let $\rH_n\subset \C^{n\times n}$ be the real space of Hermitian matrices.  Denote by $\rH_{n,+}\supset \rH_{n,+,1}$ the cone of positive semi-definite matrices, and the convex set of positive semi-definite matrices of trace one.
Recall that the notion of a mixed state on $\C^n$ corresponds to a density matrix $\rho\in\rH_{n,+,1}$.  A pure density matrix is $\x\x^{\dagger}, \|\x\|=1$, and a mixed state is a convex combination of pure density states.
To generalize  this notion to $d$-partite mixed states, called here \emph{density tensors},  we introduce the notion of Hermitian tensors, denoted as $\rH_{\bn}\subset \C^{\bn\times \bn}:=\C^{\bn}\otimes \C^{\bn}$ if for $\cT=[t_{\bi,\bj}], \bi,\bj\in[\bn]$ the following condition hold: $t_{\bi,\bj}=\overline{t_{\bj,\bi}}$ for $\bi,\bj\in[\bn]$.
Denote by $\rH_{\bn,+}\supset \rH_{\bn,+,1}$ the cone of positive semi-definite tensors, viewed as matrices,  and the cone of positive semi-definite tensors with trace one.  A pure density tensor is of the form $\cX\otimes \overline{\cX}, \cX\in \C^{\bn}, \|\cX\|=1$.   The ``right'' norm on $\rH_{\bn}$ is $\|\cB\|_{spec}:=\max_{\cX\in \Pi_{\bn}}|\langle \cX\otimes\overline{\cX},\cB\rangle|$.   Note that $\|\cB\|_{spec}\le \|\cB\|_{\infty}$, and equality holds if $\pm\cB\in\rH_{\bn,+}$.   A density tensor is called \emph{separable} if it is a convex combination of pure density tensors corresponding to product states $\cX\otimes \overline{\cX}, \cX\in \Pi_{\bn}$, a generalization of \cite{Per96}.  Denote by $\|\cdot\|_{nuc}$ the dual norm of $\|\cdot\|_{spec}$ on $\rH_{\bn}$.  A generalization of the results \cite{Rud00,Per04} states that a density tensor $\cR$ is separable if and only $\|\cR\|_{nuc}=1$.  

Denote by $\rB\rH_{n^{\times d}},\rA\rH_{n^{\times d}}\subset \rH_{n^{\times d}}$ the subspace of \emph{bi-symmetric} and \emph{bi-skew-symmetric} Hermitian tensors. respectively.  Each of this subspace is a linear span of pure density tensors $\cX\otimes\overline{\cX}$, where $\cX$ is either in $\rS^d\C^n$ or $\rA^d\C^n$ respectively.   The density tensors in $\rB\rH_{n^{\times d},+,1}$ and  $\rA\rH_{n^{\times d},+,1}$ correspond to the Bosons and Fermions density tensors respectively.   The spectral norm on $\rB\rH_{n^{\times d}}$ is $\|\cB\|_{bspec}:=\max_{\x\in\C^n,\|\x\|=1} |\langle\x^{\otimes d}\otimes \bar\x^{\otimes d},\cB\rangle|$, while on $\rA\rH_{n^{\times d}}$ one takes the restriction of the spectral norm on $\rH_{n^{\times d}}$.  Denote by $\|\cdot\|_{bnuc}$ and $\|\cdot\|_{anuc}$ the dual norms the $\|\cdot\|_{bspec}$ and $\|\cdot\|_{nuc}$ on $\rB\rH_{n^{\times d}}$ and $\rA\rH_{n^{\times d}}$ respectively.   A density tensor in $\rB\rH_{n^{\times d},+,1}$ is called strongly separable if it is a convex combination of pure density tensors of the form $\x^{\otimes d}\otimes \bar\x^{\otimes d}$, where $\x\in\C^n, \|\x\|=1$.
We show that a bi-symmetric density tensor $\cR$ is separable if and only if it is strongly separable,  i.e.,  $\|\cR\|_{bnuc}=1$.   A bi-skew-symmetric density tensor is called strongly separable if it is a convex combination of $(\x_1\wedge\cdots\wedge\x_d)\otimes (\bar\x_1\wedge\cdots\wedge\bar\x_d)$, where $\x_1,\ldots,\x_d$ is an orthonormal set in $\C^n$.  Then $\cR\in \rA\rH_{n^{\times d},+,1}$ is strongly separable if and only if $\|\cR\|_{anuc}=d!$.

We discuss the computability of tensor norms and separability of given tensors.
It is well known that most problems in tensors are NP-hard.  (See subsection \ref{subsec:NP} for a definition of NP-hardness.)  Similarly,  a bi-partite separability for $\cR\in \rH_{n^{\times 2},+,1}$ is NP-hard.  The main result of this paper is that for fixed $n\ge 2$ the computations of the spectral and nuclear norms on $\rS^d\C^n, \rS^d\C^n\otimes\rS^d\C^n$, and separability in $\rB\rH_{n^{\times d},+,1}$ are poly-time computable.
\subsection{Related works}\label{subsec:prior}
The notion of entanglement goes back to the famous papers of Einstein-Podolsky-Rosen \cite{EPR35} and Schr\"{o}dinger \cite{Sch35,Sch36}.   
Maximally entangled bi-partite qubits appear in Bell's inequalities as Bell states \cite{Bel64}.
Shimony introduced the notion of the geometric measure of entanglement in \cite{Shi95}.    We now mention those papers which are related to the results of our paper.  Barnum-Linden \cite{BL01} discuss multi-partite case, but not the GME.
They discuss Slater rank for Bosons (symmetric tensors), and Fermions (skew-symmetric tensors),  that we briefly touch upon in subsection \ref{subsec:ferm}.
Wei-Goldbart \cite{WG03} discuss GME for multi-partite state without introducing the spectral norm.  The also discuss briefly the GME of special Bosons.  For density tensor entanglement they use the \emph{convex roof} that we don't discuss.  Tamaryan-Wei-Park \cite{TWP09} give the proof that the W-state is the most entangled $3$-qubit state with respect to the GME,  see also \cite{CXZ10}.   H\"ubener et al \cite{Hubetall09} discuss the GME for Bosons.  They rediscover Banach's theorem \cite{Ban38} that the closest $d$-partite product state to a given Boson can be chosen to be a Boson $\x^{\otimes d}, \|\x\|=1$.   Aulbach-Markham-Murao \cite{AMM10} give examples of $d$-symmetric qubits for $d=4,\ldots,12$, which they assume to have the maximal GME.  Their examples are motivated by the Majorana model.  In \cite{FW20} we verified the numerical values of these examples using our software, and we could not find symmetric qubits with higher GME.  Friedland-Wang \cite{FW20}  describe a way to compute the GME of Bosons using solutions of polynomial equations.
Gross-Flammia-Eisert \cite{GFE09} show that most of $d$-qubits are maximally entangled for $d\gg 1$,  and are not useful as computational resources.    Friedland-Kemp \cite{FK18} showed that most of $d$-Bosons in a fixed $\C^n$ are maximally entangled for $d\gg 1$.  A simple alternating algorithm for finding the spectral norm and the GME of entanglement of a multi-partite state is given in \cite{SG22}.

Werner \cite{Wer89} showed that a bi-partite separable density tensor can be modeled by a hidden-variable theory.   Werner density tensors are separable for $p\ge \frac{1}{2}$  and entangled for $p< \frac{1}{2}$.   Peres \cite{Per96} introduced the notion of  the partial positivity transport (PPT) for separability.  
All entangled Werner states violate the PPT separability criterion.  M.-P.-R.  Horodecki \cite{Hor96} showed that for $2-2$-qubits and $2-3$ bi-partite qubits the density tensor is separable if and only if the PPT property holds.   P. Horodecki \cite{Hor97} showed that there exists non-separable bi-partite 3-qubits which have PPT property.
Rudolph \cite{Rud00} showed that a bi-partite tensor is separable if and only if its nuclear norm is one.   This result was extended to $d$-partite density tensors for $d\ge 3$ by Perez-Garcia \cite{Per04}.    Chen-Wu \cite{CW02} gave a computable cross norm or realignment criterion for bi-partite separability.
Gurvits-Barnum \cite{GB02} found the maximum ball of separable density tensors centered at  the maximally mixed bi-partite density tensor, which is a scalar times the identity operator.  Doharty-Parrilo-Spedalieri in series of three papers \cite{DPS02,DPS04,DPS05} introduced the methods of symmetric extensions and SDP hierarchies to determine the separability.
A remarkable result of Gurvits \cite{Gur03} claims that it is NP-hard to determine if a given density tensor is within $\varepsilon>0$ distance to the set of separable bi-partite tensors.   This result shows that even if one has an algorthm to determine separability it can not be polynomial, assuming that NP-hard problems are not polynomial-time computable.
See also S. Gharibian \cite{Gha10}. 

Chen-Chu-Qian-Shen \cite {CCQS} discuss density tensors which are symmetric, (they call such tensors supersymmetric).  This tensors are real valued.  They show that a $2d$-symmetric density tensor is separable if and only if it strongly separable: it is a convex combinations of $\x^{\otimes (2d)}, \x\in\R^n, \|\x\|=1$. Recent paper \cite{DIG25} discusses chiral symmetries and multiparticle entanglement. 

Derksen-Friedland-Lim-Wang \cite{DFLW17} give alternating algorithms for computing the nuclear norm of general and symmetric tensors.  (For a density tensor $\cR$ the nuclear norm measures the inseparability of $\cR$.) Weinbrenner-G\"uhne  \cite{WG25} discuss the problem of quantifying entanglement from the physics prospective. 

\subsection{A short summary of the paper}\label{subsec:sum}
Section \ref{sec:tensors} discusses some results in tensors that we use in this paper.
We consider our tensors over a field $\F$ either of real numbers $\R$ or the complex numbers $\C$.
Subsection \ref{subsec:basnot} discusses some basic notions in tensors: rank of tensor,
product states,  spectral and nuclear norms, and some of their properties.  
Subsection \ref{subsec:symtens} discusses $d$-symmetric tensors.  We recall the one-to-one correspondence between the symmetric polynomials and homogeneous polynomials of degree $d$.  Next we recall the nice and useful Banach characterization of the spectral norm which implies an analogous characterization of the nuclear norm.
Subsection \ref{subsec:ferm} discusses \emph{Fermions}, which are \emph{skew-symmetric} tensors of norm one.   Skew-symmetric tensors are elements of the Exterior (Grassmann) algebra.  It turns out that one has an analog to Banach's theorem for skew-symmetric tensors.   Subsection \ref{subsec:nucrank} discusses briefly the notion of the nuclear rank.  Nuclear rank, unlike the tensor rank, is lower semicontinuous.

Section \ref{sec:ent} discusses two major notions of entanglements: the geometric measure of entanglement (GME) and the nuclear norm of tensor states. Subsection \ref{subsec:2ment} deals with two measures of entanglement of tensor states expressed in terms of the spectral and nuclear norms.   The maximum entangled states are the same with respect to theses two measurements.
Subsection \ref{subsec:krmes} deals with the minimum value of the tensor norm of the state tensors, and gives some lower bounds.  Subsection \ref{subsec:quadrit} reproduces the unpublished result of \cite{DFLW17}, with some more details, that the Higuchi-Sudbery $4$-qubit \cite{HS00} is the most geometric entangled 4-qubit.
Subsection \ref{subsec:bigd} discusses entanglement of  $d$-qubits and $d$-Boson qubits for big $d$.  It is known that most of such qubits are almost maximally entangled.

Section \ref{sec:sep} discusses the notion of density tensors and their separability.
In subsection \ref{subsec:dentensep} we first define the notion of Hermitian tensors 
$\H_{\bn}\subset \otimes^2\C^{\bn},$ which can be viewed as a linear transformation on $\C^{\bn}$.  A Hermitian tensor has a spectral decomposition 
$\cB=\sum_{k=1}^{N(\bn)}\lambda_k(\cB)\cX_k\otimes\overline{\cX_k}, \cX_k\in\C^{\bn},\langle \cX_k,\cX_l\rangle=\delta_{kl}, k,l\in[N(\bn)]$.
On $\rH_{\bn}$ we define 
a  spectral norm $\|\cB\|_{spec}=\max_{\cX\in\C^{\bn},\|\cX\|=1}$ $|\langle \cX, \cB\cX\rangle|$,  which is majorized  by $\|\cB\|_\infty$.   The dual norm to $\|\cdot\|_{spec}$ on $\rH_{\bn}$  is $\|\cdot\|_{nuc}$, which majorizes the norm $\|\cdot\|_{1}$.
The quantity $\log\|\cR\|_{nuc}$ of a density tensor $\cR$ can be viewed as the \emph{energy} of a \emph{minimal} local pseudomixture \cite{STV98}.
Theorem \ref{charsep} states that the set of separable states is $\textrm{Sep}_{\bn}$ is characterised by $\|\cR\|_{nuc}=1$.  Furthermore,  $\textrm{Sep}_{\bn}$ is a maximal dimensional face of the unit ball centered at the origin with respect to $\|\cdot\|_{nuc}$.  
We also characterize the maximal entangled density tensors that maximize the nuclear norm.   In subsection \ref{subsec:bisymt} we discuss the bi-symmetric Hermitian tensors, denoted as $\rB\rH_{n^{\times d}}$.  We show that they have the following spectral decomposition: 
$\cB=\sum_{k=1}^{{n+d-1\choose d}}\lambda_k(\cB)\cX_k\otimes\overline{\cX_k}, \cX_k\in\rS^d\C^n,\langle \cX_k,\cX_l\rangle=\delta_{kl}, k,l\in[{n+d-1\choose d}]$, see also \cite{M-S25}.
On $\rB\rH_{n^{\times d}}$ we define another spectral norm $\|\cB\|_{bspec}=\max_{\x\in\C^n,\|\x\|=1} |\langle \x^{\otimes d}\otimes \bar\x^{\otimes d},\cB\rangle|$,
which is majorized $\|\cB\|_{spec}$.   The dual norm of $\|\cdot\|_{bspec}$ is denoted by $\|\cdot\|_{bnuc}$, which majorizes $\|\cdot\|_{nuc}$.   A bi-symmetric Hermitian density tensor is called strongly separable if it is a convex combination of $\x^{\otimes d}\otimes \bar\x^{\otimes d}$ for $\x\in\C^n, \|\x\|=1$.  Lemma \ref{bisymsep} shows that a bi-symmetric density tensor $\cR$ is separable if and only if it is strongly separable.  Denote by $\textrm{Seps}_{n^{\times d}}$ the set of strongly separable density tensors.
Theorem \ref{charseps} is an analog Theorem \ref{charsep}.  In particular, it shows that a bi-symmetric Hermitian density tensor $\cR$ is strongly separable if and only if $\|\cR\|_{bnuc}=1$.   Subsection \ref{subsec;symher} discusses a subspace of bi-symmetric Hermitian tensors, which are symmetric tensors: $\rS\rH_{n^{\times d}}:=\rS^{2d}\R^n\cap\rB\rH_{n^{\times d}}$.  Symmetric tensors are called supersymmetric in \cite{CCQS}.   Theorem \ref{charsepsym} is an analog of Theorem \ref{charseps}.  In particular, we show that a symmetric density tensor $\cR$ is separable as real density tensor if and only if $\|\cR\|_{1,\R}=1$.  This result gives another proof of Theorem  \cite[Theorem 14]{CCQS}:  Every real separable symmetric density tensor is a convex combination of symmetric rank one density tensors $\x^{\otimes (2d)}, \x\in\R^n,\|\x\|=1$.   Subsection \ref{subsec:bisksymt} deals with bi-skew-symmetric Hermitian tensors, denoted by $\rA\rH_{n^{\times d}}$.  These tensors have a spectral decomposition 
$\cB=\sum_{k=1}^{{n\choose d}}\lambda_k(\cB)\cX_k\otimes\overline{\cX_k}, \cX_k\in\rA^d\C^n,\langle \cX_k,\cX_l\rangle=\delta_{kl}, k,l\in[{n\choose d}]$.
The spectral norm on $\rA\rH_{n^{\times d}}$ is the restriction of the norm $\|\cdot\|_{\infty}$.   A basic separable $\rA\rH_{n^{\times d}}$ density tensor of Slater rank one is $(\x_1\wedge\cdots\wedge\x_d)\otimes(\bar\x_1\wedge\cdots\wedge\bar\x_d)$, where $\x_1,\ldots,\x_d$ is an orthonormal set of vectors in $\C^n$.   
A bi-skew-density tensor is called separable if it is a convex combinations of basic separable tensors in $\rA\rH_{n^{\times d}}$.
Denote by $\|\cdot\|_{anuc}$ the dual norm to $\|\cdot\|_\infty$ on $\rA\rH_{n^{\times d}}$. we show
$\|\cB\|_{nuc}\ge\|\cB\|_{anuc}\ge \|\cB\|_{1,\C}$ for bi-skew-Hermitian tensors.
Theorem\ref{charsepa} is analog of Theorem \ref{charseps}.   In particular we show that 
a bi-skew-symmetric density tensor is separable if and only if $\|\cR\|_{anuc}=d!$.
The main result of subsection \ref{subsec:semal} is the claim that the 
separability sets $\mathrm{Sep}_{\bn}$,  $\mathrm{Seps}_{n^{\times d}}$ and their restrictions to $\rS^{2d}\R^n_{+,1}$ are semi-algebraic sets.  

Section \ref{sec:comp} discusses the complexity results on tensors, entanglement and separablity.  Subsection \ref{subsec:NP} explains briefly the notions of 
poly-time,  NP and NP-hard problems.  The main results of these section is that the following norms are NP-hard to compute: $\|\cdot\|_{bspec}$ and $\|\cdot\|_{bnuc}\|$ for tensors in $\rB\rH_{n^{\times 2},+,1}$ and  $\rB\rH_{n^{\times 2}}$ respectively;
the norms $\|\cdot\|_{\infty,\C}$ and $\|\cdot\|_{1,\C}$ for  tensors in $\rH_{n^{\times 2m},+,1}$ and  $\rH_{n^{\times 2m}}$for $m\in\N$ respectively;
the norms $\|\cdot\|_{\infty,\R}$ and $\|\cdot\|_{1,\R}$ for  tensors in $\Re\rH_{n^{\times 2m},+,1}$ and  $\Re\rH_{n^{\times 2m}}$for $m\in\N$ respectively;
the norms $\|\cdot\|_{spec}$ and $\|\cdot\|_{nuc}$ for  tensors in $\rH_{n^{\times 2m},+,1}$ and  $\rH_{n^{\times 2m}}$for $m\in\N$ respectively;
the norms $\|\cdot\|_{\infty,\F}$ and  $\|\cdot\|_{1,\F}$ are NP-hard to compute on $\F^{n^{\times (4m)}}$ for $m\in\N$.  Subsection \ref{subsec:dpolt} shows how to code symmetric and bi-symmetric tensors as polynomials.   Subsection  \ref{subsec:aprsn} shows how to approximate spectral and nuclear norms of symmetric and bi-symmetric  tensors using the results of subsection \ref{subsec:dpolt}.  Subsection \ref{subsec:prfpt} proves the most important result of this paper.  The following norms are poly-time computable in $d$ for a fixed $n\ge 2$:
the norms $\|\cdot\|_{\infty,\F}$ and $\|\cdot\|_{1,\F}$ on $\rS^d\F^n$,
the norms  $\|\cdot\|_{\infty,\F}$ and $\|\cdot\|_{1,\F}$ for $\rS^d\F^n\otimes\rS^d\F^n$.
 In particular, the separability of bi-symmetric Hermitian density tensors, and separability of symmetric density tensors are poly-time computable.
 
 Section \ref{sec:oprb} states a small number of open problems.
 Appendix \ref{sec:eucgen} compare a general norm $\nu$ on $\R^n$ of $\C^n$ to the Euclidean norm.  The important result is that the maximal and the minimum distortions 
 of $\nu$ nad its dual $\nu^*$ are reciprocal.   In subsection \ref{subsec:nurank} we dicuss the notion of $\nu$$\rank$.
\section{Tensors}\label{sec:tensors}
\subsection{Basic notions}\label{subsec:basnot}
Let 
\begin{equation}\label{basnot}
\bn=(n_1,\ldots,n_d)\in\N^d,  N(\bn)=\prod_{i=1}^d n_, d\in\N, [\bn]=[n_1]\times\cdots\times [n_d],\1_d=(1,\ldots,1)\in\N^d. 
\end{equation}
Assume that $\F$ is a field.
The space of $d$-mode tensors is the tensor product $\otimes_{j=1}^d \F^{n_j}$ denoted as $\F^{\bn}=\F^{n_1\times\cdots\times n_d}$, where $\bn=(n_1,\ldots,n_d)\in\N^d, d\in\N$. The cases: $d=1$ corresponds to vectors $\x=(x_1,\ldots,x_n)^\top\in \F^n$; $d=2$ corresponds to matrices $X=[x_{i,j}]\in \F^{m\times n}$; $d\ge 3$ corresponds to $\cX=[x_{i_1,\ldots,i_d}]\in \F^{\bn}$.  For $d\ge 2$ a tensor $\cX\in \F^{\bn}$ is a rank-one tensor if $\cX=\otimes_{j=1}^d \x_j=[x_{i_1,1}\cdots x_{i_d,d}]\ne 0$.   Denote by $\Sigma_{\bn,\F}\subset \F^{\bn}$ the Segre variety of rank-one tensors \cite{Lan12}.  The rank of $\cX\ne 0$ is the minimal $r$ such that
$\cX$ is is a sum of $r$ rank-one tensors.  It is shown in \cite{Has90} that determining the rank of $\cX$ for $d\ge 3$ is an NP-complete problem for a finite field $\F$, and an NP-hard problem for $\F=\R,\C$.

In what follows we assume that $\F=\R,\C$.   Whenever the field $\F$ is omitted, we assume that $\F=\C$.   A standard inner product on $\F^{\bn}$ is 
$\langle \cX,\cY\rangle=\sum_{(i_1,\ldots,i_s)\in[\bn]}\overline{x_{i_1,\ldots,i_d}} y_{i_1,\ldots,i_d}$, and the Euclidean norm $\|\cX\|=\sqrt{\langle \cX,\cX\rangle}$.
Clearly, for $\cX=\otimes_{j=1}^d\x_j$ we have the equality  $\|\cX\|=\prod_{i=1}^n\|\x_i\|$.  Denote by $\Pi_{\bn,\F}\subset \Sigma_{\bn,\F}$ the set of rank one tensors whose norm is one:
\begin{equation}\label{defPibn}
\Pi_{\bn,\F}=\{\otimes_{j=1}^d, \x_j\in\F^{n_j}, \|\x_j\|=1, j\in[d]\},
\end{equation}
which is called the set of product states.

The following norms for $\cX\in\cH$ are called the spectral and nuclear norms, respectively \cite{FL18}:
\begin{equation}\label{defspecnuc}
\begin{aligned}
&\|\cX\|_{\infty,\F}=\max_{\cY\in \Pi_{\bn,\F}} \Re\langle\cY,\cX\rangle=\max_{\cY\in \Pi_{\bn,\F}} |\langle\cY,\cX\rangle|,\\
&\|\cX\|_{1,\F}=\inf_{\cX=\sum_{i=1}^{r} \cX_i, \cX_i\in  \Sigma_{\bn,\F},r\in\N} \sum_{i=1}^{r} \|\cX_i\|\\
\end{aligned}
\end{equation} 
 Schatten \cite{Sch50} introduced the notions of various types of cross-norms including the  least (injective) and greatest (projective) cross-norms, which correspond to the spectral and nuclear norms respectively.
These norms were generalized by Grothendieck \cite{Gro55}, under the name \emph{injective and projective tensor norms} to two Banach spaces.   Since we are dealing with the finite dimensional case one can put an explicit  upper bound on $r$.   This upper bound and other results are stated in the following theorem:
\begin{theorem}\label{1inftynormlem}
Let $d\ge 2$ and $2\1_d\le \bn\in\N^d$.   Assume that $0\ne\cX\in\F^{\bn}$.  Then
\begin{enumerate}
\item The following characterization holds:
\begin{equation}\label{nucchar}
\begin{aligned}
&\|\cX\|_{1,\F}=\min_{\cX=\sum_{i=1}^{r} \cX_i, \cX_i\in  \Sigma_{\bn,\F}, r\le r(\bn,\F)} \sum_{i=1}^{r} \|\cX_i\|  \\
&r(\bn,\R)=N(\bn)+1, \quad r(\bn,\C)=2N(\bn)+1.
\end{aligned}
\end{equation}
\item The following sharp inequalities hold
\begin{equation}\label{infty21in}
\|\cX\|_{\infty,\F}\le \|\cX\|\le \|\cX\|_1.
\end{equation}
Equality in one the inequalities hold if and only if $\cX$ is a rank-one matrix.
\item Let $\alpha_{s,\bn,\F}\le \beta_{s,\bn,\F}$ be the constants defined in \eqref{albedef} for the norms $\nu_s=\|\cdot\|_{s,\F}$ for $s=\infty,1$:
\begin{equation}\label{defalbeinf1}
\begin{aligned}
&\alpha_{\infty,\bn,\F}=\min_{\cX\in\F^{\bn},\|\cX\|=1} \|\cX\|_{\infty,\F}, \quad \beta_{\infty,\bn,\F}=\max_{\cX\in\F^{\bn},\|\cX\|=1} \|\cX\|_{\infty,\F},\\
&\alpha_{1,\bn,\F}=\min_{\cX\in\F^{\bn},\|\cX\|=1} \|\cX\|_{1,\F}, \quad \beta_{1,\bn,\F}=\max_{\cX\in\F^{\bn},\|\cX\|=1} \|\cX\|_{1,\F}.
\end{aligned}
\end{equation}
Then 
\begin{equation}\label{alphbetsineq}
\begin{aligned}
&\alpha_{\infty,\bn,\F}<1=\beta_{\infty,\bn,\F},\\
&\alpha_{1,\bn,\F}=1<\beta_{1,\bn,\F}=\frac{1}{\alpha_{\infty,\bn,\F}}.
\end{aligned}
\end{equation}
Furthermore for $\|\cX\|=1$ one has the implications 
\begin{equation}\label{alphbetimpl}
\alpha_{\infty,\bn,\F}=\|\cX\|_{\infty,\bn,\F}\iff \beta_{1,\bn,\F}=\|\cX\|_{1,\bn,\F}.
\end{equation}
\item Let $e\in \N,  \bm\in\N^{e}$.
Assume that $0\ne \cY\in \F^{\bm}$.  Then
\begin{equation}\label{tenprodeq}
\|\cX\otimes\cY\|_{\infty,\F}=\|\cX\|_{\infty,\F}\|\cY\|_{\infty,\F}, \quad \|\cX\otimes\cY\|_{1,\F}=\|\cX\|_{1,\F}\|\cY\|_{1,\F}.
\end{equation}
\end{enumerate}
\end{theorem}
\begin{proof}\emph{(1)}.  It is enough to prove the characterization \eqref{nucchar}, where $\|\cX\|_{1,\F}=1$.  
We now repeat briefly the arguments in \cite{FL18} that the extreme points of $\rB_{\|\cdot\|_1}(\0,1)$ is $\Pi_{\bn,\F}$.   In view of the characterization of $\|\cdot\|_{\infty,\F}$ given in \eqref{defspecnuc} it follows that convex span of $\Pi_{\bn,\F}$ is a unit ball of a norm $\nu$.  The characterization \eqref{defnustar2} yields that $\nu^*=\|\cdot\|_{\infty,\F}$.    Assume that $\cX=\otimes_{i=1}^d \x_i\in \Pi_{\bn,\F}$.  Clearly,
 $\max_{\cY, \|\cY\|=1}\Re\langle\cY,\cX\rangle\le 1$,
 where equality holds if and only if $\cY=\cX$.  Hence, $\Pi_{\bn,\F}$ is a closed set of the extreme points of $\rB_{\|\cdot\|_{1,\F}}(\0,1)$.   As  $\|\cX\|_{1,\F}=1$ it follows that 
 $\cX\in \rB_{\|\cdot\|_{1,\F}}(\0,1)$.   
 Assume first that $\F=\R$.   As the real dimension of $\R^{\bn}$ is $N(\bn)$, Caratheodory's theorem yields that $\cX$ is a convex combination of at most $N(\bn)+1$ vectors of $\Pi_{\bn,\F}$.  In particular,  $\cX$ is a convex combination of exactly $r\le N(\bn)+1$ vectors of $\Pi_{\bn,\F}$:
 \begin{equation*}
 \cX=\sum_{i=1}^r \cX_i=\sum_{i=1}^r a_i\cY_i, \quad \cX_i=a_i\cY_i,a_i>0,\cY_i\in\Pi_{\bn,\F}, i\in[r], \sum_{i=1}^ra_i=\|\cX\|_{1,\F}=1.
 \end{equation*}
 Assume now that  $\cX=\sum_{i=1}^r \cX_i$, and $\cX_i\in \Sigma_{\bn,\R}, i\in[r]$.  Clearly, $\cX_i=\|\cX_i\| \cY_i$, where $\cY_i=\frac{1}{\|\cX_i\|}\cX_i\in \Pi_{\bn,\R}, i\in[r]$.  Use the triangle inequality to deduce
 \begin{equation*}
 1=\|\cX\|_{1,\R}=\|\sum_{i=1}^r \|\cX_i\|\cY_i\|_{1,\R}\le \sum_{i=1}^r\|\cX_i\|.
 \end{equation*}
 This proves the characterizations in \eqref{defspecnuc} and \eqref{nucchar} for $\|\cX\|_{1,\R}$.
 
 To show the  the characterizations in \eqref{defspecnuc} and \eqref{nucchar} for $\|\cX\|_{1,\C}$,  observe that $\C^{\bn}\sim\R^{\bn}\oplus \R^{\bn}$.  That is, the real dimension of  $\C^{\bn}$ is $2N(\bn)$.  Then use Caratheodory's theorem.
 
 \noindent \emph{(2)}  As 
 $$\|\cX\|=\max_{\cY, \|\cY\|=1} \Re\langle \cY,\cX\rangle\ge \max_{\cY, \cY\in\Pi_{\bn,\F}}\Re\langle \cY,\cX\rangle=\|\cX\|_{\infty,\F},$$
 we deduce the first inequality in \eqref{infty21in}.   
 Assume that $\|\cX\|=1$.
 Clearly, if $\cX\in\Pi_{\bn,\F}$ then $\|\cX\|_{\infty,\F}=1$.  Suppose that $\cX\not\in\Pi_{\bn,\F}$.  Observe that for $\|\cY\|=1$ one has the inequality $\Re\langle \cY,\cX\rangle\le 1$, and equality holds if and only if and only if $\cY=\cX$.  Hence, $\|\cX\|_{\infty,\F}<1$.
 
 Let $0<\|\cX\|_{1,\F}=\sum_{i=1}^r a_i$, where $a_i>0$ and $\cX=\sum_{i=1}^r a_i\cY_i,  \cY_i\in \Pi_{\bn,\F}$.  Clearly, 
 $$\|\cX\|\le \sum_{i=1}^r ||a_i\cY_i\|=\sum_{i=1}^r a_{i}=\|\cX\|_{1,\F},$$
 which establishes the second inequality in \eqref{infty21in}.   Suppose that $\|\cX\|=\|\cX\|_{1,\F}$.   We claim that $\cY_1=\cdots=\cY_r$.  For $r=1$ there is nothing to prove.  Assume that $r>1$.  Then we have an equality in the triangle inequality for the Euclidean norm.  Thus implies that $a_iY_i=t_i a_1Y_1$ for $t_i>0$ and $i\ge 2$.
 Hence, $\cY_i=\cY_1$ for $i>1$, and $\cX\in \Sigma_{\bn,\F}$.
 
 \noindent \emph{(3)} The results of \emph{(2)} yield that 
 $$\alpha_{\infty,\bn,\F}<\beta_{\infty,\bn,\F}=\alpha_{1,\bn,\F}=1<\beta_{1,\bn,\F}.$$
 As the norms $\|\cdot\|_{\infty,\bn,\F}$ and $\|\cdot\|_{1,\bn,\F}$ are conjugate norms, 
 \eqref{albetm} yields that $\alpha_{\infty,\bn,\F}\beta_{1,\bn,\F}=1$.  This proves 
 \eqref{alphbetsineq}. The implication \eqref{alphbetimpl} follows from \eqref{maxmineq}.
 
 \noindent  \emph{(4)}  The first equality in \eqref{alphbetsineq} follows straightforward from the observation that $\Pi_{(\bn,\bm),\F}=\Pi_{\bn,\F}\otimes\Pi_{\bm,\F}$ and the definition of $\|\cX\otimes\cY\|_{\infty,\F}$.
To prove the second equality we argue as follows.  Let $\|\cX\|_{1,\F}$ and $\|\cY\|_{1,\F}$ correspond to the following decompositions of $\cX$ and $\cY$ respectively:
\begin{equation*}
\begin{aligned}
&\cX=\sum_{i=1}^p \cX_i, \quad \cX_i\in \Sigma_{\bn,\F}, \quad \|\cX\|_{1,\F}=\sum_{i=1}^p\|\cX_i\|,\\
&\cY=\sum_{j=1}^q \cY_j \quad \cY_j\in \Sigma_{\bm,\F}, \quad \|\cY\|_{1,\F}=\sum_{j=1}^q\|\cY_j\|.
\end{aligned}
\end{equation*}
Hence,
\begin{equation*}
\cX\otimes\cY=\sum_{i=j=1}^{p,q}\cX_i\otimes \cY_j\Rightarrow \|\cX\otimes\cY\|_{1,\F}\le \sum_{i=j=1}^{p,q}\|\cX_i\otimes \cY_j\|\le  \|\cX\|_{1,\F} \|\cY\|_{1,\F}.
\end{equation*}
 We now show the reverse inequality.  Clearly,  $ \|\cX\otimes\cY\|_{1,\F}=\max_{\cZ\in\F^{(\bn,\bm)}, \|\cZ\|_{\infty,\F}=1}|\langle \cZ, \cX\otimes\cY\rangle|$.
 Now choose 
 $$\cZ=\cZ_1\otimes \cZ_2, \cZ_1\in\F^{\bn}, \cZ_2\in\F^{\bm},\|\cZ_1\|_{\infty,\F}=\|\cZ_2\|_{\infty,\F}=1\Rightarrow \|\cZ\|_{\infty,\F}=1.$$ 
 Observe that $\langle \cZ_1\otimes \cZ_2, \cX\otimes\cY\rangle=\langle \cZ_1,\cX\rangle \langle \cZ_2,\cY\rangle$.  Take maximum of $|\langle \cZ, \cX\otimes\cY\rangle|$ on all such $\cZ$ to deduce: $ \|\cX\otimes\cY\|_{1,\F}\ge \|\cX\|_{1,\F}\|\cY\|_{1,\F}$.  Combine this inequality with a previous one to deduce $ \|\cX\otimes\cY\|_{1,\F}= \|\cX\|_{1,\F}\|\cY\|_{1,\F}$.   
\end{proof}

The following result is well known for matrices.
\begin{lemma}\label{alphbetdim2} Let  $1<m\le n$ be integers.  Let $T\in \F^{m\times n}$.  Let $\sigma_1\ge \cdots\ge\sigma_m\ge 0$ the the singular values of $T$.   
Then
\begin{equation}\label{inty1normT}
\|T\|_{\infty,\F}=\sigma_1, \quad \|T\|_{1,\F}=\sum_{i=1}^m \sigma_i.
\end{equation}
Assume $\|T\|=1$.  Then 
\begin{enumerate}
\item The equality $\|T\|_1=\beta_{1,(m,n),\F}$ holds if and only if the $m$ singular values of $T$ are $\frac{1}{\sqrt{m}}$.
\item The equality $\|T\|_\infty=\alpha_{\infty, (m,n),\F}$ holds if and only if the $m$ singular values of $T$ are $\frac{1}{\sqrt{m}}$.
\end{enumerate}
In particular 
\begin{equation}\label{alphbetd=2}
\beta_{1,(m,n),\F}=\sqrt{m}, \quad \alpha_{\infty,(m,n),\F}=\frac{1}{\sqrt{m}}.
\end{equation}
\end{lemma} 
\begin{proof}  The first equality of \eqref{inty1normT} is well known.  The second equality is shown in \cite[Corollary 7.5.13]{Frib}.
 Assume that $\|T\|=1$
Then $1=\|T\|^2=\sum_{i=1}^m \sigma_i^2$.  Use Cauchy-Schwarz inequality to deduce that
$\|T\|_1^2=\left(\sum_{i=1}^m \sigma_i\right)^2\le m \left(\sum_{i=1}^m \sigma_i^2\right)=m$.  Equality holds if and only if all singular values of $T$ are $\frac{1}{\sqrt{m}}$.

Observe next that from the equality $1=\sum_{i=1}^m \sigma_i^2$ we deduce that $1\le m\sigma_1^2$.   Hence $\|T\|_{\infty}=\sigma_1\ge \frac{1}{\sqrt{m}}$.  Equality holds if and  only if all singular values of $T$ are $\frac{1}{\sqrt{m}}$.
\end{proof}
\begin{definition}\label{defunitmat}
Denote by $\rU_n\supset \rS\rU_n$ the group of unitary matrices and the subgroup of the special unitary matrices of order $n$.   Let $\mu_n$ be the Haar measure on $\rU_n$.
Denote
\begin{equation*}
\rU_{\bn}:=\otimes_{j=1}^d \rU_{n_j}, \quad \rS\rU_{\bn}=\otimes_{j=1}^d \rS\rU_{n_j},
\end{equation*}
are subgroups of $\rU_{N(\bn)}$ and $\rS\rU_{N(\bn)}$ respectively that act on $\C^{\bn}$: $\cX\mapsto U\cX$.
Denote by $\phi_d: \rU_n\to \rU_{n^{\times d}}$ the diagonal embedding $A\mapsto \otimes^d A$, where
\begin{equation}\label{defnd}
n^{\times d}=(\underbrace{n,\ldots,n}_d).
\end{equation}
\end{definition}
Observe that $\rU_{\bn}$ is a subgroup of isometries of the spectral and nuclear norms on $\C^{\bn}$.
\subsection{Symmetric tensors}\label{subsec:symtens}
A tensor $\cS=[s_{i_1,\ldots,i_d}]\in\otimes^d\F^n$ is called symmetric if $s_{i_1,\ldots,i_d}=
s_{i_{\omega(1)},\ldots,i_{\omega(d)}}$ for every permutation $\omega:[d]\to[d]$.  
Denote by $\rS^d\F^n\subset \otimes^d\F^n$ the vector space of $d$-mode symmetric tensors on $\F^n$.
It is well known that $\dim \rS^d\F^n={n+d-1\choose d}$ \cite{FK18}.
In what follows we assume that $\cS$ is a symmetric tensor and $d\ge 2$, unless stated otherwise.  A tensor
$\cS\in\rS^d\F^n$ defines a unique homogeneous polynomial of degree $d$ in $n$ variables
\begin{equation}\label{defpolfx}
f(\x)=\langle\overline{\cS},\otimes^d\x\rangle=\sum_{0\le j_k\le d,k\in[n], j_1+\cdots +j_n=d} \frac{d!}{j_1!\cdots j_n!} f_{j_1,\ldots,j_n} x_1^{j_1}\cdots x_n^{j_n}.
\end{equation}
Conversely, a homogeneous polynomial $f(\x)$ of degree $d$ in $n$ variables defines a unique symmetric $\cS\in\rS^d\F^n$ by the following relation.
Consider the multiset $\{i_1,\ldots,i_d\}$, where each $i_l\in [n]$.  Let $j_k$ be the number of times the integer $k\in [n]$ appears in the multiset  $\{i_1,\ldots,i_d\}$.
Then $s_{i_1,\ldots,i_d}=f_{j_1,\ldots,j_n}$.  Furthermore
\begin{equation}\label{symtenhsnorm}
\|\cS\|^2=\sum_{0\le j_k\le d,k\in[n],j_1+\cdots + j_n=d}\frac{d!}{j_1!\cdots j_n!} |f_{j_1,\ldots,j_n}|^2,
\end{equation}
where $s_{i_1,\ldots,i_d}=f_{j_1,\ldots,j_n}$.  

A remarkable result of Banach \cite{Ban38} claims that the spectral norm of a symmetric
tensor can be computed as a maximum on the set of rank one symmetric tensors:
\begin{equation}\label{Banthm}
\|\cS\|_{\sigma,\F}=\max_{\x\in\F^n, \|\x\|=1}\pm\Re\langle \otimes^d\x,\cS\rangle. 
\end{equation}
(The $\pm$ sign needed only if $\F=\R$ and $d$ is even.)
This result was rediscovered several times since 1938.  In quantum information theory (QIT), for the case $\F=\C$, it appeared in \cite{Hubetall09}.  In mathematical literature, for the case $\F=\R$, it appeared in \cite{CHLZ12,Fri13}.  (Observe that  a natural generalization of Banach's theorem to partially symmetric tensors is given in \cite{Fri13}.)

The analog of Banach's theorem for the nuclear norm of symmetric tensors was stated in \cite{FL18}.   Namely, for $0\ne \cS\in \rS^d\F^n$ we have the following minimal characterization
\begin{equation}\label{FLBthm}
\begin{aligned}
&\|\cS\|_{1,\F}=\min_{\cS=\sum_{i=1}^r\varepsilon_i\otimes^d \x_i, \;\x_i\in\F^n\setminus\{\0\},\varepsilon_i\in\{1,-1\}, r\le r(n,d,\F)} \sum_{i=1}^r \|\x_i\|^d,\\
&r(n,d,\R)={n+d-1\choose d}+1, \quad r(n,d,\C)=2{n+d-1\choose d}+1.
\end{aligned}
\end{equation}
We can assume that $\varepsilon_i=1$ unless $\F=\R$ and $d$ is even. 
Furthermore,  Caratheodory's theorem yields that $r\le {n+d-1\choose d}+1$ for $\F=\R$   and $r\le 2{n+d-1\choose d}+1$ for $\F=\C$.

Let
\begin{equation}\label{defalbeinf1s}
\begin{aligned}
&\alpha_{\infty,n^{\times d},s,\F}=\min_{\cX\in\rS^d\F,\|\cX\|=1} \|\cX\|_{\infty,\F}, \quad \beta_{\infty,n^{\times d},s,\F}=\max_{\cX\in\rS^d\F^n,\|\cX\|=1} \|\cX\|_{\infty,\F},\\
&\alpha_{1,n^{\times d},s,\F}=\min_{\cX\in\rS^d\F^n,\|\cX\|=1} \|\cX\|_{1,\F}, \quad \beta_{1,n^{\times d},s,\F}=\max_{\cX\in\rS^d\F^n,\|\cX\|=1} \|\cX\|_{1,\F}.
\end{aligned}
\end{equation}
We state an analog of Theorem \ref{1inftynormlem} and  \ref{alphbetdim2}.  The proof of this theorem is similar to Theorem \ref{1inftynormlem} and Lemma \ref{alphbetdim2}
and we leave it to the reader.
\begin{theorem}\label{relalphbetsym}  Let $n,d\ge 2$ be integers.  Then 
\begin{enumerate}
\item 
\begin{equation}\label{alphbetsineqs}
\begin{aligned}
&\alpha_{\infty,n^{\times d},s,\F}<1=\beta_{\infty,n^{\times d},s,\F},\\
&\alpha_{1,n^{\times d},s,\F}=1<\beta_{1,n^{\times d},s,\F}=\frac{1}{\alpha_{\infty,n^{\times d},s,\F}}.
\end{aligned}
\end{equation}
\item Assume that $\cS\in \rS^d\F^n$ and $\|\cS\|=1$.  Then $\|\cS\|_{\infty,\F}=\alpha_{\infty,n^{\times d},s,\F}$ if and only if $\|\cS\|_{1,\F}=\beta_{1,n^{\times d},s,\F}$.
\item $\alpha_{\infty,n^{\times 2},s,\F}=\frac{1}{\sqrt{n}},\, \beta_{1,n^{\times 2},s,\F}=\sqrt{n}$.  
\item Assume that $S$ is an $n\times n$ complex valued symmetric matrix having Frobenius norm one: $\|S\|=1$.  Then $\|S\|_{1,\F}=\beta_{1,n^{\times 2},s,\F}$ if and only if $\sqrt{n}S$ is a unitary matrix.
\end{enumerate}
\end{theorem}

\subsection{Fermions}\label{subsec:ferm}
Denote by $\Omega_d$ the group of permutations acting on $[d]$.  Then sign$:\Omega_d\to \{-1,1\}$ is the character  such that sign$(\tau)=-1$ for a transposition $\tau$.
The group $\Omega$ acts on $\otimes^d\F^n=\F^{n^{\times d}}$ by permutations on the $d$ factors of $\otimes^d\F^n$: 
\begin{equation}\label{defTsigma}
T_{d,n}(\sigma)(\otimes_{j=1}^d \x_j)=\otimes_{j=1}^d x_{\sigma(j)}, \quad \x_i\in\F^n,  i\in[d],  2\le d\in\N, \sigma\in\Omega_d.
\end{equation}

Denote by $\mathrm{A}^d\mathbb{F}^n\subset \mathbb{F}^{n^{\times d}}$ the subspace of fermions,
i.e., antisymmetric (skew-symmetric) tensors.  That is, $\mathcal{F}=[f_{i_1,\ldots,i_d}]\in \mathrm{A}^d\mathbb{F}^n$ if $f_{\tau(i_1),\ldots,\tau(i_d)}=-f_{i_1,\ldots,i_d}$ for each transposition $\tau:[d]\to [d]$.   
Note that $\mathrm{A}^d\mathbb{F}^n=\{0\}$ if $d>n$.  Thus we assume that $d\le n$. 
\begin{lemma}\label{norAdinf1lem}
\begin{enumerate}
\item For a rank-one tensor $\otimes_{i=1}^d\mathbf{x}_i$ denote
\begin{equation}\label{deferm}\mathbf{x}_1\wedge\cdots\wedge\mathbf{x}_d:=\frac{1}{\sqrt{d!}} \sum_{\sigma\in\Omega_d}\mathrm{sign}(\sigma)\otimes_{j=1}^d\x_{\sigma(j)}.
\end{equation}
Then $\mathbf{x}_1\wedge\cdots\wedge\mathbf{x}_d=0$ if and only if $\mathbf{x}_1,\ldots,\mathbf{x}_d$ are linearly dependent.   Assume that $\mathbf{x}_1,\ldots,\mathbf{x}_d$ are linearly independent.   Let $\W=\mathrm{span}(\x_1,\ldots,\x_d)\subset \F^n$.    Then, for each $\y_1,\ldots,\y_d\in\W$ the tensor $\y_1\wedge\cdots\wedge\y_d$ is colinear with   $\x_1\wedge\cdots\wedge\x_d$.  
Furthermore,  $\|\otimes_{j=1}^d\x_j\|\ge \|\mathbf{x}_1\wedge\cdots\wedge\mathbf{x}_d\|$.  Equality holds if and only if $\x_1,\ldots,\x_d$ is an orthogonal set of vectors.
Assume that $\|\x_1\wedge\cdots\wedge\x_d\|=1$. Then $\x_1\wedge\cdots\wedge\x_d=\zeta (\bu_1\wedge\cdots\wedge\bu_d)$, where $\bu_1,\ldots,\bu_d$ is any orthonormal basis in $\W$, and $\zeta\in\F, |\zeta|=1$ is uniquely determined by $\x_1\wedge\cdots\wedge\x_d$.  
\item
Let $\rT_{d,n}(\sigma):\otimes^d\F^n\to \otimes^d\F^n$ be defined by \eqref{defTsigma}.   Then for $\cX,\cY\in\otimes^d\F^n$ and $\sigma,\eta\in\Omega_{d}$ the following equalities hold:
\begin{equation}\label{sigeq}
\begin{aligned}
&\rT_{d,n}(\sigma)\rT_{d,n}(\eta)=\rT_{d,n}(\sigma\eta),\\
&\langle \rT_{d,n}(\sigma)\cX,\cY\rangle=\langle \cX,\rT_{d,n}(\sigma^{-1})\cY\}\rangle,\\
&\| \rT_{d,n}(\sigma)\cX\|=\|\cX\|, \quad \| \rT_{d,n}(\sigma)\cX\|_{\infty,\F}=\|\cX\|_{\infty,\F}, \quad \| \rT_{d,n}(\sigma)\cX\|_{1,\F}=\|\cX\|_{1,\F}.
\end{aligned}
\end{equation}
\item The linear operator 
\begin{equation}\label{defPAd}
\rP_{d,n}=\frac{1}{d!}\sum_{\sigma\in\Omega_d}\mathrm{sign}(\sigma)\rT_{d,n}: \otimes^d\F^n\to \rA^d\F^n
\end{equation}
is an orthogonal projection on $\rA^d\F^n$ with respect to $\langle\cdot,\cdot\rangle$.
Furthermore,  $\rP_{d,n}$ is nonexpansive, (norm decreasing), with respect to the spectral and nuclear norms
\end{enumerate}
\end{lemma}
\begin{proof} 
\noindent \emph{(1)}  It is straightforward to show that 
\begin{equation}\label{Tsigprop}
\rT(\sigma)\x_1\wedge\cdots\wedge\x_d=\textrm{sign}(\sigma)\x_{\sigma(1)}\wedge\cdots\wedge\x_{\sigma(d)}.
\end{equation}
Hence,  $\x_1\wedge\cdots\wedge\x_d\in \rA^d\F^n$.    Observe that $\x_1\wedge\cdots\wedge\x_d$ is multi linear.  Hence,  $\x_1\wedge\cdots\wedge\x_d=0$ if $x_i\ne x_j$ are linearly dependent for  some $i\ne j$.  Furthermore,   $\x_1\wedge\cdots\wedge\x_d=0$ if $\x_1,\ldots,\x_n$ are linearly dependent.  Assume,  that $\x_1,\ldots,\x_d$ are linearly independent.  Let $\y_i=\sum_{j=1}^d a_{i,j}\x_j$ for $i\in[d]$.  Let $A=[a_{i,j}]\in \F^{d\times d}$.  It is straightforward to show that $\y_1\wedge\cdots\wedge\y_d=(\det A)\x_1\wedge\cdots\wedge\x_d$.   This shows that the tensor $\y_1\wedge\cdots\wedge\y_d$ is colinear with   $\x_1\wedge\cdots\wedge\x_d$.  

Perform the Gram-Schmidt process without normalization  on $\x_1,\ldots,\x_d$ to obtain 
an orthogonal set of vectors $\y_1,\ldots,\y_d$ such that  $y_1\wedge\cdots\wedge\y_d=\x_1\wedge\cdots\wedge\x_d$.   Recall that $\|\y_i\|\le \|\x_i\|$ for $i\in[d]$.
As $\y_1,\dots,y_d$ is an orthogonal set if follows that 
\begin{equation*}
\|\y_1\|\cdots\|\y_d\|=\| \y_1\wedge\cdots\wedge\y_d\|=\|\x_1\wedge\cdots\wedge\x_d\|\le \|\x_1\|\cdots\|\x_d\|.
\end{equation*}
Assume that $\|\x_1\wedge\cdots\wedge\x_d\|=1$.  In the above formula we can renormalize that orthogonal vectors $z_i=t_i\y_i, a_i>0,i\in[d]$ such that $\z_1,\ldots,\z_d$ is an orthonormal basis in $\W$.   Change the orthonormal  basis $\z_1,\ldots,\z_d$ of $\W$ to an orthonormal basis $\bu_1,\ldots,\bu_d$ in $\W$ to deduce the last 
statement of \emph{(1)}.\\

\noindent
\emph{(2)} Assume that $\cX=\otimes_{j=1}^d \x_i, \cY=\otimes_{j=1}^d \y_i$.   Then
\begin{equation*}
\begin{aligned}
& \rT_{d,n}(\sigma) \rT_{d,n}(\eta)\cX=\otimes_{j=1}^d \x_{\sigma(\eta(j))}= \rT_{d,n}(\sigma\eta)\cX,\\
&\langle \rT_{d,n}(\sigma)\cX,\cY\rangle=\prod_{j=1}\langle\x_{\sigma(j)},\y_j\rangle=
\prod_{k=1}\langle\x_{k},\y_{\sigma^{-1}(k)}\rangle=\langle \cX,\rT_{d,n}(\sigma^{-1})\cY\}\rangle.
\end{aligned}
\end{equation*}
For a general $\cX,\cY$ write down $\cX$ and $\cY$ as a sum of rank-one matrices, and use the above equalities to deduce the first two equalities in \eqref{sigeq}. 

Next observe 
\begin{equation*}
\|\rT(\sigma)\cX\|^2=\langle \rT(\sigma)\cX,\rT(\sigma)\cX\rangle=\langle \cX,\rT(\sigma^{-1})\rT(\sigma)\cX\rangle=\|\cX\|^2.
\end{equation*}
This establishes the third  equality in \eqref{sigeq}. 

Let $\cY\in\Pi_{n^{\times^d}}$.  Observe that $\langle \cY,\rT_{d,n}(\sigma)\cX\rangle=
\langle \rT_{d,n}(\sigma^{-1})\cY,\cX\rangle$.  Clearly, $\rT_{d,n}(\sigma^{-1})\Pi_{n^{\times^d}}=\Pi_{n^{\times^d}}$.  Use the characterization of $\|\cX\|_{\infty,\F}$ to deduce the fourth equality in \eqref{sigeq}.   To deduce the last equality in \eqref{sigeq} use the characterization of the norm $\nu^*$ in terms of the norm $\nu$ given in \eqref{defnustar}.\\

\noindent \emph{(3)}  Observe that 
\begin{equation}\label{wedxPiden}
\x_1\wedge\cdots\wedge\x_d=\sqrt{d!}\,\rP_{d,n}\otimes_{j=1}^d \x_j\in \rA^d. 
\end{equation}
 Express $0\ne \cX$ as a sum of rank-one tensors to deduce that $\rP_{d,n}\cX\in\rA^d\F^n$.    
Observe that the equality \eqref{Tsigprop} yields that $\rP_{d,n}\x_1\wedge\cdots\wedge\x_d=\x_1\wedge\cdots\wedge\x_d$.  Hence,  $\rP_{d,n}\cX=\cX$ for  $\cX$ is skew-symmetric.  Assume that $\cY\in \otimes^d\F^n$ is orthogonal to $\rA^d\F^n$.   Hence, $\langle \cY,\cX\rangle =0$ for $\cX=\x_1\wedge\cdots\wedge\x_d$.  The equalities in \eqref{sigeq} and \eqref{Tsigprop} yield
\begin{equation*}
\langle\rT(\sigma)\cY,\cX\rangle=\langle \cY,\rT(\sigma^{-1})\cX\rangle=\langle \cY,\mathrm{sign}(\sigma^{-1})\cX\rangle=0.
\end{equation*}
This shows that $\rT(\sigma)\cY$ is orthogonal to $\rA^d\F^n$.  Hence,  $\rP_{n,d}$ is an orthogonal projection on $\rA^d\F^n$.

The equalities last two equalities in \eqref{sigeq} yield:
\begin{equation*}
\begin{aligned}
&\|\rP_{d,n}\|_{\infty,\F}\le \frac{1}{d!}\sum_{\sigma\in\Omega_d}\|\rT(\sigma)\|_{\infty,\F}\le 1, \\
&\|\rP_{d,n}\|_{1,\F}\le \frac{1}{d!}\sum_{\sigma\in\Omega_d}\|\rT(\sigma)\|_{1,\F}\le 1.
\end{aligned}
\end{equation*}
\end{proof}

Recall that the variety of all $d$-dimensional subspaces $\W$ in $\F^n$ is called Grassmannian Gr$(d,\F^n)$.   The span of all $\x_1\wedge\cdots\wedge\x_d$ is $\rA^d\F^n$, which is also denoted as $\bigwedge^d\F^n$.  The following lemma is well known,  and we bring its proof for completeness.
\begin{lemma}\label{strucAdFn} 
\begin{enumerate}
\item The linear subspace $\rA^d\F^n$ of $\F^{n^{\times d}}$ has dimension $n\choose d$.
\item Let $\be_1,\ldots,\be_n$ be an orthonormal basis of $\F^n$.  Then the set
$\be_{i_1}\wedge\cdots\wedge\be_{i_d}$ for $1\le i_1<\cdots<i_d\le n$
 is an orthonormal basis of $\rA^d\F^n$.
 \item Let $0\ne \cF\in\rA^d\F^n$.   Then there exists a unique subspace $\W$ of a minimal dimension $k\ge d$  in $\F^n$ such that $\cF\in \rA^d\W$.
\end{enumerate}
\end{lemma}
\begin{proof} \emph{(1)+(2)} The proof is by induction on $d\ge 2$.  For $d=2$ a skew-symmetric 2-tensor correspond to a skew-symmetric matrix.   Clearly,  $\frac{1}{\sqrt{2}}(\be_i\be_j^\top-\be_j\be_i^\top)$ is an orthonormal basis for skew symmetric matrices.   Assume the claim holds for $k, 2\le k<n-1$, and let $d=k+1$.   Assume that $0\ne \cF=[f_{i_1,\ldots,i_d}]\in \rA^d\F^n$.    Set
\begin{equation*}
\begin{aligned}
&\cF=\sum_{(i_1,\ldots,i_d)\in[n]^d} f_{i_1,\ldots,i_d}\otimes_{j=1}^d \be_{i_j},\\
&\cF_l=\sum_{(i_1,\ldots,i_{k})\in[n]^{k}}f_{i_1,\ldots,i_{k},l}\otimes_{j=1}^k\be_{i_j}.
\end{aligned}
\end{equation*}
Observe that $\cF_l$ can be viewed as a tensor in $\rA^k\F$.  The induction hypothesis yields that each $\cF_i$ is a  linear combination of $\be_{i_1}\wedge\cdots\wedge\be_{i_k}$ for $1\le i_1<\cdots<i_k\le n$.
Observe that $\cF=\sum_{l=1}^n \cF_l\otimes\be_l$.  Combine the above two observations and the fact that $\cF$ is a skew-symmetric tensor to deduce \emph{(1)+(2)}.

\noindent \emph{(3)}.  Unfold $\cF$ as a matrix $F$ whose rows are $\cF_1,\ldots,\cF_n$.  Then $\W$ is the column space of $F$.
\end{proof}
\begin{definition}\label{deffrank}
\begin{enumerate}
\item For a given $k\in[n]$ and $\W\in \mathrm{Gr}(k,\F^n)$ denote by $[\W]$ the set of all vectors $\x_1\wedge\cdots\wedge\x_k$,  where $\x_1,\ldots,\x_k$ is a basis in $\W$,  and by
$[\W]_1$ the set  of vectors $\x_1\wedge\cdots\wedge\x_k$, where $\x_1,\ldots,\x_k$ is an orthonormal basis in $\W$.  Then $[\W]_1=\{\cY\in \rA^d\F^n: \cY=\zeta \x_1\wedge\cdots\wedge\x_k, \zeta\in\F, |\zeta|=1\}$ for a fixed orthonormal basis $\x_1,\ldots,\x_k$ in $\W$.
\item Let $0\ne \cF\in \rA^d\F^n$.  
Then $\varphi$$\mathrm{rank}\,\cF$ is the minimal number of $\W_1,\ldots,\W_r\in  \mathrm{Gr}(d,\F^n)$ such that $\cF=\sum_{i=1}^r \cY_i, \cY_i\in[\W_i], i\in[r]$.
\end{enumerate}
\end{definition}
The $\varphi$rank  is called Slater rank in \cite{BL01}.
\begin{theorem}\label{fermnorms} Assume that $0\ne \mathcal{F}\in\mathrm{A}^d\mathbb{F}^n$, where $n\ge d$.  Then
\begin{equation}\label{fcharspecnrm}
\begin{aligned}
&\|\mathcal{F}\|_{\infty,\mathbb{F}}=\max_{\y_1,\ldots,\y_d\in\F^n, \langle \y_j,\y_k\rangle=\delta_{j,k},j,k\in[d]}\Re\langle\frac{1}{\sqrt{d!}}\y_1
\wedge\cdots\wedge\y_d,\cF\rangle=\\
&\max_{\cY\in [\W]_1, \W\in \mathrm{Gr}(d,\F^n)}\Re\langle \frac{1}{\sqrt{d!}}\cY,\cF\rangle,\\
&\|\mathcal{F}\|_{1,\mathbb{F}}=\min_{\cF=\sum_{i=1}^r \frac{a_i}{\sqrt{d!}}\cY_i,\cY_i\in[\W_i]_1, \W_i\in\mathrm{Gr}(d,\F^n), r\le r(n,d,\F)}\sum_{i=1}^r |a_i|,\\
&r(n,d,\R)={n\choose d}+1, \quad r(n,d,\C)=2{n\choose d}+1.
\end{aligned}
\end{equation}
In particular, the restrictions of $\|\cdot\|_{\infty,\F}$ and $\|\cdot\|_{1,\F}$ on $\otimes^d\F^n$ to $\rA^d\F^n$ are conjugate norms.
\end{theorem}
\begin{proof} 
Recall that for $0\ne \cF\in\rA^d\F^n$  the following characterization holds:
$$\|\cF\|_{\infty,\F}=\max_{\x_j\in\F^n, \|\x_j\|=1,j\in[d]}\Re\langle \otimes_{j=1}^d\x_j,\cF\rangle=\langle \otimes_{j=1}^d\y_j,\cF\rangle>0, \quad \|\y_j\|=1, j\in[d].$$
Use the equality  $\rP_{d,n}\cF=\cF$, the second equaity in \eqref{sigeq}, and\eqref{wedxPiden} to deduce
\begin{equation}\label{ferid}
\begin{aligned}
&\langle \otimes_{j=1}^d\x_i,\cF\rangle=
\frac{1}{\sqrt{d!}}\langle\x_1\wedge\cdots\wedge\x_d,\cF\rangle,\\
&\|\cF\|_{\infty,\F}=\langle \otimes_{j=1}^d\y_i,\cF\rangle=
\frac{1}{\sqrt{d!}}\langle\y_1\wedge\cdots\wedge\y_d,\cF\rangle.
\end{aligned}
\end{equation}
As $\|\x_1\wedge\cdots\wedge\x_d\|\le 1$ the above equality yields that $\y_1,\ldots,\y_d$ is an orthonormal set of vectors.  This proves the first characterization in \eqref{fcharspecnrm}.

We next recall that $\|\cF\|_{1,\F}=\max_{\|\cX\|_{\infty,\F}\le 1,\cX\in\otimes^d\F^n}\Re\langle \cX,\cF\rangle$.   As $\rP_{d,n}$ is an orthogonal projection on $\rA^d\F^n$ we deduce $\langle \cX,\cF\rangle=\langle \cX, \rP_{d,n}\cF\rangle=\langle\rP_{d,n}\cX,\cF\rangle$.  Assume that $\|\cX\|_{\infty,\F}\le 1$.  Then $\cY=\rP_{d,n}\cX\in \rA^d\F^n$, and part \emph{(3)} of Lemma \ref{norAdinf1lem} yields that $\|\cY\|_{\infty,\F}\le \|\cX\|_{\infty,\F}\le 1$.   Hence, 
\begin{equation}\label{intchar1nrm}
\|\cF\|_{1,\infty}=\max_{\cY\in\rA^d\F^n, \|\cY\|_{\infty,\F}\le 1} \Re \langle \cY,\cF\rangle.
\end{equation}

That is, the restrictions of $\|\cdot\|_{\infty,\F}$ and $\|\cdot\|_{1,\F}$ on $\otimes^d\F^n$ to $\rA^d\F^n$, denoted as $\nu$ and $\nu^*$ respectivel,y are conjugate norms. on $\rA^d\F^n$.  
As in the proof of Theorem \ref{1inftynormlem}, the characterization \eqref{fcharspecnrm} yields that the compact set of $\{\frac{1}{\sqrt{d!}}\cY:\cY\in [\W]_1, \W\in \mathrm{Gr}(d,\F^n)\}$ is the set of the extreme points of the $\nu^*$.   We deduce the characterization of $\|\cF\|_{1,\F}$ in \eqref{fcharspecnrm} as in the  proof of Theorem \ref{1inftynormlem}.
\end{proof}
The first equality in \eqref{fcharspecnrm} appears in \cite[Proposition 16]{ZCH10}.
The analog of the constants $\alpha_{s,\bn,\F}$ and  $\beta_{s,\bn,\F}$ for $s=\infty,1$ given in \eqref{defalbeinf1} are
\begin{equation}\label{defalbeinff}
\begin{aligned}
&\alpha_{\infty,n^{\times d},a,\F}=\min_{\cX\in\rA^d\F^n,\|\cX\|=1} \|\cX\|_{\infty,\F}, \quad \beta_{\infty,n^{\times d},a,\F}=\max_{\cX\in\rA^d\F^n,\|\cX\|=1} \|\cX\|_{\infty,\F},\\
&\alpha_{1,n^{\times d},a,\F}=\min_{\cX\in\F^n,\|\cX\|=1} \|\cX\|_{1,\F}, \quad \beta_{\infty,d,n\F}=\max_{\cX\in\rA^d\F^n,\|\cX\|=1} \|\cX\|_{1,\F}.
\end{aligned}
\end{equation}
The following theorem is an analog of Theorem \ref{relalphbetsym}.  
\begin{theorem}\label{relalphbet}  Let $2\le d< n$ be integers.  Then 
\begin{enumerate}
\item The following inequalities and equalites hold.
\begin{equation}\label{alphbetsineqa}
\begin{aligned}
&\alpha_{\infty,n^{\times d},a,\F}<\beta_{\infty,n^{\times d},a,\F}=\frac{1}{\sqrt{d!}},\\
&\alpha_{1,n^{\times d},a,\F}=\sqrt{d!}<\beta_{1,n^{\times d},a,\F}=\frac{1}{\alpha_{\infty,n^{\times d},a,\F}}.
\end{aligned}
\end{equation}
The equalites $\beta_{\infty,n^{\times d},a,\F}=\|\cX\|_{\infty,\F}, \alpha_{1,n^{\times d},a,\F}=\|\cY\|_{1,\F}$ achieved only for tensors  of the form $\x_1\wedge\cdots\wedge\x_d$, where $\x_1,\ldots,\x_d$ are orthonormal vectors in $\F^n$.
\item Assume that $\cF\in \rA^d\F^n$ and $\|\cF\|=1$.  Then $\|\cF\|_{\infty,\F}=\alpha_{\infty,n^{\times d},a,\F}$ if and only if $\|\cF\|_{1,\F}=\beta_{1,n^{\times d},a,\F}$.
\item For an even $n$: $\alpha_{\infty,n^{\times 2},a,\F}=\frac{1}{\sqrt{n}}\,\; \beta_{1,n^{\times 2},a,\F}=\sqrt{n}$.  Assume that $F$ is an $n\times n$ complex valued skew-symmetric matrix having Frobenius norm one: $\|F\|=1$.  Then $\|F\|_{1,\F}=\beta_{1,n^{\times 2},a,\F}$ if and only if 
$\sqrt{n}F$ is a unitary matrix.
\item For an odd $n$: $\alpha_{\infty,n^{\times 2},a,\F}=\frac{1}{\sqrt{n-1}}\,\; \beta_{1,n^{\times 2},a,\F}=\sqrt{n-1}$.   Assume that $G$ is an $n\times n$ complex valued skew-symmetric matrix having Frobenius norm one: $\|G\|=1$.  Then $\|G\|_{1,\F}=\beta_{1,2,n,\F}$ if and only if $G=F\oplus 0$ and 
$\sqrt{n-1}F$ is a unitary matrix.
\end{enumerate}
\end{theorem}
\begin{proof}
\emph{(1)}  
Recall that $\rA^d\F^n$ has a following orthonormal basis: 
\begin{equation}\label{onbAdFn}
\x_{i_1}\wedge\cdots\wedge\x_{i_d},  1\le i_1<\cdots<i_d\le n, \quad \langle \x_i,\x_j\rangle=\delta_{ij}, i,j\in[n].
\end{equation}
Let $\y_1,\ldots,\y_n$ be an orthonormal basis of $\F^n$.
 Assume $\cF\in\rA^d\F^n, \|\cF\|=1$.  Then 
\begin{equation*}
\begin{aligned}
&\cF=\sum_{1\le i_1<\cdots< i_d\le n} c_{i_1,\ldots,i_d} \y_{i_1}\wedge\cdots \wedge\y_{i_d}, \sum_{1\le i_1<\cdots <i_d\le n} |c_{i_1,\ldots,i_d}|^2=1, \\
&c_{i_1,\ldots,i_d}=\langle \y_{i_1}\wedge\cdots\wedge \y_{i_d},\cF\rangle, \quad 1\le i_1<\cdots<i_d\le n.
\end{aligned}
\end{equation*}
Then $|\frac{1}{\sqrt{d!}}\langle \y_1\wedge\cdots\wedge\y_d,\cF\rangle|=\frac{|c_{1,\ldots,d}|}{\sqrt{d!}}|\le \frac{1}{\sqrt{d!}}$.   Thus, $|\frac{1}{\sqrt{d!}}\langle \y_1\wedge\cdots\wedge\y_d,\cF\rangle|=\frac{1}{\sqrt{d!}}$ if and only if $\cF=\z_1\wedge\cdots\wedge\z_d$ for some set of $d$ orthonormal vectors $\z_1,\ldots,\z_d$.  This shows the first set of inequalities in \eqref{alphbetsineqa}.
The second set of inequalities follows from the equalities 
$\beta_{1,n^{\times d},a,\F}=\frac{1}{\alpha_{\infty,n^{\times d},a,\F}}$ and $\beta_{\infty,n^{\times d},a,\F}=\frac{1}{\alpha_{1,n^{\times d},a,\F}}$ that follows from \eqref{albetm}. 

\noindent\emph{(2)}.  Follows from \eqref{maxmineq}.

\noindent\emph{(3)-(4)}.  Follow from the observation that the determinant of an odd skew-symmetric matrix is zero \cite{HJ13}.
\end{proof}

In part \emph{(3)} of Theorem \ref{charsepa} we show 
\begin{equation}\label{alph2nin}
\frac{1}{\sqrt{d!{n\choose d}}}\le \alpha_{\infty,n^{\times d},a,\F}.
\end{equation}

\subsection{Nuclear rank}\label{subsec:nucrank}
The notion of the nuclear rank was introduced in \cite{FL18}, see also \cite{BFZ}:
\begin{definition}\label{defnucrank} Let $2\le d\in\N$.  
\begin{enumerate}
\item
Let $\bn\in\N$.  Assume that $0\ne \cX\in\F^{\bn}$.
Then the nuclear rank of $\cX$, denoted as $\mathrm{nrank}\,\cX$, is the minimal $r$ such that 
\begin{equation*}
\|\cX\|_{1,\F}=\sum_{i=1}^r \|\cX_i\|, \quad \cX=\sum_{i=1}^r \cX_i, \cX_i\in \Sigma_{\bn,\F}, i\in[r].
\end{equation*}
\item Let $0\ne \cS\in\rS^d\F^n$.   Denote by $\mathrm{snrank}\,\cS$ the minimal $r$ such that 
\begin{equation}\label{defsnrank}
\cS=\sum_{i=1}^r \varepsilon_i\x^{\otimes d}_i,  \quad \0\in \x_i\in\F^n, \varepsilon_i=\pm 1,i\in[r], \quad \|\cS\|_{1,\F}=\sum_{i=1}^r\|\x_i\|^r.
\end{equation}
\item Let $0\ne \cF\in \rA^d\F^n$ .  Denote by $\mathrm{fnrank}\,\cF$ the minimal $r$ such that  
\begin{equation}\label{defnrank}
\cF=\sum_{i=1}^r \frac{1}{\sqrt{d!}}\cY_i,\cY_i\in[\W_i]\setminus\{0\}, \W_i\in\mathrm{Gr}(d,\F^n), \quad \|\cF\|_{1,\F}=\sum_{i=1}^r\|\cY_i\|.
\end{equation}
\end{enumerate}
For a zero tensor all the above nuclear ranks are $0$.
\end{definition}
Note that the ranks above depend on the field $\F=\R,\C$.
Observe that the notion of the nuclear rank coincides with the notion of the $\nu$rank, where $\nu$ is the nuclear norm on $\F^n,\rS^d\R^n,\rA^d\F^n$ respectively, see subsection \ref{subsec:nurank}.

Observe that for $d=2$
\begin{equation}\label{d2ranks} 
\begin{aligned}
&\textrm{nrank}\, T=\rank T,  \quad T\in \F^{m\times n},\\
&\textrm{snrank}\, T=\rank T,  \quad T\in\rS^2\F^n,\\
&\varphi\textrm{nrank}\, T=\frac{\rank T}{2},  \quad T\in\rA^2\F^n.
\end{aligned}
\end{equation}
Let 
\begin{equation}\label{defmaxnrank}
\begin{aligned}
&\textrm{maxnrank}(\bn):=\max_{\cX\in\F^{\bn}} \textrm{nrank}\,\cX,\\
&\textrm{maxsnrank}(n^{\times d}):=\max_{\cS\in\rS^d\F^n} \textrm{snrank}\,\cS,\\
&\textrm{max}\varphi\textrm{nrank}(d,n):=\max_{\cF\in\rA^d\F^n} \varphi\textrm{nrank}\,\cF.
\end{aligned}
\end{equation}
\section{Tensor formulations of entanglement }\label{sec:ent}
\subsection{Two measures of entanglement}\label{subsec:2ment}
Assume that $\cH$ is a finite dimensional  vector space  over $\F$ with an inner product $\langle \x,\y\rangle$ and the norm $\|\x\|=\sqrt{\langle \x,\x\rangle}$.  By considering a standard orthonormal basis $\be_i=(\delta_{1i},\ldots,\delta_{ni})^\top, i\in[n]:=\{1,\ldots,n\}$ we identify $\cH$ with $\F^n$.   In the rest of this section we assume that $\F=\C$.   Furthermore, when we omit the field $\F$, we assume that $\F=\C$.
Dirac's notation for the standard basis in $\cH$ of dimension $n$ is $|0\rangle,\ldots, |n-1\rangle$.  Then $\langle \x,\y\rangle= \langle \x|\y\rangle$, which corresponds to $\x^{\dagger}\y$ in $\C^n$.   In quantum physics,
the set of states in $\cH$ is $\mathbb{S}(\cH)=\{\x\in\cH: \|\x\|=1\}$.  Two states $\x,\y \in \mathbb{S}(\cH)$ are identical if $\y=z\x$ for $z\in \mathbb{S}(\C)=\mathbb{S}^1$.
Let $\cH_j$ be a Hilbert space of dimension $n_j$ with a basis $\be_{1,j},\ldots,\be_{n_j,j}$ for $j\in[d]$, where $d\ge 2$.   Set $\cH:=\otimes_{i=1}^d \cH_i$, and $\bn=(n_1,\ldots,n_d)\in\N^d$. Then $\dim \cH=N(\bn)$,  and  
 $\cH\sim\C^{\bn}:=\otimes_{i=1}^d \C^{n_i}$.    A standard basis in $\cH$ in Dirac notation is $|(i_1-1)\cdots(i_d-1)\rangle, (i_1,\ldots,i_d)\in \bn$.

A pure (unentanged) state is a rank-one tensor of norm one.  The set of pure states is
 $\otimes_{i=1}^d \mathbb{S}(\cH_i)=\{\otimes_{i=1}^d \x_i, \|\x_i\|=1, i\in[d]\}=\Pi_{\bn}\subset \C^{\bn}$ is the Segre variety of unentangled states in $\cH$ \cite{Lan12}. 
The rank of a state $\cX$
is the first simple measurement of entanglement: $\cX$ entangled if and only if $\rank\cX\ge 2$.  In \cite{BFZ} the authors discuss the rank of tensors in the context of entanglement.  The problem with the quantity $\rank \cX$ is that it is not stable with respect to perturbation.  In particular, every tensor can be approximated by a tensor of generic rank \cite{BFZ}.  Eisert-Briegel \cite{EB01} use a convex roof construction to give an entanglement monotone function based on rank of tensors.

The geometric measure of entanglement (GME) of a state $\cX$ is dist$(\cX,\otimes_{i=1}^d \mathbb{S}(\cH_i))$ \cite{Shi95,BL01,WG03}.    It is straightforward to show that dist$(\cX,\otimes_{i=1}^d \mathbb{S}(\cH_i))=\sqrt{2(1-\|\cX\|_{\infty,\C})}$.  For the bi-partite state the GME is $\sqrt{2(1-\sigma_1(X))}$.   The maximal geometric entangled (MGE) state in $\cH$ is $\cX^\star$ corresponding to $\alpha_{\infty,\bn,\C}$ defined in \eqref{defalbeinf1}.
It is well known that for the bi-partite case the MGE state is the Bell state: $X=U[m^{-1/2}I_m|0]V$, where $U,V$ are unitary matrices, and $2\le m\le n$.  That is, $\alpha_{\infty,(m,n),\F}=m^{-1/2}$.  
In \cite{FK18} the authors introduced the following measurement of entanglement $-2\log \|\cX\|_{\infty,\C}$.   It is zero if and only if a state $\cX$ is pure state, and otherwise is positive.  It is shown in \cite{GFE09} and \cite{FK18} that most states and Boson states respectively are almost maximally entangled.

The nuclear norm of a state is another measure of entanglement.  Namely,  for a state $\cX\in\C^{\bn}$, express $\cX$ as a linear combination  of $r$ pure states:
$\cX=\sum_{i=1}^r a_i\cY_i,$ where $\cY_i\in \Pi_{\bn}$.  The characterization \eqref{defspecnuc} yields that $2\log\|\cX\|_{1,\C}$ is the minimal \emph{energy} to create $\cX$ from pure states.  Thus $2\log\|\cX\|_{1,\C}\ge 0$ and equality holds if and only $\cX$ is a pure state.  Then inequality \eqref{dualem1} yields that 
\begin{equation}\label{log1intyineq}
2\log \|\cX\|_{1,\C}\ge -2\log\|\cX\|_{\infty,\C}, \quad  \cX\in \otimes_{j=1}^d \cH_j,\|\cX\|=1.
\end{equation}
The equality $\alpha_{\infty,\bn,\F}\beta_{1,\bn,\F}=1$ in \eqref{alphbetsineq} yields
\begin{equation}\label{menteq}
\max_{\cX\in \C^{\bn}, \|\cX\|=1} -2log \|\cX\|_{\infty,\C}=\max_{\cX\in \C^{\bn}, \|\cX\|=1} 2log \|\cX\|_{1,\C}=-2\log\alpha_{\infty,\bn,\C}.
\end{equation}
Furthermore,  \eqref{alphbetsineq} implies that a GME state $\cX$ is also the most entangled with respect to the nuclear norm.

\subsection{Some known results on $\alpha_{\infty,\bn,\F}$}\label{subsec:krmes}
The following lower and upper estimate are known, and we state them for $\F=\R,\C$:
\begin{lemma}\label{lem:alphaineq}
Let $2\le d, 2\le n_1\le \cdots\le n_{d+1}$.  Then
\begin{equation}\label{alphlebetge}
\begin{aligned}
&\alpha_{\infty,(n_1,n_2,\dots,n_{d+1}),\F}\geq \frac{\alpha_{\infty,(n_1,n_2,\dots,n_d),\F}}{\sqrt{n_{d+1}}},\\
&\beta_{1,(n_1,n_2,\dots,n_{d+1}),\F}\le \sqrt{n_{d+1}}\beta_{1,(n_1,n_2,\dots,n_{d}),\F}.
\end{aligned}
\end{equation}
In particular,
\begin{equation}\label{lubalfbet}
\begin{aligned}
&\alpha_{\infty,(n_1,n_2,\dots,n_{d+1}),\F}\geq \frac{1}{\sqrt{n_1\cdots n_{d-1}}}\\
&\beta_{1,(n_1,n_2,\dots,n_{d+1}),\F}\le \sqrt{n_1\cdots n_{d-1}}.
\end{aligned}
\end{equation}
\end{lemma}
\begin{proof}
Assume that  $\alpha_{\infty, (n_1,\ldots,n_{d+1},\F}=\|\cT^{\star}\|_{\infty,\F}$.  Write
\begin{equation*}
\cT^\star=a_1\cT_1\otimes |0\rangle+\cdots +a_{n_{d+1}}\cT_{n_{d+1}}\otimes |n_{d+1}-1\rangle, 
\end{equation*}
where ${\mathcal T}_1,\dots,{\mathcal T}_{n_{d+1}}\in\F^{(n_1,\ldots,n_d)}$ are states, and $(a_1,\dots,a_{n_{d+1}})^\top\in \F^{n_{d+1}}$ has norm one.   For some $i$ we have $|a_i|\geq \frac{1}{\sqrt{n_{d+1}}}$.
There exist a pure state ${\mathcal X}\in\Pi_{(n_1,\ldots,n_d),\F}$ with $\|\cT_i\|_{\infty,\F}=|\langle {\mathcal X}, {\mathcal T}_i\rangle|\geq\alpha_{\infty,(n_1,\dots,n_{d}),\F}$.  Hence, 
$$
\|{\mathcal T}\|_{\infty,\F}\geq |\langle {\mathcal X}\otimes |i-1\rangle, {\mathcal T}\rangle|=|a_i|\cdot |\langle{\mathcal X}, {\mathcal T}_i\rangle|\geq \frac{1}{\sqrt{n_{d+1}}} \alpha_{\infty,(n_1,\dots,n_d),\F}.
$$
The second inequality in \eqref{alphlebetge} follows from the equality 
$\alpha_{\infty,(n_1,\dots,n_k),\F}\beta_{1,(n_1,\dots,n_k),\F}=1$ for $k=n_{d+1},n_d$.
The inequalities \eqref{lubalfbet} follow from \eqref{alphbetd=2} and \eqref{alphlebetge}. 
\end{proof}

We now discuss the MGE 3-qubit, i.e.  $n=2, d=3,\F=\C$.
Consider the tensor the 3-tensor $\cW$  \cite{DVC00} in Dirac's notation:
\begin{equation}\label{defWstate}
{\mathcal W}=|W\rangle=\frac{|100\rangle+|010\rangle+|001\rangle}{\sqrt{3}}.
\end{equation}
By Banach's theorem, the spectral norm of $|W\rangle$ is achieved on a pure symmetric state
$$
|S\rangle =(x\langle 0|+y\langle 1|)^{\otimes 3}=(x\langle 0|+y\langle 1|)\otimes(x\langle 0|+y\langle 1|)\otimes(x\langle 0|+y\langle 1|)
$$
with value
$$
|\langle S|W\rangle|=|\sqrt{3}yx^2|=\sqrt{3}|y| |x|^2.
$$
where $x,y\in \C$ with $|x|^2+|y|^2=1$. An easy calculus exercise shows that the maximum is achieved when $|x|=\sqrt{2}/\sqrt{3}$ and $|y|=1/\sqrt{3}$. We obtain
$$
\|{\mathcal W}\|_{\infty,\C}=\sqrt{3}\left(\frac{\sqrt{2}}{\sqrt{3}}\right)^2\frac{1}{\sqrt{3}}=\frac{2}{3}\mbox{ and }
\|{\mathcal W}\|_{1,\C}=\frac{3}{2}.
$$
(See also \cite[Section 6]{FL18}.)  It is shown in  \cite{CXZ10} that a  qutrit state $\cT\in\otimes^3\C^2$ is the MGE qutrit if and only if it is locally untiary equivalent to $\cW$.  
The orbit of $\cW$ is defined as
\[\textrm{orb}(\cW)=\{\cT, \;\cT=(A_1\otimes A_2\otimes A_3) \cW, A_1,A_2,A_3\in \rU_2\}.\]
Thus, $\cT$ is MGE if and only if $\cT\in$orb$(\cW)$.  In particular,
\begin{equation}\label{valalphbet23}
\alpha_{\infty,2^{\times 3},\C}=\alpha_{\infty,2^{\times 3},s,\C}=\frac{2}{3}, \quad \beta_{1,2^{\times 3},\C}=\beta_{1,2^{\times 3},s,\C}=\frac{3}{2}.
\end{equation}
A symmetric rank-one decomposition that achieves the nuclear norm of $\cW$ is:
$$
{\mathcal W}=\frac{1}{6\sqrt{3}}\left[\Big(\textstyle  \sqrt{2}\,|0\rangle +|1\rangle\Big)^{\otimes 3}+\zeta^2 \Big(\textstyle \sqrt{2} \,|0\rangle +\zeta |1\rangle\Big)^{\otimes 3}+
\zeta \Big(\textstyle \sqrt{2} \,|0\rangle +\zeta^2|1\rangle\Big)^{\otimes 3}\right], \zeta=\frac{-1+\sqrt{3}\bi}{2}.
$$ 

Combine Lemmas \ref{lem:alphaineq} with  \eqref{alphbetd=2}, \eqref{valalphbet23} 
and the equality $\alpha_{\infty, \bn,\F}\beta_{1,\bn,\F}=1$ to deduce
 \begin{corollary}\label{specnrmminineq}
 \begin{equation}\label{specnrmminineq1}
 \alpha_{\infty,n^{\times (d+1)},\F}\ge \frac{1}{\sqrt{n}}\alpha_{\infty,n^{\times d},\F}= \quad \textrm{ for } d\ge 2.
 \end{equation}
 In particular
 \begin{eqnarray}\label{specnrmminineq2}
 &&\alpha_{\infty,n^{\times d},\R}\ge n^{\frac{1-d}{2}}, \quad \beta_{1,n^{\times d},\R}\le n^{\frac{d-1}{2}}, \quad d\ge 2,\\
 &&\alpha_{\infty,2^{\times d},\C}\ge \left(\frac{2}{3}\right) 2^{\frac{-(d-3)}{2}},\quad 
 \beta_{1,2^{\times d},\C}\le \left(\frac{3}{2}\right) 2^{\frac{(d-3)}{2}}   \quad d\ge 3.\label{specnrmminineq3}
 \end{eqnarray}
 \end{corollary}
 We now show that the inequality \eqref{specnrmminineq2} is sharp for $n=2$ and  each $d\ge 2$.

For $1\le d\in \N$ and $\lambda\in \C$ with $|\lambda|=1$, we define a tensor
\begin{equation}\label{defTnlmab}
{\mathcal T}_{d,\lambda}=\frac{1}{\sqrt{2}}\left(\lambda \Big(\frac{|0\rangle +i |1\rangle}{\sqrt{2}}\Big)^{\otimes d}+\overline{\lambda} \Big(\frac{|0\rangle -i |1\rangle}{\sqrt{2}}\Big)^{\otimes d}\right)
\end{equation}
and we will use the convention ${\mathcal T}_d={\mathcal T}_{d,1}$. For example, we have
$$
{\mathcal T}_3=\frac{|000\rangle-|110\rangle-|101\rangle-|011\rangle}{2},
$$
$$
{\mathcal T}_4=\frac{|0000\rangle-|1100\rangle-|1010\rangle-|1001\rangle-|0110\rangle-|0101\rangle-|0011\rangle+|1111\rangle}{2\sqrt{2}}
$$
and
$$
{\mathcal T}_{4,-i}=\frac{|1000\rangle+|0100\rangle+|0010\rangle+|0001\rangle-|1110\rangle-|1101\rangle-|1011\rangle-|0111\rangle}{2\sqrt{2}}.
$$

As a complex tensor, the state ${\mathcal T}_{d,\lambda}$ is not much entangled in the sense of the nuclear or spectral norm. 
\begin{lemma}\label{Tnineq}  Let ${\mathcal T}_{d,\lambda}$ be defined as above.  Then
\begin{equation}\label{Tnineq1}
\|{\mathcal T}_{d,\lambda}\|_{\infty,\C}=1/\sqrt{2}, \quad \|{\mathcal T}_{d,\lambda}\|_{1,\C}=\sqrt{2}.
\end{equation}
\end{lemma}
\begin{proof}  Clearly, $\cT_{d,\lambda}$ is symmetric :
\begin{equation}\label{defTnlmab}
\cT_{d,\lambda}=\frac{1}{\sqrt{2}}(\otimes^d\bu +\otimes^d\bar\bu), \quad\bu=\frac{\lambda^{1/d}}{\sqrt{2}}(1,\bi)\trans.
\end{equation}
Furthermore, $\bu$ and $\bar\bu$ is an orthonormal basis in $\C^2$.  For $d\ge 2$ view the tensor $\cT_{d,\lambda}$ as a matrix $T$ of dimension $2^{d_1}\times 2^{d_2}$,
where $d_1=\lfloor\frac{d}{2}\rfloor, d_2=\lceil\frac{d}{2}\rceil$.  Hence \eqref{defTnlmab} is a singular value decomposition of $T$.  Thus $\sigma_1(T)=\sigma_2(T)=\frac{1}{\sqrt{2}}$
and $\|T\|_1=\sqrt{2}$.  Therefore, $\|\cT_{d,\lambda}\|_{\infty}\le \sigma_1(T)$ and $\|\cT_{d,\lambda}\|_1\ge \|T\|_1$.  However, this singular value decomposition are realized by 
left and right singular vectors which are rank one tensors.  Hence \eqref{Tnineq1} holds.\qed
\end{proof}
\begin{theorem}\label{maxrealentqub}
For every mixed real  $d$-qubit state ${\mathcal T}$ we have $\|{\mathcal T}\|_{\infty,\R}\geq 2^{(1-d)/2}$ and $\|{\mathcal T}\|_{1,\R}\leq 2^{(d-1)/2}$ and these inequalities are tight when ${\mathcal T}={\mathcal T}_{d,\lambda}$. In particular, we have
$$
\alpha_{\infty, 2^{\times n},\R}=2^{(d-1)/2}\mbox{ and }\beta_{1,2^{\times d},\R}=2^{(1-d)/2}.
$$
\end{theorem}
\begin{proof} The inequalities $\|{\mathcal T}\|_{\infty,\R}\geq 2^{(1-d)/2}$ and $\|{\mathcal T}\|_{1,\R}\leq 2^{(d-1)/2}$ are  the inequalities \eqref{specnrmminineq2} for $n=2$.

To calculate $\|{\mathcal T}_{n,\lambda}\|_{\infty,\R}$, we use Banach's theorem. We have to optimize
$$
\Big\langle (x\langle0|+y\langle 0|)^{\otimes d}, {\mathcal T}_{d,\lambda}\Big\rangle=\sqrt{2}\left|\Re\left(\lambda \Big(\frac{x+iy}{\sqrt{2}}\Big)^d\right)\right|
$$
under the constraint $x^2+y^2=1$. The optimal value clearly is equal to $2^{(1-d)/2}$ which shows that  $\|{\mathcal T}_{d,\lambda}\|_{\infty,\R}=2^{(1-d)/2}$.
It follows now that ${\mathcal T}_{d,\lambda}$ also has an optimal nuclear norm, which must be equal to $2^{(d-1)/2}$.
\end{proof}
\subsection{The most geometric entangled 4-qubit}\label{subsec:quadrit}
The following theorem was proven in \cite{DFLW17}.   We bring its detailed proof for reader's convenience. 
\begin{theorem}\label{M4maxentg} Let  $\cM_4=|M_4\rangle$ be the Higuchi-Sudbery state \cite{HS00}:
\begin{equation}\label{4qubM4}
\cM_4,=\frac{1}{\sqrt{6}}\big(|0011\rangle+|1100\rangle+
\zeta(|1010\rangle+|0101\rangle)
+\zeta^2(|1001\rangle+|0110\rangle\big),\;\zeta=e^{2\pi\bi/3}.
\end{equation}
Then $\|\cM_4\|_{\infty}=\frac{\sqrt{2}}{3}=\frac{1}{\sqrt{2}}\alpha_{\infty, 2^{\times 3},\C}$.  Hence  
\begin{equation}\label{beta4qub}
\alpha_{\infty,2^{\times 4},\C}=\frac{\sqrt{2}}{3}, \quad \beta_{1,2^{\times 4},\C}=\frac{3}{\sqrt{2}},
\end{equation}
\end{theorem}
and $\cM_4$ is the most entangled state.

\begin{proof} Let 
\begin{equation}\label{defcW01}
\begin{aligned}
&\cW_0=\frac{1}{\sqrt{3}}(|110\rangle+\zeta|101\rangle+\zeta^2|011\rangle),\\
&\cW_1=\frac{1}{\sqrt{3}}(|001\rangle+\zeta|010\rangle+\zeta^2|100\rangle).
\end{aligned}
\end{equation}
Let $\phi_3:\rU_2\to \otimes^3 \rU_2$ be the diagonal map $A\mapsto \otimes^3 A$ on the unitary $2\times 2$ matrices.  So $\phi_3(\rU_2)$ acts on $3$-qubits. 
Observe that 
\begin{equation*}
\begin{aligned}
&\cM_4=\frac{1}{\sqrt{2}}(\cW_0\otimes |0\rangle+\cW_1\otimes|1\rangle),\\
&\cW_1=(V_1\otimes V_2\otimes I_2)\cW, \quad V_1=\diag(\zeta^2,1), V_2=\diag(\zeta,1)\Rightarrow \|\cW_1\|_{\infty,\C}=\frac{2}{3},\\
&\cW_0=\phi_3(T)\cW_1, \quad T=\begin{bmatrix}0&1\\1&0\end{bmatrix}\Rightarrow \|\cW_0\|_{\infty,\C}=\frac{2}{3},
\end{aligned}
\end{equation*}
where $\cW$ is given by \eqref{defWstate}.
We claim that the two dimensional space $\W=$span$(\cW_0,\cW_1)$ is invariant under the action $\phi_3(\rU_2)$.  Consider first: 
\begin{equation}\label{diagac}
(\otimes^3 A)\cW_0=a b^2\cW_0,  (\otimes^3 A)\cW_1=a^2 b\cW_1, \quad A=\diag(a,b), |a|=|b|=1.
\end{equation}
Hence $\phi_3(A)\W=\W$.  It is left to show that $\W$ is invariant under the action of $\phi_3(\rS\rU_2)$.  Hence, it is enough to show $\W$ is invariant under the action of Lie group of $\phi_3(\rS\rU_2)$.
The generators of the $\rS\rU_2$ are $i$ times the Pauli matrices:
\begin{equation}\label{defBgen}
B_1=\left[\begin{array}{cc} i&0\\0&-i\end{array}\right], B_2=\left[\begin{array}{cc} 0&1\\-1&0\end{array}\right], B_3=\left[\begin{array}{cc} 0&i\\i&0\end{array}\right].
\end{equation}
The generators of $\phi_3(\rS\rU_2)$ are
\begin{equation}\label{defCgen}
C_j=B_j\otimes I_2\otimes I_2+I_2\otimes B_j\otimes I_2+I_2\otimes I_2\otimes B_j, \quad j\in [3].
\end{equation}
We need to  show that $C_j\cW_k\in\W$ for $j\in[3]$ and $k=0,1$.
In view of \eqref{diagac} it is enough to show the above containments for $j=2,3$.
Assume that $k=0$.
Observe
\begin{equation*}
\begin{aligned}
&C_2(|110\rangle)=|010\rangle +|100\rangle -|111\rangle ,\\
&C_2\zeta |101\rangle=\zeta(|001\rangle-|111\rangle+|100\rangle),\\
&C_2\zeta^2 |011\rangle=\zeta^2(-|111\rangle+|001\rangle+ |010\rangle).
\end{aligned}
\end{equation*}
As $1+\zeta+\zeta^2=0$
\begin{equation*}
C_2\cW_0=(\zeta+\zeta^2)\cW_1
\end{equation*}
Observe next
\begin{equation*}
\begin{aligned}
&\bar i C_3(|110\rangle=|010\rangle +|100\rangle+|111\rangle,\\
&\bar i C_3\zeta |101\rangle=\zeta(|001\rangle+|111\rangle+|100\rangle,\\
&\bar i C_3\zeta^2 |011\rangle=\zeta^2(|111\rangle+|001\rangle+|010\rangle)
\end{aligned}
\end{equation*}
Hence
\begin{equation*}
C_3\cW_0=\bi(\zeta+\zeta^2)\cW_1
\end{equation*}
Hence, $\phi_3(A)\cW_0\in \W$ for any $A\in \rU_2$.  Observe next that $\phi_3(T)\phi(A)=\phi(B)\phi_2(T)$, where $B=TAT$
As $\phi_3(T)\cW_0=\cW_1$ we deduce that $C_j\cW_1\in\W$ for $j\in[3]$.
Hence,  $\W$ invariant under the action $\phi_3(\rU_2)$.

Furthermore, the action of $\phi_3(\rU_2)$ on $\W$ is identical to the action of $\rU_2$ on the two dimensional subspace $\W$, with respect to the orthogonal
basis $\cW_0,\cW_1$.  That is, given a state $(a,b)\trans \in\C^2$, there exists $A\in \rU_2$
such that $\phi_3(A)\cW_0=a\cW_0+b\cW_1$.  Hence
\[\|a\cW_0+b\cW_1\|_{\infty}=\|\phi_3(A)\cW_0\|_{\infty}=\frac{2}{3} \quad \textrm{ for } |a|^2+|b|^2=1.\]
Observe that
\[\cM_4=\frac{1}{\sqrt{2}}\left(\cW_0\otimes|0\rangle + \cW_1\otimes|1\rangle\right).\]
Let $\cX=\x\otimes\y\otimes\bu\otimes \bv=\cY\otimes \bv$, where $\|\x\|=\|\y\|=\| \bu\|=\|\bv\|=1$, be a product state.  Then
\[\sqrt{2}|\langle \cX,\cM_4\rangle|=|\langle \cY,a\cW_0+b\cW_1\rangle|, \quad a=\langle \bv,|0\rangle\rangle, b=\langle\bv,|1\rangle\rangle.\]
Clearly $|a|^2+|b|^2=\|\bv\|^2=1$.  Thus if we maximize on all product states $\cY$ and keep $\bv$ fixed we get this this maximum is
$\|a\cW_0+b\cW_1\|_{\infty}=\frac{2}{3}$.  This shows that $\|\cM_4\|_{\infty}=\frac{\sqrt{2}}{3}$.  The inequaltity \eqref{specnrmminineq3} for $d=4$
yields that $\alpha_{\infty,2^{\times 4},\C}=\frac{\sqrt{2}}{3}$.  This equality yields the second equality in \eqref{beta4qub}.  Furthermore, $\cM_4$ is the most entangled
$4$-qubit.
\end{proof}

Numerical simulations point out that that  \eqref{specnrmminineq3} is not sharp for $d=5$. 
\subsection{Geometric measure of entanglement of Bosons}\label{subsec:geb}
Recall that the distance of a Boson $\cS\in\rS^d\C^n$ to product rank-one state is $\sqrt{2(1-\|\cS\|_{\infty,\C})}$.  
Then $\cS$ is a MGE Boson if $\|\cS\|_{\infty,\C}=\alpha_{\infty,n^{\times d},s,\C}$, see \eqref{defalbeinf1s}.
Theorem \ref{relalphbetsym} states that $\alpha_{\infty, n^{\times 2},s,\F}=\frac{1}{\sqrt{n}}$.
As $|W\rangle$ is a symmetric 3-qubit, we deduce that $\alpha_{\infty,2^{\times 3},s,\C}=\frac{2}{3}$. 

Aulbach-Markham-Murao \cite{AMM10} give examples of $d$-symmetric qubits for $d=4,\ldots,12$, which they assume to have the maximal GME.  Their examples are motivated by the Majorana model.  The most  geometric entangled symmetric $4$-qubit is conjectured to be \cite[Example 6.1]{AMM10}
\[\frac{1}{\sqrt{3}}|0000\rangle+\frac{1}{\sqrt{6}}(|0111\rangle+|1011\rangle+|1101\rangle + |1110\rangle).\]
Its spectral norm is $\frac{1}{\sqrt{3}}\approx 0.5774$ \cite{AMM10}.

In \cite{FW20} we verified the numerical values of these examples using our software, and we could not find symmetric qubits with higher GME.  Friedland-Wang \cite{FW20}  describe a way to compute the GME of Bosons using solutions of polynomial equations.

We claim that 
\begin{equation}\label{lualbet'in}
\frac{1}{\sqrt{{n+d-1\choose d}}}\le \alpha_{\infty,n^{\times d},s,\C}, \quad  \beta_{1,n^{\times d},s,\C}\le \sqrt{{n+d-1\choose d}}, \quad 2\le d,n.
\end{equation}
The first inequality is shown in \cite[Eq, (3.1)]{FK18}.   The second inequality follows from part \emph{(1)} of Theorem \ref{relalphbetsym}.   (We reprove this inequality in part \emph{(4)} of Corollary \ref{charsepscb}.) Hence
\begin{equation}\label{logsymin}
-2\log_2\alpha_{\infty,n^{\times d},s,\C}\le \log_2{n+d-1\choose d}.
\end{equation}
\subsection{Geometric measure of entanglement of Fermions}\label{subsec:geb}
An unentangled Fermion in $\cF\in\rA^d\C^n$ is of the form
\begin{equation}\label{unentf}
\cF=\x_1\wedge\cdots\wedge\x_d, \quad \x_i\in\C^n, \langle \x_i,\x_j\rangle=\delta_{i,j},i,j\in[n], 2\le d<n.
\end{equation}
Denote by $\Phi_{n^{\times d}}\subset \rA^d\C^n$ the set of  unentangled Fermions.
\begin{proposition}
Let $\cF\in\rA^d\C^n, \|\cF\|=1$ be a Fermion.  Then the GME of $\cF$ is given by 
\begin{equation}\label{GMEF}
\mathrm{dist}(\cF,\Phi_{n^{\times d}})=\sqrt{2(1-\sqrt{d!}\|\cF\|_{\infty,\C})}.  
\end{equation}
Most entangled Fermions correspond to the minimum $\alpha_{\infty, n^{\times d},a,\C}$ and maximum $\beta_{1, n^{\times d},a,\C}$.
\end{proposition}
\begin{proof} The formula \eqref{GMEF} follows from the first characterization of $\|\cF\|_{\infty,\C}$ in \eqref{fcharspecnrm}.  The characterization of a most entangled Fermion follows from Theorem \ref{albetm}.
\end{proof}
\subsection{Entanglements of $d$-qubits for big $d$}\label{subsec:bigd}
 Inequality \eqref{specnrmminineq3} yields that 
$$-2\log_2\alpha_{\infty,2^{\times d}}\le d-5 +2\log_2 3 \textrm{ for }d\ge 3.$$
The concentration result of \cite{GFE09} claims that most of $d$-qubits with respect to the corresponding Haar measure,  satisfy the inequality
$$-2\log_2\|\cT\|_{\infty,\C}\ge d-2\log_2 d-3 \textrm{ for most of } \cT\in\otimes^d\C^2 \textrm{ for } d\gg 1.$$
Hence, for a big $d$ the ratio $\frac{\alpha_{\infty, 2^{\times (d+1)},\C}}{\alpha_{\infty, 2^{\times d},\C}}$ approached to $\frac{1}{\sqrt{2}}$.

The following concentration inequality complementary to \eqref{logsymin} for Bosons is shown in \cite[Eq. (1.2)]{FK18} 
\[-2\log_2\|\cS\|_{\infty,\C}\ge  \log_2 {n+d-1\choose d} -\log_2\log_2 {n+d-1\choose d}-3\log_2 n-1, \textrm{ for most } \cS\in\rS^d\C^n \textrm { for } d\gg 1.\]

It seems to us the correct version of the inequality \eqref{alphlebetge} for $\alpha_{\infty,n^{\times d},s,\C}$ is:
\begin{equation}\label{conjinal'}
\alpha_{\infty, 2^{\times (d+1)},s,\C}= \frac{\sqrt{d+1}}{\sqrt{d+n}}\left(1+O(1/d)\right)\alpha_{\infty, 2^{\times d},s,\C}  \textrm{ for } d\gg 1.
\end{equation}

\section{Separability}\label{sec:sep}
\subsection{Density tensors and separability}\label{subsec:dentensep}
 Recall that  $\Pi_{\bn}$ is the set of product states.
\begin{definition}\label{defherm} Let $d,  n_1,\ldots, n_d\in\N$.  Set $\bn=(n_1,\ldots,n_d)$.   A tensor $\cB=[b_{i_1,\ldots,i_d,j_1,\ldots,j_d}]\in \C^{\bn\times \bn}:= \C^{(\bn,\bn)}=\C^{\bn}\otimes\C^{\bn}$ is called Hermitian 
\begin{equation*}
b_{i_1,\ldots,i_d,j_1,\ldots,j_d}=\overline{b_{j_1,\ldots,j_d,i_1,\ldots,i_d)}} \textrm{ for all } (i_1,\ldots,i_d),(j_1,\ldots,j_d)\in[\bn].
\end{equation*}
Denote by $\rH_{\bn}\subset \C^{\bn\times \bn}$ the real space of Hermitian tensors.
One can view $\cB$ as a Hermitian matrix $B=[b_{\bi,\bj}]\in \rH_{N(\bn)}$,
where $\bi=(i_1,\ldots,i_d),\bj=(j_1,\ldots,j_d)\in[n_d]$.   (Thus, $\rH_{N(\bn)}$ is unfolding of  $\rH_{\bn}$.)  

A Hermitian tensor is viewed a linear transformation $\cB:\C^{\bn}\to\C^{\bn}$:
\begin{equation*}
(\cB \cX)_{\bi}=\sum_{\bj\in[\bn]}b_{\bi,\bj}x_{\bj}, \quad \bi\in[\bn], \cX=[x_{\bj}]\in\C^{\bn}.
\end{equation*}  
It has the following spectral decomposition:
\begin{equation}\label{specdec}
\begin{aligned}
&\cB=\sum_{k=1}^{N(\bn)}\lambda_k(\cB)\cX_k\otimes\overline{\cX_k},\\
&\lambda_{\max}(\cB)=\lambda_1(\cB)\ge \cdots\ge \lambda_{N(\bn)}(\cB)=\lambda_{\min}(\cB),\\
&\overrightarrow{\lambda}(\cB)=(\lambda_1(\cB),\ldots,\lambda_{N(\bn)}),\\
&\cB\cX_k=\lambda_k(\cB)\cX_k, \quad \langle \cX_k,\cX_l\rangle=\delta_{kl}, k,l\in[N(\bn)].
\end{aligned}
\end{equation}
Let 
\begin{equation*}
\begin{aligned}
&\tr \cB=\sum_{(i_1,\ldots,i_d)\in [\bn]}b_{i_1,\ldots,i_d,i_1,\ldots,i_d}=\sum_{k=1}^{N(\bn)}\lambda_k(\cB),\\
&\tr_k\cB=\cB_k=[b^{(k)}_{i_1,\ldots,i_{k-1},i_{k+1},\ldots, i_d,j_1,\ldots,j_{k-1},j_{k+1},\ldots, j_d}]\in \C^{\bn_k\times \bn_k}, \bn_k=(n_1,\ldots,n_{k-1},n_{k+1},\ldots,n_d),\\
&b^{(k)}_{i_1,\ldots,i_{k-1},i_{k+1},\ldots, i_d,j_1,\ldots,j_{k-1},j_{k+1},\ldots,j_d}=\sum_{i_k=1}^{n_k} b_{i_1,\ldots,i_d,j_1,\ldots,j_{k-1}, i_k, j_{k+1}, \ldots,j_d}, \quad k\in[d].
\end{aligned}
\end{equation*}
The set of positive semi-definite Hermitian tensors, denoted as $\rH_{\bn,+}$ corresponds to $\rH_{N(\bn),+}$.  The set of positive semi-definite Hermitian tensors with trace one is denoted by $\rH_{\bn,+,1}$, is the convex set of density tensors. 
A rank-one density tensor is of the form $\cX\otimes\overline{\cX}, \cX\in\Pi_{\bn}$.
 
We define the inner product on $\rH_{\bn}$ by 
$\langle \cB,\cC\rangle=\sum_{\bi,\bj\in [\bn]}\overline{b_{\bi,\bj}}c_{\bi,\bj}$, which is real valued.
One associates with a Hermitian tensor $\cB$ a multi-sesquilinear form
\begin{equation}\label{mlsql}
\begin{aligned}
&\langle \otimes_{j=1}^d\x_j, \cB(\otimes_{j=1}^d\x_j)\rangle=
\sum_{(i_1,\ldots,i_d),(j_1,\ldots,j_d)\in [\bn]}b_{i_1,\ldots,i_d,j_1,\ldots,j_d}\prod_{(i_1,\ldots,i_d)\in[\bn]}\bar x_{i_1,1}x_{j_1,1}\cdots\bar x_{i_d,d} x_{j_d,d}, \\
&\x_k=(x_{1,k},\ldots, x_{n_k,k})^\top\in\C^{n_k}, k\in[d].
\end{aligned}
\end{equation}

\end{definition}
The  notion of Hermitian tensor for $d=2$  was introduced in \cite{FL18}, and was called there bi-Hermitian.
Recall that the operator norm on $\rH_{\bn}$ and its dual are the Schatten norms
$\|\overrightarrow{\lambda}(\cB)\|_\infty$ and $\|\overrightarrow{\lambda}(\cB)\|_1$.
Furthermore,
\begin{equation}\label{charopnrm}
\begin{aligned}
&|\overrightarrow{\lambda}(\cB)\|_\infty=\max(\lambda_{\max}(\cB),-\lambda_{\min}(\cB))= \max_{\cX\in\C^{\bn}, \|\cX\|=1} |\langle \cX,\cB\cX\rangle|,\\
&|\overrightarrow{\lambda}(\cB)\|_1=\sum_{i=1}^{N(\bn)} |\lambda_i(\cB)|.
\end{aligned}
\end{equation}
Hence, the extreme points of the Schatten norm $\|\overrightarrow{\lambda}(\cB)\|_1$
are $\pm \cX\otimes \overline{\cX}$ for $\cX\in\C^{\bn}, \|\cX\|=1$.

We define the spectral and the nuclear norms on $\rH_{\bn}$ as follows:
\begin{lemma}\label{specnuclem} Let $2\le d\in\N, 2\1_d\le \bn\in\N^d$.  Assume that $\cB\in\rH_{\bn}$.   Let
\begin{equation}\label{defspecnrm}
\|\cB\|_{spec}:=\max_{\cX\in\Pi_{\bn}}|\langle \cX,\cB\cX\rangle|.
\end{equation}
Then $\|\cB\|_{spec}$ is a norm on $\rH_{\bn}$.  Furthermore,
\begin{enumerate}
\item
\begin{equation}\label{specin}
\|\cB\|_{spec} \le \|\cB\|_{\infty,\C}\le \|\overrightarrow{\lambda}(\cB)\|_\infty.
\end{equation}
Equality holds in the first inequality holds if $\pm \cB\in\rH_{\bn,+}$.   Equality in the second inequality holds if and only if $\cB$ has an eigenvector $\cX\in\Pi_{\bn}$ which corresponds to an eigenvalue $\lambda_k(\cB)$ such that $|\lambda_k(\cB)|=\|\overrightarrow{\lambda}(\cB)\|_\infty$.    In particular, $ \|\cB\|_{spec} = \|\overrightarrow{\lambda}(\cB)\|_\infty\iff \|\cB\|_{\infty,\C}= \|\overrightarrow{\lambda}(\cB)\|_\infty$.  
\item Let $\|\cB\|_{nuc}$ be the dual norm of $\|\cdot\|_{spec}$:
\begin{equation}\label{nucnrm}
\|\cB\|_{nuc}=\max_{\cC\in\rH_{\bn},\|\cC\|_{spec}=1} \langle\cC,\cB\rangle.
\end{equation}
The extreme points of $\rB_{\|\cdot\|_{nuc}}(0,1)$ are $\pm\cX\otimes \overline{\cX}, \cX\in\Pi_{\bn}$.  Hence,
\begin{equation}\label{charnuc}
\|\cB\|_{nuc}=\min_{\cX=\sum_{i=1}^r a_i \cX_i\otimes \overline{\cX_i}, a_i\in\R, \cX_i\in\Pi_{\bn}, i\in[r], r\in\N }\sum_{i=1}^r |a_i|.
\end{equation}
Furthermore, the minimum is achieved for some $r\le N^2(\bn)+1$.
\item The following inequalities hold:
\begin{equation}\label{nucin}
\|\cB\|_{nuc}\ge\|\cB\|_{1,\C}\ge \|\overrightarrow{\lambda}(\cB)\|_1=\sum_{i=1}^{N(\bn)}|\lambda_i(\cB)|\ge |\tr \cB|.
\end{equation}
\item Denote
\begin{equation}\label{defalbeinh}
\begin{aligned}
&\alpha_{\infty,\bn,h}=\min_{\cX\in\rH_{\bn},\|\cX\|=1} \|\cX\|_{spec}, \quad \beta_{\infty,\bn,h}=\max_{\cX,\in\rH_{\bn},\|\cX\|=1} \|\cX\|_{spec},\\
&\alpha_{1,\bn,h}=\min_{\cX\in\rH_{\bn},\|\cX\|=1} \|\cX\|_{nuc}, \quad \beta_{1,\bn,h}=\max_{\cX\in\rH_{\bn},\|\cX\|=1} \|\cX\|_{1,nuc}.
\end{aligned}
\end{equation}
Then 
\begin{equation}\label{alpbethinh}
\begin{aligned}
&\alpha_{\infty,\bn,h}< \beta_{\infty,\bn,h}=1, \quad \alpha_{1,\bn,h}=1<\beta_{1,\bn,h},\\
&\alpha_{\infty,\bn,h}\beta_{1,\bn,h}=1.
\end{aligned}
\end{equation}
\item Let $\cB\in \rH_{\bn}, \cC\in \rH_{\bm}$.  Then $\cB\otimes\cC\in \rH_{(\bn,\bm)}$.  Moreover,
\begin{equation}\label{prodHmnrm}
\|\cB\otimes\cC\|_{spec}=\|\cB\|_{spec} \|\cC\|_{spec},  \quad \|\cB\otimes\cC\|_{nuc}=\|\cB\|_{nuc} \|\cC\|_{nuc}.
\end{equation}
\item The group $\rU_{\bn}$ acts on $\rH_{\bn}$: 
\begin{equation*}
\begin{aligned}
&\cX\otimes \overline{\cX}\mapsto((\otimes_{j=1}^d U_j)\cX)\otimes\overline{((\otimes_{j=1}^d U_j)\cX)},\\
&(\otimes_{j=1}^d \x_j)\otimes (\otimes_{j=1}^d\bar \x_j)\mapsto (\otimes_{j=1}^d U_j\x_j)\otimes (\otimes_{j=1}^d\overline{U_j\x_j})
\end{aligned}
\end{equation*}
as a group of isometries of the norms $\|\cdot\|_{spec}$ and $\|\cdot\|_{nuc}$.
Furthermore this action preserves the trace.
Let $\mu_{\bn}$ be the Haar measure on $\rU_{\bn}$.   Then for any $\cR\in\rH_{\bn,+,1}$ the following equality holds:
 \begin{equation}\label{Haareq}
 \int_{\rU_{\bn}} ((\otimes_{j=1}^d U_j)\cR \overline{(\otimes_{j=1}^dU_j)}d\mu_{\bn}=\frac{1}{N(\bn)} \otimes_{j=1}^d I_{n_j}.
 \end{equation}
\end{enumerate}
\end{lemma}
\begin{proof} It is clear that we have the following three  properties of a norm:
\begin{equation*}
\begin{aligned}
&0\le\|\cB\|_{spec}, \\
&\|b\cB\|_{spec}=|b|\|\cB\|_{spec} \textrm{ for } b\in\R,\\
&\|\cB_1+\cB_2\|_{spec}\le \|\cB_1\|_{spec} +\|\cB_2\|_{spec}.
\end{aligned}
\end{equation*}
Assume that $\|\cB\|_{spec}=0$.  Let $0\ne \cX=\otimes_{j=1}^d \x_j$.  Then $\langle \cX,\cB\cX\rangle$ is a real multi-sesquilinear form which is zero identically.   
 We prove by induction on $d$ that $\cB=0$.  If $d=1$,
then $\cB=0$ as $\|\cB\|_{spec}$ is the spectral norm of a Hermitian matrix.   Assume that  the induction hypothesis holds for $d=m$ and let $d=m+1$. Fix $\x_1,\ldots,\x_{m+1}\ne \0$.   Then $\langle \cX,\cB\cX\rangle$ is a real sesquilinear form which is zero identically.   Hence, we obtain $n_{m+1}^2$ real sesquilinear forms in $\x_1,\ldots,\x_{m}$ which identically zero.
The induction hypothesis implies that $\cB=0$.

\noindent \emph{(1)}  Observe that $\rH_{\bn}\subset \C^{\bn\times\bn}$.  Hence
\begin{equation}\label{cBinfchar}
\|\cB\|_{\infty,\C}=\max_{\cX,\cY\in\Pi_{\bn}}\Re \langle \cY, \cB\cX\rangle=\max_{\cX,\cY\in\Pi_{\bn}}|\langle \cY, \cB\cX\rangle|.
\end{equation}
Use the maximal characterizations \eqref{charopnrm},  \eqref{cBinfchar} and \eqref{defspecnrm} of
$|\overrightarrow{\lambda}(\cB)\|_\infty$,  $\|\cB)\|_{\infty,\C}$ and $\|\cB\|_{spec}$ respectively,  and the inclusion $\Pi_{\bn}\subset \mathbb{S}(\C^{\bn})$ to deduce \eqref{specin}.    The equality case in the first inequality for a positive semi-definite $\cB$ follows from the Cauchy-Schwarz inequality 
\begin{equation*}
|\langle \cY,\cB\cX\rangle|^2\le \langle \cY,\cB\cY\rangle \langle \cX,\cB\cX\rangle\Rightarrow \|\cB\|_{\infty,\C}=\max_{\cX\in\Pi_{\bn}} \langle \cX,\cB\cX\rangle \textrm{ for } \cB\in\rH_{\bn,+}.
\end{equation*}
Same inequality holds if $-\cB\in\rH_{\bn,+}$.  

Assume the equality $\|\cB\|_{\infty,\bn,\C}=\|\langle \cY^{\star}, \cB\cX^{\star}\rangle=\|\overrightarrow{\lambda}(\cB)\|_\infty$ for some $\cX^\star, \cY^\star\in\Pi_{\bn}$.  Then $\cY^*=z\cX^\star$, and $\cX^\star$ is an eigenvector $\cB:\C^{\bn}\to \C^{\bn}$ with eigenvalue $\lambda$ such that $|\lambda|=\|\overrightarrow{\lambda}(\cB)\|_\infty$.   Vice versa, assume that $\cB$ has an eigenvector $\cX\in\Pi_{\bn}$ which corresponds to an eigenvalue $\lambda_k(\cB)$ such that $|\lambda_k(\cB)|=\|\overrightarrow{\lambda}(\cB)\|_\infty$.   
 Then $\|\cB\|_{spec}\ge \|\overrightarrow{\lambda}(\cB)\|_\infty$ which yields  equalities in \eqref{specin}.

\noindent \emph{(2)} The characterization \eqref{defspecnrm} yields that the set of extreme points of the unit ball of the norm $\|\cdot\|_{nuc}$ is contained in  the set $\rE:=\{\pm \cX\otimes\overline{\cX}:\, \cX\in \Pi_{\bn}\} $.  Observe that $\rE$ is compact and contained in the set of the extreme points of the Euclidean norm in $\C^{\bn\times \bn}$.  Hence, $\rE$ is the set of the extreme points of the unit ball of the norm $\|\cdot\|_{nuc}$.  Therefore, as in the proof of the characterization of $\|\cX\|_{1,\F}$ in \eqref{defspecnuc} we deduce the characterization \eqref{charnuc}.
As $\dim\rH_{\bn}=\dim \rH_{N(\bn)}=N^2(\bn)$, Caratheodory's theorem yields that in the characterization \eqref{charnuc} we can assume that $r\le N^2(\bn)+1$. 

\noindent \emph{(3)} Consider the minimal characterizations of the norms $\|\cB\|_{1,\C}$ and $\|\cB\|_{1,nuc}$ given by \eqref{defspecnuc} for $\C^{\bn}\otimes \C^{\bn}$ 
and by \eqref{charnuc} respectively.   As the set of the extremal points of the unit ball $\|\cdot\|_{nuc}$ is contained the of the extremal points of the unit ball $\|\cdot\|_{1,\C}$ we deduce the first inequality in \eqref{nucin}.    The second  inequality in \eqref{nucin}
is deduced similarly.  The last inequality in \eqref{nucin} follows from the first equality after \eqref{specdec}.

\noindent\emph{(4)}  Straightforward. 

\noindent\emph{(5)}  Given $\cB=[b_{\bi,\bj}]\in \rH_{\bn}, \cC=[c_{\bk,\bl}]\in \rH_{\bm}$, we define $\cB\otimes\cC=\cD=[d_{\bi,\bk,\bj,\bl}] \in \rH_{(\bn,\bm)}$ by the equality $d_{\bi,\bk,\bj,\bl}:=b_{\bi,\bj}c_{\bk,\bl}$.
It is straightforward to show that if $\cB\in \rH_{\bn}, \cC\in \rH_{\bm}$ then $\cB\otimes\cC\in \rH_{(\bn,\bm)}$.   The  proof of equality \eqref{prodHmnrm} is similar to the proof of the equality \eqref{tenprodeq}.

\noindent\emph{(6)}  Assume that $\cR=(\otimes_{j=1}^d\x_j)\otimes (\otimes_{j=1}^d\bar\x_j)$, where $\|\x_j\|=1, j\in[d]$.   Clearly,  $1=\tr\cR=\tr  (\otimes_{j=1} U_j)\cR\overline{(\otimes_{j=1} U_j)}$.   As the extreme points of the unit ball $\|\cdot\|_{nuc}$ are $\pm \cX\otimes \overline{\cX}, \cX\in \Pi_{n^{\times d}}$ it follows that the action of $\rU_{\bn}$ on $\rH_{\bn}$ preserve the trace.

Observe that for $d=1$ the equality
$\int_{\rU_n}(U\x)(U\x)^*d\mu_n=\frac{1}{n}I_n$ for any $\x\in\C^n, \|\x\|=1$ follows from the decomposition $\frac{1}{n}I_n=\frac{1}{n}\sum_{i=1}^n \y_i\y_i^*$, where $\langle \y_i,\y_j\rangle=\delta_{i,j}, i,j\in[n]$.  As $d\mu_{\bn}=\prod_{j=1}^d\mu_{n_j}$ we deduce the equality \eqref{Haareq}.    Assume now that $\cR$ is any density tensor 
in $\rH_{\bn,+,1}$.  Let $\cQ=\int_{\rU_{\bn}} ((\otimes_{j=1}^d U_j)\cR \overline{(\otimes_{j=1}^dU_j)}d\mu_{\bn}$.   Observe the following equalities for $\cX\in \Pi_{\bn}$:
\begin{equation*}
\langle \cX,\cQ\cX\rangle=\langle\cX\otimes\overline{\cX}, \cQ\rangle=\langle \int_{\rU_{\bn}}(\otimes_{j=1}^d U_j)\cX)\otimes \overline{(\otimes_{j=1}^d U_j)\cX)}, \cR\rangle=\langle\frac{1}{N(\bn)}\otimes_{j=1}^d I_{n_j},\cR\rangle=\frac{1}{N(\bn)}.
\end{equation*}
That is,  the sesquilinear form $\langle \cX,\cQ\cX\rangle$ is constant on $\Pi_{\bn}$, which is equal to $\frac{1}{N(\bn)}$.  Use an induction on $d$ to deduce \eqref{Haareq}.  
\end{proof}
\begin{corollary}\label{corrh}
\begin{enumerate}
\item A rank-one Hermitian tensor is a nonzero real scalar times  a rank-one density tensor.
\item  Every nonzero Hermitian tensor is a linear combination of at most $N(\bn)^2$ rank-one density tensors.
\end{enumerate}
\end{corollary}
\begin{proof}\emph{(1)} Straightforward.

\noindent\emph{(2)} Let $0\ne \cB\in \rH_{\bn}$.  Then $\|\cB\|_{nuc}^{-1}\cB$ is a convex combination of  the extreme points, of the unit ball of $\|\cdot\|_{nuc}$, which are $\pm$ rank-one density tensors.   Caratheodory's theorem implies that $\|\cB\|_{nuc}^{-1}\cB$ is a convex combination of at most $\dim \rH_{\bn}+1$ extreme points.  In view of \eqref{specdec} $\rH_{\bn}$ is isomorphic to $\rH_{N(\bn)}$.  Hence,  $\dim \rH_{\bn}=N^2(\bn)$.  Thus, $\cB$ is a linear combination of $N^2(\bn)+1$ rank-one density tensors.  Since $N^2(\bn)+1$ rank-one density tensors are linearly dependent it follows that $\cB$ is a linear combination of $N(\bn)^2$ rank-one density tensors.
\end{proof}

For $0\ne \cB\in \rH_{\bn}$ we define the Hermitian rank of $\cB$, denoted as $\hrank\cB$,   the minimal $r$ such that
$\cB$ is a sum of $r$ rank-one Hermitian tensors.  
It seems to the author that for $d\ge 2$ the computation of $\hrank \cB$ is NP-hard.

The set $\rH_{\bn,+,1}$ corresponds to the density tensors  on the space $\C^{\bn}$.
A density tensor $\rho\in \rH_{\bn,+,1}$ is called a product density tensor if 
$\rho=\otimes_{j=1}^d \rho_j$, where $\rho_j\in \rH_{n_j}, j\in[d]$.  Viewing this tensor product as a Kronecker tensor product, we deduce that $\rank \rho$, viewed as a matrix in $\rH_{N(\bn)}$ is equal to $\prod_{j=1}^{d}\rank\rho_j$.   
For each $\rho_j$ write down its spectral decomposition.   It now follows that a product density tensor is a convex combination of at most $N(\bn)$ rank-one product density. tensors   
\begin{definition}\label{defsepdt}  The set of convex combinations of rank-one product density tensors in $\rH_{\bn,+,1}$ is called the set of separable density tensors, and denoted by $\mathrm{Sep}_{\bn}$.
\end{definition}

The set of separable density tensors $\mathrm{Sep}_{\bn}$ is considered to be the analog of unentangled states see \cite{Wer89}.   A local pseudomixture of a density tensor $\cR\in\rH_{\bn,+,1}$ is a decomposition of $\cR$ to a linear combination $\cR=\sum_{i=1}^M a_i\cR_i$, where each  $\cR_i$ is a rank-one density tensor.
It is shown in \cite{STV98} that for $\bn=(2,2)$ a separable $\cR$ has $\hrank\cB\le 4$ and a nonseparable $\cR$ satisfy $4\le \hrank\cR\le 5$.  As any $0\ne \cX\in\rH_{(2,2)}$ is equal to $\cY-\cZ, \cY,\cZ\in\rH_{(2,2),+}$ it follows that $\hrank \cX\le 10$.  This is much better than part \emph{(2)} of Corollary \ref{corrh}, which states that $\hrank \cX \le 16$.
We are not aware of improvement of Caratheodory's theorem for $\bn\ne (2,2)$.
See problem (5) in Section \ref{sec:oprb}. 

The following theorem states that the correct measurement of inseparability of $\cR$ is its nuclear norm,  which corresponds to a decomposition of $\cR=\sum_{i=1}^M a_i\cR_i$, where  $\sum_{i=1}^M|a_i|(=\|\cR\|_{nuc})$ is minimal.  
\begin{theorem}\label{charsep}  Let $1\le d\in\N$ and $2\le n_1,\ldots,n_d$.
Then
\begin{enumerate}
\item  A density tensor $\cR\in\rH_{\bn,+,1}$ satisfies the inequalities
\begin{equation}\label{lubsndt}
\begin{aligned}
&\frac{1}{N(\bn)}\le \|\cR\|_{spec}\le \lambda_1(\cR)\le 1,\\
&1\le \|\cR\|_{nuc}\le \beta'_{1,\bn,h}:=\max_{\cX\in \C^{\bn}, \|\cX\|=1}\|\cX\otimes\overline{\cX}\|_{nuc}\le  \beta_{1,\bn,h}.
\end{aligned}
\end{equation}
Equalities in some of the above inequalities hold under the following conditions:
\begin{equation}\label{lubsndteq}
\begin{aligned}
&\|\cR\|_{spec}=\lambda_1(\cR)\iff \exists \cX\in \Pi_{\bn} \textrm{ such that } \langle \cX,\cX_k\rangle = 0 \textrm{ if } \lambda_k <\lambda_1(\cR),\\
&\|\cR\|_{spec}=1\iff \cR=\cX\otimes \overline{\cX}, \cX\in\Pi_{\bn},\\
&\frac{1}{N(d)}=\|\cR\|_{spec} \iff \cR=\frac{1}{N(\bn)}\otimes^d_{j=1}I_{n_j},\\
&1=\|\cR\|_{nuc} \iff \cR\in\mathrm{Sep}_{\bn},\\
& \|\cR\|_{nuc}=\beta'_{1,\bn,h} \Leftarrow \cR=\cX\otimes \overline{\cX}, \cX\in\C^{\bn}, \|\cX\|=1, \|\cX\otimes\overline{\cX}\|_{1,\bn,h}=\beta'_{1,\bn,h}.
\end{aligned}
\end{equation}

\item The set $\mathrm{Sep}_{\bn}$ is a maximal face of $\rB_{\|\cdot\|_{nuc}}(0,1)$, the unit ball of $\|\cdot\|_{nuc}$, of the maximal dimension $N^2(\bn)-1$.  Its supporting hyperplane is
\begin{equation}\label{suphypsep}
 \tr \cB=\langle \otimes_{j=1}^d I_{n_j},\cB\rangle\le 1 \textrm{ for } \|\cB\|_{nuc}\le 1.
 \end{equation}
 Equality holds if and only if $\cB$ is separable.
In particular, any separable density tensor is a convex combinations of at most $N^2(\bn)$ rank-one density tensors.
\end{enumerate}
\end{theorem}
\begin{proof}  
\emph{(1)}
Let $\cR\in\rH_{\bn,+,1}$.  Consider the spectral decomposition $\cR=\sum_{k\in [N(\bn)]} \lambda_k(\cR)\cX_k\otimes \overline{\cX_k}$.  Then $\lambda_i(\cR)\ge 0$ for $i\in[N(\bn)]$ and $\tr\cR=\sum_{i=1}^{N(\bn)} \lambda_i(\cR)=1$.   
Thus, for $\|\cX\|=1$  one has the inequalities 
\begin{equation*}
|\langle \cX,\cR\cX\rangle|=\langle \cX,\cR\cX\rangle=\sum_{k\in [N(\bn)]} \lambda_k(\cR)|\langle\cX,\cX_k\rangle|^2\le \lambda_1(\cR) \sum_{k\in [N(\bn)]}|\langle\cX,\cX_k\rangle|^2=\lambda_1(\cR)\le 1.
\end{equation*}
This shows the last two inequalities in the first inequality in \eqref{lubsndt}.
Assume that  $\|\cR\|_{spec}=\lambda_1(\cR)$.  Then,  there exists $\cX\in\Pi_{\bn}$ such that $\|\cR\|_{spec}=\lambda_1(\cR)=\lambda_1(\cR)\langle \cX,\rangle \cR\cX\rangle$.  This equality holds if and only if that $\langle \cX,\cX_k\rangle=0$ for $\lambda_k(\cR)<\lambda_1(\cR)$.   This show that first equality in \eqref{lubsndteq}.

Assume that $\|\cR\|_{spec}=1$ then $\lambda_1(\cR)=1$, and $\lambda_k(\cR)=0$ for $k>1$.  Now use the implications for the equality $\|\cR\|=\lambda_1(\cR)$ to deduce the second equality case in \eqref{lubsndteq}.

Let $t=\min_{\cR\in \rH_{\bn,+,1}}\|\cR\|_{spec}=\|\cR^\star\|_{spec}$.    As $\rU_{\bn}$ acts as a group of isometries on $\|\cdot\|_{spec}$ it follows that $\|U\cR^\star\bar U\|_{spec}=t$ for each $U\in \rU_{\bn}$.  Hence,
\begin{equation*}
\begin{aligned}
&\frac{1}{N(\bn)}=\|\frac{1}{N(\bn)}\otimes_{j=1}^d I_{n_j}\|_{spec}=\|\int_{\rU_{\bn}} ((\otimes_{j=1}^d U_j)\cR^\star \overline{(\otimes_{j=1}^dU_j)}d\mu_{\bn}\|_{spec}\le \\
&\int_{\rU_{\bn}} \|((\otimes_{j=1}^d U_j)\cR^\star \overline{(\otimes_{j=1}^dU_j)}\|d\mu_{\bn}=t
\end{aligned}
\end{equation*}
This establishes the lower bound in the first inequality of \eqref{lubsndt}.  We now discuss the equality case.
Assume that $\frac{1}{N(\bn)}=\|\cR^{\star}\|_{spec}=\langle \cX,\cR^\star\cX\rangle=\langle \cX\otimes \overline{\cX},\cR^\star\rangle$ for some $\cX\in    \Pi_{\bn}$.  The equality case in the above inequality yields that for each $U\in\rU_{\bn}$ one has the equality $\frac{1}{N(\bn)}= \langle (U\cX), \cR^\star(U\cX)\rangle$.   That is,  the sesqulinear form $\langle\cY,\cR^\star\cY\rangle$ has value $\frac{1}{N(\bn)}$ on $\Pi_{\bn}$.  As in the proof of part \emph{(5)} of Lemma \ref{specnuclem}  we deduce that $\cR^\star=\frac{1}{N(\bn)} \otimes_{j=1}^d I_{n_j}$.

We now discuss the inequalities for $\|\cR\|_{nuc}$ in \eqref{lubsndt}. 
Use the decomposition of $\cR$, and the triangle inequality to deduce
\begin{equation*}
\|\cR\|_{nuc}\le \sum_{k\in[N(\bn)]}\lambda_k(\cR)\|\cX_k\otimes \overline{\cX_k}\|_{nuc}\le \max_{\cX\in\C^{\bn}, \|\cX\|=1}\|\cX\otimes \overline{\cX}\|_{nuc}=\beta'_{1,\bn,h}.
\end{equation*}
The definition \eqref{lubsndt} yields that $\beta'_{1,\bn,h}\le \beta_{1,\bn,h}$.  
The last implication in \eqref{lubsndteq} follows form the definitions.

The inequality \eqref{nucin} yields that $\|\cR\|_{nuc}\ge 1$.  Assume that $\cR$ is separable.  Then $\cR$ is a convex combination of rank-one density tensors, which are extreme points of $\rB_{\|\cdot\|_{nuc}(0,1)}$.  Hence,
$\|\cR\|_{nuc}\le 1$.  Use the inequality \eqref{nucin} to deduce that $\|\cR\|_{nuc}=
 1$.   Assume that $\|\cR\|_{nuc}=1$.   Then $\cR$ is a convex combination of the extreme points of the ball $\rB_{\|\cdot\|_{nuc}}(0,1)$.  As $\tr \cR=1$, $\cR\in \mathrm{Sep}_{\bn}$.  This conclude the proof of part \emph{(1)}.  
 
 \noindent  \emph{(2)}  Assume that $\|\cB\|_{nuc}\le 1$.   Then $\cB$ is a convex combination of $\pm \cX\otimes \overline{\cX},$ where $\cX\otimes \overline{\cX}$ is a rank-one density tensor.  Hence,  $\tr \cB\le 1$, and equality holds if and only if $\cB$ is a convex combination of rank-one density tensor.  
Therefore,  $\mathrm{Sep}_{\bn}$ is a maximal face of $\rB_{\bn}(0,1)$.  

It is left to show that the dimension of the convex set $\mathrm{Sep}_{\bn}$ is $N^2(\bn)-1$.  This is the dimension of the hyperplane $\{\cB\in\rH_{\bn}:\,\tr \cB=1\}$.
Observe that the dimension of $\mathrm{Sep}_{\bn}$  is the dimension of the convex set $\cB\in \rH_{\bn}:\cB=\frac{1}{2}(\cR_1-\cR_2), \cR_1,\cR_2\in \mathrm{Sep}_{\bn}$.  This is exactly the convex set obtained by the intersection of unit ball $\rB_{\|\cdot\|_{spec}}$ with the hyperplane $\tr \cB=0$.  Hence, $\dim \mathrm{Sep}_{\bn}=N^2(\bn)-1$.
\end{proof}

 The inequality $\|\cR\|_{nuc}\ge 1$ and the implication $\|\cR\|_{nuc}=1\iff \cR\in\mathrm{Sep}_{\bn}$ in Theorem \ref{charsep}  was proved in \cite{Rud00} for $d=2$ and for general $d$ in \cite{Per04}.
 \begin{corollary}\label{charsepsc}  Denote by $\mathrm{Sep}_{\bn,int}$ the relative interior of $\mathrm{Sep}_{\bn}$.  
\begin{enumerate}
\item The sets $\mathrm{Sep}_{\bn}$ and $\mathrm{Sep}_{\bn,int}$ are invariant under the action of $\rU_{\bn}$.
\item Let $\cR\in \mathrm{Sep}_{\bn,int}$.  There exists $r(\cR)>0$ such that $\rB_{\|\cdot\|}(\cR,r(\cR))\cap\rH_{\bn,+,1}\subset \mathrm{Sep}_{\bn}$.
\item The density tensor $\frac{1}{N(\bn)} \otimes_{j=1}^d I_{n_j}$ is in $\mathrm{Sep}_{\bn,int}$, where $r(\frac{1}{N(\bn)} \otimes_{j=1}^d I_{n_j})\ge \frac{2}{N(\bn)\beta_{1,\bn,h}}$.
\end{enumerate} 
 \end{corollary}
 \begin{proof}
 \emph{(1)} is straightforward.   
  
 \noindent \emph{(2)}   follows from the fact that $\dim\mathrm{Sep}_{\bn,int}=N^2(\bn)-1$.
 
\noindent \emph{(3)}  Let $0\ne \cB\in \rH_{\bn}, \tr \cB=0$.   We claim that 
\begin{equation*}
\begin{aligned}
&\cB=\sum_{l=1}^p a_l(\otimes_{j=1}^d \x_{j,l})\otimes(\otimes_{j=1}^d \bar\x_{j,l})-
\sum_{m=1}^q b_m(\otimes_{j=1}^d \y_{j,m})\otimes(\otimes_{j=1}^d \bar\y_{j,m}),\\
&a_l>0,  \|\x_{j,l}\|=1, b_m>0, \|\y_{j,m}\|=1,l\in[p],m\in[q], \\
&\|\cB\|_{nuc}=\sum_{l=1}^pa_l+\sum_{m=1}^q b_m,  \quad\sum_{l=1}^p a_l-\sum_{m=1}^q b_m=0.
\end{aligned}
\end{equation*}
The expansion of $\cB$, and the first equality in the last row of the above equalities follow from the  observation that $\frac{1}{\|\cB\|_{nuc}}\cB$
 is a convex combination of the extreme points of the ball $\{\cB\in\rH_{\bn}: \|\cB\|_{nuc}\le 1\}$.   The last equality follows from the assumption $\tr B=0$.   Hence, $\sum_{l=1}^p a_l=\sum_{m=1}^q b_m=\frac{\|\cB\|_{nuc}}{2}$.   Assume that $\|\cB\|\le r=\frac{2}{N(\bn)\beta_{1,\bn,h}}$.
 Then 
 $$\|\cB\|_{nuc}\le r\beta_{1,\bn,h}=\frac{2}{N(\bn)}\Rightarrow \sum_{l=1}^p a_l=\sum_{m=1}^q b_m\le \frac{1}{N(\bn)}.$$ 
 
 Fix $m\in[q]$.  Complete the vector $\y_{j,m}$ to an orthonormal basis $\bw_{i_j,j,m}, i_j\in[n_j]$ in $\C^{n_j}$ such that $\bw_{1,j,m}=\y_{j,m}$,  for $j\in[d]$ and $m\in[q]$.
 Observe that $I_{n_j}=\sum_{i_j=1}^q \bw_{i_j,j,m}\bw_{i_j,j,m}^\dagger$.
 Hence,
 \begin{equation*}
 \begin{aligned}
 &\otimes_{j=1}^{n_j}I_{n_j}=\sum_{[i_1,\ldots,i_d]\in [\bn]}(\otimes_{j=1}^d\bw_{i_j,j,m})\otimes (\otimes_{j=1}^d\bar\bw_{i_j,j,m})\Rightarrow\\
& \frac{1}{N(\bn)-1}(\otimes_{j=1}^d I_{n_j}-(\otimes_{j=1}^{d}\bw_{1,j,m})\otimes(\otimes_{j=1}^{d}\bar\bw_{1,j,m}))\in \mathrm{Sep}_{\bn}.
\end{aligned}
 \end{equation*}
 Thus,
 \begin{equation*}
 \begin{aligned}
 &\frac{1}{N(\bn)}\otimes_{j=1}^d I_{n_j}+\cB=\sum_{l=1}^p a_l(\otimes_{j=1}^d \x_{j,l})\otimes(\otimes_{j=1}^d \bar\x_{j,l})+(\frac{1}{N(\bn)}-\sum_{m=1}^q b_m)\otimes_{j=1}^{n_j}I_{n_j}+\\
&\sum_{m=1}^q b_m(\otimes_{j=1}^d I_{n_j}-(\otimes_{j=1}^{d}\bw_{1,j,m})\otimes(\otimes_{j=1}^{d}\bar\bw_{1,j,m}))\Rightarrow\\
&\frac{1}{N(\bn)}\otimes_{j=1}^d I_{n_j}+\cB\in \mathrm{Sep}_{\bn}.
 \end{aligned}
 \end{equation*}
 \end{proof}
 For $d=2$ it is shown in \cite{GB02} that $r= \frac{1}{N(\bn)}$ is the maximal $r$ such that $\rB_{\|\cdot\|}(\frac{1}{N(\bn)} \otimes_{j=1}^d I_{n_j}),r)\cap\rH_{\bn,+,1}\subset \mathrm{Sep}_{\bn}$.

\subsection{Bi-symmetric Hermitian tensors}\label{subsec:bisymt}
\begin{definition}\label{defbisymher}
Assume that $\bn=n^{\times d}$ and $n,d\ge 2$.  
Denote $\bi=(i_1,\ldots,i_d)\in[n^{\times d}]$.  For  a permutation $\omega:[d]\to[d]$ denote $\omega(\bi)=(i_{\omega(1)},\ldots,i_{\omega(d)})$.  
A Hermitian tensor $\cT\in \rH_{n^{\times d}}$ is called bi-symmetric Hermitian if for any permutation $\omega\in\Omega_d$ one has the equality
\begin{equation}\label{bisymdef}
 t_{\omega(\bi), \bj}=t_{\bi,\bj} \textrm{ for all  }\bi,\bj\in[n^{\times d}], \omega\in\Omega_d\quad \cT=[t_{\bi,\bj}]\in \rH_{n^{\times d}}. 
\end{equation}
Denote by $\rB\rH_{n^{\times d}} \subset \rH_{n^{\times d}}$ the subspace of bi-symmetric Hermitian tensors.  
\end{definition}
Observe that  if $\cT$ is bi-symmetric Hermitian it follows that for permutations $\omega$ and $\psi$ we have the equalities $t_{\omega(\bi)\psi(\bj)}=t_{\bi\bj}$. 
The following lemma is known in physics literature \cite{M-S25}:
\begin{lemma}\label{spdcbsh}  Let $2\le d$ 
and $0\ne\cT\in\rB\rH_{n^{\times d}}$.  Denote by $\Lambda\subset [n^d]$ the set of nonzero eigenvalues of $\cT$ given by \eqref{specdec}.
Then,  $\cX_k\in\rS^d\C^n$ for $k\in\Lambda$. Thus, $\mathrm{range}\,\cT\subset \rS^d\C^n$.  
If in addition $\cT$ has real entries then $\cX_k\in \rS^d\R^n$ for $k\in\Lambda$.
Hence, 
\begin{equation}\label{dimHnd}
\begin{aligned}
&\dim \rB\rH_{n^{\times d}} ={n+d-1\choose d}^2,\\
&\dim \Re\rB\rH_{n^{\times d}} ={n+d-1\choose d}\left({n+d-1\choose d}+1\right)/2,
\end{aligned}
\end{equation}
where $\Re\rB\rH_{n^{\times d}}\subset \rB\rH_{n^{\times d}} $ is the set of real valued bi-symmetric Hermitian tensors.

\end{lemma}
\begin{proof}
Assume that $k\in\Lambda$.  Then $\cX_k=[x_{\bi,k}], \bi\in[n^{\times d}]$ is an eigenvector of $\cT$ of length one corresponding to the eigenvalue $\lambda_k(\cT)\ne 0$:
\begin{equation*}
\begin{aligned}
&\sum_{\bj\in[n^{\times d}]}t_{\bi,\bj}x_{\bj,k}=\lambda_k(\cT)x_{\bi,k}, \quad \bi\in[n^{\times d}]\Rightarrow\\
&\lambda_k(\cT)x_{\bi,k}=\sum_{\bj\in[n^{\times d}]}t_{\omega(\bi),\bj}x_{\bj,k}=\lambda_k(\cT)x_{\omega(\bi),k}\Rightarrow\\
&x_{\bi,k}=x_{\omega(\bi),k} \textrm{ for all } \bi\in[n^{\times d}], \omega\in\Omega_d, k\in\Lambda.
\end{aligned}
\end{equation*}
For $\cT\in \Re\rB\rH_{n^{\times d}}$, use the real spectral decomposition of $\cT$ as a real symmetric matrix $T\in \rS^2\R^{N(n^{\times d})}$.

As each $\cX_i\in\rS^d\F^n$ in the decomposition \eqref{specdec} we deduce $\mathrm{range}\,\rT\subset \rS^d\F^n$.
Recall that $\dim\rS^d\F^n={n+d-1\choose d}$ \cite{FW20}.  Hence,  \eqref{dimHnd} holds.
\end{proof}
\begin{corollary}\label{bodt}
The set of positive semidefinite bi-symmetric Hermitian tensors with trace one, denoted by $\rB\rH_{n^{\times d},+,1}$, correspond to the convex set of bi-symmetric density tensors generated by Bosons on $\rS^d\C^n$.
\end{corollary}
\begin{definition} \label{defstsep}
A bi-symmetric Hermitian density tensor  $\cR\in\rB\rH_{n^{\times d}}$ is called strongly separable if $\cR$ is a convex combinations of $\x^{\otimes d}\otimes  \bar\x^{\otimes d} ,\x\in\C^n, \|\x\|=1$.   Denote by $\mathrm{Seps}_{n^{\times d}}$ the set of strongly separable bi-symmetric Hermitian density tensors in $\rB\rH_{n^{\times d}}$. 
\end{definition}
The following Lemma was observed by O.  G\"{u}hne,  see the arguments in \cite{TG09}:
\begin{lemma}\label{bisymsep}  A bi-symmetric Hermitian density tensor  $\cR\in\rB\rH_{n^{\times d}}$ is strongly separable if and only if it is separable.
\end{lemma}
\begin{proof}  Let  $\rP_s:\otimes^d\C^n\to \rS^d\C^n$ be the orthogonal projection.  Note that $\|\rP_s(\cX)\|\le |\cX\|$, and equality holds if and only $\cX\in\rS^d\C^n$.  In particular, if $\cX\in \Pi_{n^{\times d}}$ then $\cX\in\rS^d\C^n$ if and only if $\cX=\x^{\otimes d}, \x\in\C^n,\|\x\|=1\}$.  

Assume that $\cR\in\rB\rH_{n^{\times d}}$ is separable.   Then $\cR=\sum_{k=1}^r a_i \cX_i\otimes \overline{\cX_i}, \sum_{k=1}^r a_k=1$, where $a_i>0,\cX_i\in \Pi_{n^{\times d}}$ for $i\in[r]$.  Observe that 
\begin{equation*}
\begin{aligned}
&\cR=\sum_{k=1}^r a_i \cX_i\otimes \overline{\cX_i}=(P_s\otimes P_s)(\cR)=\sum_{k=1}^r a_i \rP_s(\cX_i)\otimes \overline{\rP_s(\cX_i})\Rightarrow\\
&1=\tr \cR=\sum_{k=1}^r a_i \|\cX_i\|^2=\sum_{k=1}^r a_i \|\rP_s(\cX_i)\|^2\Rightarrow\\
&\cX_i\in\rS^d\C^n\rightarrow \cX_i=\x_i^{\otimes d}, \|\x_i\|=1, i\in [r]\Rightarrow \cR \textrm{ is strongly separable}.
\end{aligned}
\end{equation*}
\end{proof}

\begin{lemma}\label{bisymspecnuclem} Assume that  $2\le n,d\in\N$.   Let 
\begin{equation*}
\cI_{{n+d-1\choose d}}=\sum_{k=1}^{{n+d-1\choose d}} \cX_k\otimes \bar\cX_k, \quad
\cX_k\in\rS^d\C^n, \langle \cX_k,\cX_j\rangle =\delta_{k,j}, k,j\in[{n+d-1\choose d}]
\end{equation*}
be the identity tensor in $\cB\in\rB\rH_{n^{\times d}}$ acting on $\rS^d\C^n$.  Assume that $\cB\in\rB\rH_{n^{\times d}}$.   
\begin{enumerate}
\item  The following equality holds 
\begin{equation}\label{B1nrmchar}
\|\cB\|_{1,\C}=\max_{\x,\y\in\C^n, \|\x\|=\|\y\|=1}|\langle\x^{\otimes d},\cB\y^{\otimes d}\rangle|=\max_{\x,\y\in\C^n, \|\x\|=\|\y\|=1}\Re\langle\x^{\otimes d},\cB\y^{\otimes d}\rangle.
\end{equation}
\item 
Define
\begin{equation}\label{defbspecnrm}
\|\cB\|_{bspec}:=\max_{\x\in\C^n,\|\x\|=1}|\langle \x^{\otimes d},\cB\x^{\otimes d}\rangle|=\max_{\x\in\C^n,\|\x\|=1}\max(\langle \x^{\otimes d},\cB\x^{\otimes d}\rangle,
-\langle \x^{\otimes d},\cB\x^{\otimes d}\rangle).
\end{equation}
Then $\|\cB\|_{bspec}$ is a norm on $\rB\rH_{n^{\times d}}$, called b-spectral norm.  
\item
\begin{equation}\label{bspecin}
\|\cB\|_{bspec}\le \|\cB\|_{spec}\le \|\cB\|_{1,\C}.
\end{equation}
Equalities hold if $\cB$ is a positive semi-definite or negative semi-definite.
\item Let $\|\cB\|_{bnuc}$ be the dual norm of $\|\cdot\|_{bspec}$, called b-nuclear norm:
\begin{equation}\label{nucnrm}
\|\cB\|_{bnuc}=\max_{\cC\in\rB\rH_{n^{\times d}},\|\cC\|_{bspec}=1} \langle\cC,\cB\rangle.
\end{equation}
The extreme points if $\rB_{\|\cdot\|_{bnuc}}(0,1)$ are $\pm\x^{\otimes d}\otimes  \bar\x^{\otimes^d},\x\in\C^n, \|\x\|=1$.  Hence,
\begin{equation}\label{charnucb}
\|\cB\|_{bnuc}=\min_{\cX=\sum_{i=1}^r a_i (\otimes^d\x_i)\otimes (\otimes^d\bar\x_i,) a_i\in\R, \x_i\in\Pi_{\bn}, i\in[r], r\in\N}\sum_{i=1}^r |a_i|.
\end{equation}
Furthermore, the minimum is achieved for some $r\le {n+d-1\choose d}^2+1$.
\item  The following inequalities hold:
\begin{equation}\label{bnucin}
\|\cB\|_{bnuc}\ge\|\cB\|_{nuc}.
\end{equation}
\item Denote
\begin{equation}\label{defalbeinbh}
\begin{aligned}
&\alpha_{\infty,n^{\times d},bh}=\min_{\cX\in\rB\rH_{n^{\times d}},\|\cX\|=1} \|\cX\|_{bspec}, \quad \beta_{\infty,n^{\times d},bh}=\max_{\cX\in\rB\rH_{n^{\times d}},\|\cX\|=1} \|\cX\|_{bspec},\\
&\alpha_{1,n^{\times d},bh}=\min_{\cX\in\rB\rH_{n^{\times d}},\|\cX\|=1} \|\cX\|_{bnuc}, \quad \beta_{1,n^{\times d},bh}=\max_{\cX\in\rB\rH_{n^{\times d}},\|\cX\|=1} \|\cX\|_{bnuc}.
\end{aligned}
\end{equation}
Then 
\begin{equation}\label{alpbethinbh}
\begin{aligned}
&\alpha_{\infty,n^{\times d},bh}< \beta_{\infty,n^{\times d},bh}=1, \quad \alpha_{1,n^{\times d},bh}=1<\beta_{1,n^{\times d},bh},\\
&\alpha_{\infty,n^{\times d},bh}\beta_{1,n^{\times d},bh}=1.
\end{aligned}
\end{equation}
\item The group $\rU_{n}$ acts on $\rB\rH_{n^{\times d}}$ via the diagonal action $\cR\mapsto (\otimes^d U)\cR \overline{(\otimes^d U)}$:
\begin{equation*}
\begin{aligned}
&\cX\otimes \overline{\cX}\mapsto((\otimes^d U)\cX)\otimes\overline{((\otimes^d U)\cX)},\\
&\x^{\otimes d}\otimes \bar\x^{\otimes d}\mapsto (U\x)^{\otimes d}\otimes (\overline{U\x})^{\otimes d}.
\end{aligned}
\end{equation*}
as a group of isometries of the norms $\|\cdot\|_{bspec}$ and $\|\cdot\|_{bnuc}$.
Let $\mu_{n}$ be the Haar measure on $\rU_{n}$.   Then for any $\cR\in\rB\rH_{n^{\times d},+,1}$ the following equality holds:
 \begin{equation}\label{Haareqs}
 \int_{\rU_{n}} ((\otimes^d U)\cR \overline{(\otimes^d U)}d\mu_{n}=\frac{1}{{n+d-1\choose d}}\cI_{{n+d-1\choose d}}.
 \end{equation}
\end{enumerate}
\end{lemma}
\begin{proof}\emph{(1)}  Assume that $\cT$ has the spectral decomposition \eqref{specdec}.  Fix $\cY=\otimes_{j=1}^d\y_j\in\Pi_{n^{\times d}}$ and  observe that  $\cT\cY=\sum_{k\in\Lambda}\lambda_k(\cT)\langle \cX_i,\cY\rangle \cX_i\in\rS^d\C^n$.  Banach's theorem yields that 
$$\max_{\cX\in\Pi_{n^{\times d}}}|\langle \cX, \cT\cY|=\max_{\x\in\C^n,\|\x\|=1}|\langle \x^{\otimes d}, \cT\cY\rangle|.$$
Use the equality $\langle \x^{\otimes d},\cT\cY\rangle=\langle \cT\x^{\otimes d},\cY\rangle$, and Banach's theorem again to deduce
\begin{equation*}
\|\cB\|_{\infty,\C}= 
\max_{\cX,\cY\in\Pi_{n^{\times d}}}|\langle \cX,\cT\cY\rangle|=\max_{\x,\y\in\C^n,\|\x\|=\|\y\|=1}|\langle \x^{\otimes d}, \cT\y^{\otimes d}\rangle|.
\end{equation*}

\emph{(2)}  It is enough to show that $\|\cT\|_{bspec}=0\Rightarrow \cT=0$.   
Assume to the contrary that $\cT\ne 0$ and $\|\cT\|_{bspec}=0$,  where $\cT$ has the expansion  \eqref{specdec}.  We view $\cT$ as a linear operator $T:\otimes^d\C^n\to \otimes^d\C^n$.
As $\cT$ is bi-symmetric Hermitian, we need to show that 
\begin{equation}\label{bspecimp}
\langle \x^{\otimes d},\cT  \x^{\otimes d}\rangle=0 \textrm{ for all }\x\in\C^n\Rightarrow \cT=0.
\end{equation}
Observe that since $\cT$ is bi-symmetric Hermitian follows that 
\begin{equation}\label{bspeceq}
\begin{aligned}
&\langle \otimes_{j=1}^d \x_j,  \cT\otimes_{j=1}^d \y_j\rangle=\langle \otimes_{j=1}^d \x_{\omega(j)},  \cT\otimes_{j=1}^d \y_{\psi(j)}\rangle=\overline{\langle \otimes_{j=1}^d \y_{\psi(j)},  \cT\otimes_{j=1}^d x_{\omega_(j)}}\rangle, \omega,\psi\in\Omega_d,\\
&\langle \otimes_{j=1}^d \x_j,  \cT\otimes_{j=1}^d \x_j\rangle\in\R.
\end{aligned}
\end{equation}
(The second equality in \eqref{defbspecnrm} follows from the observation that $\langle \x^{\otimes d},\cB\x^{\otimes d}\rangle\in\R$.)
Assume the left hand side of \eqref{bspecimp}.  Replace  $\x$ by $\x+t\x_1$ where $\x,\x_1$ are fixed  and $t\in\R$.   Hence the coefficient in front of $t$ is zero:
\begin{equation*}
2d\Re\langle\x^{\otimes (d-1)}\otimes\x_1, \cT\x^{\otimes d}\rangle=0.
\end{equation*}
Replace $\x_1$ by $z\x_1$ where $|z|=1$ to deduce that 
\begin{equation*}
\langle\x^{\otimes (d-1)}\otimes\x_1, \cT\x^{\otimes d}\rangle=0 \textrm{ for all } \x,\x_1\in\C^n.
\end{equation*}
Substitute $\x$ by $\x+t\x_2$ for $\x,\x_1,\x_2$ fixed and $t\in\R$.  As the coefficient of $t$ is zero we obtain the equality
\begin{equation*}
\begin{aligned}
&(d-1)\langle\x^{\otimes (d-2)}\otimes \x_2\otimes\x_1, \cT x^{\otimes d}\rangle+d
\langle\x^{\otimes (d-1)}\otimes\x_1 ,\cT\x^{\otimes (d-1)}\otimes\x_2\rangle=0,\\ &\textrm{ for all } \x,\x_1,\x_2\in\C^n .
\end{aligned}
\end{equation*}
Replace $\x_2$ by $z\x_2$, where $|z|=1$ to obtain the equality 
\begin{equation*}
\begin{aligned}
&(d-1)\langle\x^{\otimes (d-2)}\otimes \x_2\otimes\x_1, \cT\x^{\otimes d}\rangle+d
z^2\langle\x^{\otimes (d-1)}\otimes\x_1 ,\cT\x^{\otimes (d-1)}\otimes\x_2\rangle=0,\\ &\textrm{ for all } \x,\x_1,\x_2\in\C^n,|z|=1.
\end{aligned}
\end{equation*}
Hence $\langle\x^{\otimes (d-2)}\otimes \x_2\otimes\x_1, \cT\x^{\otimes d}\rangle=0$.
Repeat this argument untill we deduce the equality $\langle\x_{d}\otimes \x_{d-1}\otimes\cdots\x_1, \cT\x^{\otimes d}\rangle=0$.   Since $\x_1,\ldots,\x_d$ are arbitrary we deduce that $T\x^{\otimes d}=0$ for all $\x\in\C^n$. Let $x=\sum_{j=1}^d t_j\x_j$, here $\x_1,\ldots,\x_j$ are fixed and $t_1,\ldots,t_d\in\R$.
Hence the coefficient of $t_1\cdots t_d$ is zero.    It now follows that in the spectral  decomposition of $\cT$ \eqref{specdec} $\cX_i=0$ for $i\in\Lambda$ contrary to our assumption that $\cT\ne 0$.

\noindent \emph{(3)}  The inequalities \eqref{bspecin} are straightforward.   Assume now that $\cT$ positive semidefinite, i.e.,
$\lambda_k(\cT)>0$ for $k\in\Lambda$.  Set $\cT_1=\sum_{k\in\Lambda}\sqrt{\lambda_k(\cT)}\cX_i\otimes \overline{\cX_i}$.  Then $\langle \cX,\cT\cY\rangle=
\langle\cT_1\cX,\cT_1\cY\rangle$. Use Cauchy-Schwarz inequality for $\x^{\otimes d},\y^{\otimes d}, \|\x\|=\|\y\|=1$ to deduce
$$|\langle \x^{\otimes d},\cT\y^{\otimes d}\rangle|\le \sqrt{\langle \x^{\otimes d},\cT\x^{\otimes d}\rangle
\langle \y^{\otimes d},\cT\y^{\otimes d}\rangle}\le \|\cT\|_{bspec}.$$
Combine this inequality with \eqref{B1nrmchar} to deduce the equalities in \eqref{bspecin} for positive semidefinite $\cB$.
Repalce $\cB$ with $-\cB$ to deduce the equality in \eqref{bspecin} for negative semidefinite $\cB$.

\noindent \emph{(4-6)}  are deduced as in the proof of \emph{(2-4)} in Lemma \ref{specnuclem}.

\noindent\emph{(7)}    Clearly,  every $U\in \rU_n$ maps the positive and negative definite extreme points of the ball $\|\cX\|_{bnuc}\le 1$ onto themselves.  Hence,  
\begin{equation*}
\begin{aligned}
&(\otimes^{d} U)\rH_{n^{\times d}}\overline{(\otimes^{d} U)}=\rH_{n^{\times d}},\quad
(\otimes^{d} U)\rH_{n^{\times d},+}\overline{(\otimes^{d} U)}=\rH_{n^{\times d},+}, \\&(\otimes^{d} U)\rH_{n^{\times d},+,1}\overline{(\otimes^{d} U)}=\rH_{n^{\times d},+,1},\quad (\otimes^{d} U)\mathrm{Seps}_{n^{\times d}}\overline{(\otimes^{d} U)}=\mathrm{Seps}_{n^{\times d}}
\end{aligned}
\end{equation*}
for $U\in \rU_n$.
Part \emph{(6)} of Lemma \ref{specnuclem} that the action of $\rU_n$ on $\rB\rH_{n^{\times d}}$ preserves the trace.
Let 
\begin{equation*}
\cQ_{\x}=\int_{\rU_n} (U\x)^{\otimes d}\otimes \overline{(U\x)^{\otimes d}}d\mu_n,\quad \|\x\|=1.
\end{equation*}
As $\mu_n$ is the Haar measure, it follows that $\cQ_{\x}=\cQ_{\y}$ for any $\|\y\|=1$.
Thus $\cQ_{\x}=\cQ$, and $\tr \cQ=1$.  
Furthermore,  $U^{\otimes d}\cQ\overline{U^{\otimes d}}=\cQ$ for any $U\in \rU_n$.
As every $\cB\in\rB\rH_{n^{\times d}}$ is a linear combination of $\pm \x^{\otimes d}\otimes \bar \x^{\otimes d}, \|\x\|=1$ we deduce
\begin{equation}\label{intUniden}
\int_{\rU_n} U^{\otimes d}\otimes \cB\otimes \overline{U^{\otimes d}}d\mu_n =(\tr \cB)\cQ, \quad \cB\in\rB\rH_{n^{\times d}}.
\end{equation}
We next observe that $\langle \cX,\cY\rangle=\langle (\otimes^d U)\cX,(\otimes^d U)\cY\rangle$ for $\cX,\cY\in \rS^d\C^n, U\in\rU_n$.   Hence,  if $\cX_k, k\in[{n+d-1\choose d}]$ is an orthonormal basis in $\rS^d\C^n$, so is $(\otimes^d U(\cX_k, k\in[{n+d-1\choose d}]$ for any $U\in \rU_n$.   Therefore,
 $(\otimes ^d U)\cI_{n+d-1\choose d}\overline{(\otimes ^d U)}=\cI_{n+d-1\choose d}$.  
 Let $\cB=cI_{n+d-1\choose d}$ in \eqref{intUniden} to deduce that $\cQ=\frac{1}{{n+d-1\choose d}}\cI_{{n+d-1\choose d}}$.
 This proves \eqref{Haareqs}.
\end{proof}
Note that \eqref{Haareqs} follows from Schur-Weyl duality \cite[Proof of Lemma 1]{EW01}.
The following theorem is an analog of Theorem \ref{charsep}:
\begin{theorem}\label{charseps}  Let $2\le d, n\in\N$.  
Then
\begin{enumerate}
\item  A bi-symmetric density tensor $\cR\in\rB\rH_{n^{\times d},+,1}$, with spectral decomposition as in Lemma \ref{spdcbsh}, satisfies the inequalities
\begin{equation}\label{lubsndtb}
\begin{aligned}
&\frac{1}{{n+d-1\choose d}}\le \|\cR\|_{bspec}\le \lambda_1(\cR)\le 1,\\
&1\le \|\cR\|_{bnuc}\le \beta'_{1,n^{\times d},bh}:=\max_{\cX\in\rS^d\C^n,\|\cX\|=1}\|\cX\otimes \overline{\cX}\|_{bnuc}.
\end{aligned}
\end{equation}
Equalities in some of the above inequalities hold under the following conditions:
\begin{equation}\label{lubsndteqb}
\begin{aligned}
&\|\cR\|_{bspec}=\lambda_1(\cR)\iff \exists \x\in\C^n, \|\x\|=1 \textrm{ such that } \langle \x^{\otimes d},\cX_k\rangle = 0 \textrm{ if } \lambda_k(\cR) <\lambda_1(\cR),\\
&\|\cR\|_{bspec}=1\iff \cR=\x^{\otimes d}\otimes \bar\x^{\otimes d} \|\x\|=1,\\
&\frac{1}{{n+d-1\choose d}}=\|\cR\|_{bspec} \iff \cR=\frac{1}{{n+d-1\choose d}}\cI_{{n+d-1\choose d}},\\
&1=\|\cR\|_{bnuc} \iff \cR\in\mathrm{Seps}_{n^{\times d}},\\
& \|\cR\|_{bnuc}=\beta'_{1,n^{\times d},bh} \Leftarrow \cR=\cX\otimes \overline{\cX}, \cX\in\rS^d\C^n, \|\cX\|=1, \|\cX\otimes \overline{\cX}\|_{bnuc}=\beta'_{1,n^{\times d},bh}.
\end{aligned}
\end{equation}
\item The set $\mathrm{Seps}_{n^{\times d}}$ is a face of the unit ball of $\|\cdot\|_{bnuc}$ of the maximal dimension ${n+d-1\choose d}^2-1$.  Its supporting hyperplane is
\begin{equation}\label{suphypseps}
 \tr \cB=\langle \cI_{n+d-1\choose d},\cB\rangle\le 1 \textrm{ for } \|\cB\|_{bnuc}\le 1.
 \end{equation}
 Equality holds if and only if $\cB$ is strongly separable.
In particular, any strongly separable density tensor is a convex combinations of at most ${n+d-1\choose d}^2$ rank-one bi-symmetric density tensors.  
\end{enumerate}
\end{theorem}
\begin{proof} The proof this theorem is similar to the proof of Theorem \ref{charsep} and we point out  the needed modifications.
\emph{(1)}  Combine the inequality \eqref{bspecin}  with $\|\cR\|_{spec}\le \lambda_1(\cR)\le 1$ to deduce the last two inequalities in \eqref{lubsndtb}.  The conditions for the equality cases $\|\cR\|_{bspec}=\lambda_1(\cR)$ and $\|\cR\|_{bspec}=1$ are deduces as in the proof of Theorem \ref{charsep}.

The inequality $\frac{1}{{n+d-1\choose d}}\le \|\cR\|_{spec}$ follows from \eqref{Haareqs}.   The equality case as in the proof of Theorem \ref{charsep}.

The inequalities  \eqref{bnucin} and \eqref{lubsndt} yield that $\|\cR\|_{bnuc}\ge \|\cR\|_{nuc}\ge 1$.    Suppose that $\cR$ is strongly separable.
It is a convex combination of the exteme points of $\rB_{\|\cdot\|_{bnuc}}(0,1)$ of the form $\x^{\otimes d}\otimes\bar\x^{\otimes d}$.   Hence, $\|\cR\|_{bnuc}\le 1$.  Therefore,
$\|\cR\|_{bnuc}=1$.  Assume $\|\cR\|_{bnuc}=1$.  As $\tr \cR=1$ we deduce that a convex combinations of the extreme points of $\rB_{\|\cdot\|_{bnuc}}(0,1)$ that gives $\cR$ must be only of the form $\x^{\otimes d}\otimes\bar\x^{\otimes d}$. Hence,  $\cR$ is strongly separable.   

\noindent\emph{(2)}  Clearly,  $\mathrm{Seps}_{n^{\times d}}$ is a face of the unit ball of $\|\cdot\|_{bnuc}$.  It is left to show that it has the maximal dimension  ${n+d-1\choose d}^2-1$.   Observe that the dimension of the subspace spanned by $\frac{1}{2}(\cR_1-\cR_2)$, where $\cR_1,\cR_2$ are strongly separable density tensors, is the dimension of the face $\mathrm{Seps}_{n^{\times d}}$.  Consider now the intersection of the unit ball of $\|\cdot\|_{bnuc}$ with the subspace $\tr \cB=0$.  Clearly, the dimension of this convex set is ${n+d-1\choose d}^2-1$.  Take $\cT\ne 0$ in this convex set.  It is a convex combination of the extreme points of $\|\cB\|_{bnuc}\le 1$.
Clearly, it must be of the form $(\cR_1-\cR_2)/2$. This proves that the dimension of  $\mathrm{Seps}_{n^{\times d}}$ is  ${n+d-1\choose d}^2-1$.   The rest of the proof is as in the proof of Theorem \ref{charsep}.
\end{proof}
\begin{corollary}\label{charsepscb}  Denote by $\mathrm{Seps}_{n^{\times  d},int}$ the relative interior of $\mathrm{Seps}_{n^{\times d}}$.  
\begin{enumerate}
\item The sets $\mathrm{Seps}_{n^{\times d}}$ and $\mathrm{Seps}_{n^{\times d},int}$ are invariant under the action of $\rU_{n}$.
\item Let $\cR\in \mathrm{Seps}_{n^{\times d},int}$.  There exists $r(\cR)>0$ such that $\rB_{\|\cdot\|}(\cR,r(\cR))\cap\rB\rH_{n^{\times d},+,1}\subset \mathrm{Sep}_{\bn,int}$.
\item The density tensor $ \frac{1}{{n+d-1\choose d}}\cI_{{n+d-1\choose d}}$ is in $\mathrm{Seps}_{n^{\times d},int}$, where 

\noindent
$r(\frac{1}{{n+d-1\choose d}}I_{{n+d-1\choose d}})\ge \frac{2}{{n+d-1\choose d}\beta_{1,n^{\times d},bh}}$.
\item The inequalities in \eqref{lualbet'in} hold.
\end{enumerate} 
\end{corollary}
\begin{proof} The proofs of \emph{(1)-(3)} an in the proof of Corollary \ref{charsepsc}.

\noindent\emph{(4)} Assume that $\cX\in \rS^d\C^n, \|\cX\|=1$.  Then $\cR=\cX\otimes\overline{\cX}\in \rB\rH_{n^{\times d},+,1}$.   The first equality in \eqref{tenprodeq} implies that $\|\cR\|_{\infty,\C}=\|\cX\|_{\infty,\C}^2$.  
Use part \emph{(3)} of Lemma \ref{bisymspecnuclem}  to deduce $\|\cR\|_{\infty,\C}= \|\cR\|_{bspec}\ge { n+d-1\choose d}^{-1}$.  Hence,
 $\|\cX\|_{\infty,\C}\ge {n+d-1\choose d}^{-1/2}$.  Hence,  $\alpha_{\infty,n^{\times d},s,\C}\ge {n+d-1\choose d}^{-1/2}$.  Use the equality in part \emph{(1)} of Theorem \ref{relalphbetsym} to deduce $\beta_{1,n^{\times ^d},s,\C}\le {n+d-1\choose d}^{1/2}$.
\end{proof}
\subsection{Symmetric Hermitian tensors}\label{subsec;symher} 
A bi-symmetric Hermitian tensor in $\cR\in\rB\rH_{n^{\times d}}$ is called \emph{symmetric
} if $\cR\in \rS^{2d}\C^n$.  In \cite{QC19,CCQS,Ha21} $\cR$ is called completely symmetric.    In view of Definition \ref{defherm}  it follows that $\cR\in\rS^{2d}\R^n$. 
Denote by symmetric positive semi-definite and symmetric density tensors the sets 
$\rS^{2d}\R^n_{+}:=\rS^{2d}\R^{n}\cap \rH_{n^{\times d},+}$ and $\rS^{2d}\R^n_{+,1}:=\rS^{2d}\R^{n}\cap \rH_{n^{\times d},+,1}$ respectively.
\begin{definition}\label{ssep}   A symmetric density tensor in $\rS^{2d}\R^n_{+,1}$ is called a real strongly separable if it is a convex combination of $\{\x^ {\otimes (2d)}: \x\in\R^n, \|\x\|=1\}$.   
\end{definition}
It is shown in \cite{CCQS} that a separable symmetric density tensors, viewed as  $\cR\in \rH_{n^{\times d},+,1}$ is  a real strongly separable.   (One can also prove this result as in the proof of Lemma \ref{bisymsep}.)
We now give a different proof to  \cite[Theorem 14]{CCQS}.    We need the following lemma:
\begin{lemma}\label{snnrsym}  Let $2\le d,n\in\N$.  Then  for $\rS\in\rS^{2d}\R^n$ the following equalities hold:
\begin{equation}\label{normeqsym}
\|\cS\|_{spec}=\|\cS\|_{bspec}=\|\cS\|_{\infty,\R}, \quad \|\cS\|_{nuc}=\|\cS\|_{bnuc}=\|\cS\|_{1,\R}. 
\end{equation}
\end{lemma}
\begin{proof} Banach's theorem states that 
\begin{equation*}
\|\rS\|_{\infty,\R}=\max_{\x\in\R^n,\|\x\|=1}\langle \x^{\otimes (2d)},\cS\rangle=\langle  \x^{\otimes d},\cS\x^{\otimes d}\rangle.
\end{equation*}
Compare that with \eqref{defbspecnrm} to deduce the inequality $\|\cS\|_{bspec}\ge \|\cS\|_{\infty,\R}$.   Recall the inequality \eqref{bspecin} to deduce the inequalities 
$\|\cS\|_{spec}\ge\|\cS\|_{bspec}\ge \|\cS\|_{\infty,\R}$.

We now show the inequality $\|\cS\|_{\infty,\R}\ge \|\cS\|_{spec}$.
Let $\cS=[s_{i_1,\ldots,i_d,j_1,\ldots,j_d)}]$.   
Assume that $\|\cS\|_{spec}=|\langle \cX,\rS\cX\rangle|$ for some $\cX=\otimes_{j=1}^d \x_i\in \Pi_{n^{\times d},\C}$.  We now show that we can find $\cY=\otimes_{j=1}^d \y_i\in \Pi_{n^{\times d},\R}$ such that $\|\cS\|_{spec}=|\langle \cY,\rS\cY\rangle|$ 
Fix $i_d,j_d\in [n]$ and let $\cS_{i_d,j_d}=[s_{i_1,\ldots,i_k,\ldots,i_d,j_1,\ldots,i_k,\ldots,j_d)}],k\in[d-1]$ be in $\S^{2(d-1)}\R^n$.   View $\cS_{i_d,j_d}$ as bi-symmetric Hermitian tensor in  $\rH_{n^{\times (d-1)}}$.  Consider the $(d-1)$-sesquilinear form $c_{i_d,j_d}=\langle \cX_d,\cS_{i_d,j_d} \cX_d\rangle$, where $\cX_d=\otimes_{j=1}^{d-1}\x_j\in\Pi_{n^{\times (d-1)},\C}$.  Let $C=[c_{i_d,j_d}]\in \R^{n\times n}$, and observe that $C$ is a symmetric matrix.  Hence, the spectral norm of $C$ is 
\begin{equation*}
\begin{aligned}
&\|C\|_{\infty}=\max_{\y\in\C^n, } |\langle \y ,C\y\rangle|=\max_{\y\in\R^n, } |\langle \y ,C\y\rangle|=|\langle \y_d, C\y_d\rangle|=\
\|\cS\|_{spec}=\\
&|\langle \cX_{d-1}\otimes \y_d, \cS(\cX_{d-1}\otimes \y_d\rangle|=|\langle  \y_d\otimes\cX_d, \cS(\y_d\otimes\rangle\cX_{d-1})|.
\end{aligned}
\end{equation*}
Continue in ths manner to deduce that $\|\cS\|_{spec}=|\langle \cY_{d}, \cS\cY_{d}\rangle|$  for some in $\Pi_{n^{\times d},\R}$.  Hence $\|\cS\|_{\infty,\R}\ge \|\cS\|_{spec}$.   This proves the equality  first set of equalities in \eqref{normeqsym}.

We now show the second set of equalities.  Observe the following relations between the extreme points of the unit balls  
$\|\cS\|_{1,\R}\le 1, \|\cC\|_{bnuc}\le 1, \|\cB\|_{nuc}\le 1$:
\begin{equation*}
\{\pm\x^{\otimes (2d)}: \x\in \R^n, \|\x\|=1\} \subset \{\pm\x^{\otimes d}\otimes \bar\x^{\otimes d}: \x\in \C^n, \|\x\|=1\}\subset\{\pm\cX\otimes \overline{\cX}: \cX\in\Pi_{n^{\times d},\C}\}.
\end{equation*}
Recall the minimum characterizations of the nuclear norms: \eqref{FLBthm},  \eqref{charnucb},  and \eqref{charnuc}.  Hence for $\cS\in\rS^{2d}\R^n$ we have the inequalities $\|\cS\|_{1,\R}\ge \|\cS\|_{bnuc}\ge \|\cS\|_{nuc}$.
Recall the maximum characterization of the dual norm \eqref{defnustar}.   Assume that $\cS\in \rS^{2d}\R^{2n}$.   As $\|\cdot\|_{bnuc}$ agrees with the norm $\|\cdot\|_{\infty,\R}$ on $\rS^{2d}\R^n$ we deduce that $\|\cS\|_{1,\R}\le \|\cS\|_{bnuc}$.   
Hence $\|\cS\|_{1,\R}=\|\cS\|_{bnuc}$.  Similar arguments yield that $\|\cS\|_{1,\R}=\|\cS\|_{nuc}$. This proves the second set of equalities in \eqref{normeqsym}.
\end{proof}
\begin{theorem}\label{charsepsym}  Let $2\le d, n\in\N$.  
\begin{enumerate}
\item  A symmetric density tensor $\cR\in\rS^{2d}\R^n_{+,1}$ satisfies:
\begin{equation}\label{lubsndtsym}
\begin{aligned}
&\frac{(d-1)!}{2\prod_{k=0}^{d-1}\big(\frac{n}{2}+k\big)}\le \|\cR\|_{\infty,\R}\le 1,\\
&1\le \|\cR\|_{1,\R}\le\beta_{1,n^{\times (2d)},s,\R}.
\end{aligned}
\end{equation}
Equalities in some of the above inequalities hold under the following conditions:
\begin{equation}\label{lubsndteqsym}
\begin{aligned}
&\|\cR\|_{\infty,\R}=1\iff \cR=\x^{\otimes 2d},\x\in\R^n,\|\x\|=1,\\
&\frac{(d-1)!}{2\prod_{k=0}^{d-1}\big(\frac{n}{2}+k\big)}=\|\cR\|_{\infty,\R} \iff \cR=\cR^\star:=\int_{\x\in\R^n, \|\x\|=1}\x^{\otimes(2d)}d\lambda_n,\\
&\textrm{ where } \lambda_n \textrm{ is the normalized Lebesgue measure}: \lambda_n(\{\x\in\R^n: \|\x\|=1\})=1,\\
&1=\|\cR\|_{1,\R} \iff \cR\in\mathrm{Sep}_{n^{\times d}}\iff \cR \textrm{ a real strongly separable}.\\
\end{aligned}
\end{equation}
\item The set of real strongly separable symmetric density tensors in $\rS^{2d}\R^n_{+,1}$  is a face of the unit ball of $\|\cdot\|_{1,\R}$ in $\rS^{2d}\R^{n}$ of the maximal dimension ${n+2d-1\choose 2d}-1$.  Its supporting hyperplane is
\begin{equation}\label{suphypsepsym}
 \tr \cB\le 1 \textrm{ for } \|\cB\|_{1,\R}\le 1.
 \end{equation}
 Equality holds if and only if $\cB$ is a real strongly separable.
In particular, any real strongly separable density tensor is a convex combinations of at most ${n+2d-1\choose 2d}$ rank-one symmetric density tensors.  
\end{enumerate}
\end{theorem}
\begin{proof} The proof of this theorem is similar to the proof of Theorem \ref{charseps}  and we point out the corresponding modifications.  

\noindent\emph{(1)}
Assume that $\cR\in\rS^{2d}\R^n_{+,1}$.   Recall the equalities \eqref{normeqsym}.  As $\cR\in\rB\rH_{n^{\times 2d},+,1}$ \eqref{lubsndtb} yields that $\|\cR\|_{\infty,\R}=\|\cR\|_{bspec}\le 1$.  
The conditions in \eqref{lubsndteqb} yield that $\cR=\x^{\otimes d}\otimes \bar\x^{\otimes d}$.   As $\cR$ is symmeric we deduce that $\bar\x=\x$, hence $\x\in\R^n$.  This proves the first implication in \eqref{lubsndteqsym}. 

We now show the first inequality in \eqref{lubsndtsym}.  
Let $\rS\rO_n$ the the special orthogonal group acting on $\R^n$.   Hence, $\rS\rO_n$ acts on $\rS^{2d}\R^n$ by the diagonal action: $\x^{\otimes (2d)}\mapsto (O\x)^{\otimes (2d)}$ for $O\in \rS\rO_n$.   We denote this action of $O$ on $\cB\in\rS^{2d}\R^n$ by $O(\cB)$.    Let $\nu_n$ be the Haar measure on $\rS\rO_n$.

Assume that 
\begin{equation*}
\min_{\cR\in \rS^{2d}\R^n_{+,1}}\|\cR\|_{\infty,\R}=\|\cM\|_{\infty,\R}=
\langle \y^{\otimes (2d)}, \cM\rangle,  \y\in\R^n, \|\y\|=1.
\end{equation*}
Hence,  
\begin{equation*}
\min_{\cR\in \rB\rH_{n^{\times d},+,1}}\|\cR\|_{\infty,\R}=\|O(\cM)\|_{\infty,\R}=
\langle (O(\y))^{\otimes (2d)}, O(\cM)\rangle.
\end{equation*}
Let  $\cR^{\star}=\int_{\nu_n}O(\cM)d\nu_n$.  As $\cM\in \rS^{2d}\R^n_{+,1}$ it follows that $\cR^{\star}\in \rS^{2d}\R^n_{+,1}$.
The minimality of $\cM$ and the triangle inequality yields $\min_{\cR\in \rS^{2d}_{+,1}}\|\cR\|_{\infty,\R}=\|\cR^\star\|_{\infty,\R}$.  Let 
\begin{equation*}
\cM=\sum_{k=1}^m a_k\x_i^{\otimes 2d}, \quad a_k\in\R ,\x_k\in\R^n,\|\x_k\|=1,  k\in[m], \quad \|\cM\|_{1,\R}=\sum_{k=1}^m |a_k|, 1=\tr \cM=\sum_{k=1}^m a_k.
\end{equation*}
Then  
\begin{equation*}
\cR^\star=\int_{\rS\rO_n} (O\x)^{\otimes (2d)}d\nu_n(O).
\end{equation*}
Since the action of $\rS\rO_n$ on the  $(n-1)$-sphere  $\mathbb{S}^{n-1}:=\{\x\in\R^{n}: \|\x\|=1\}$ preserves the restriction of the Lebesgue measure on $\mathbb{S}^{n-1}$ we deduce that the above formula for $\cR^{\star}$ is equivalent to $\R^{\star}=\int_{\mathbb{S}^{n-1}}\x^{\otimes (2d)}d\lambda_n$.
Furthermore,  for any $\y\in\R^n, \|\y\|=1$ one has the equality:
\begin{equation*}
\|\cR^\star\|_{\infty,\R}=\langle \y^{\otimes 2d},\cR^\star\rangle=\int_{\x\in\mathbb{S}^{n-1}} \langle \y,\x\rangle^{2d}d\lambda_n.
\end{equation*}
Choose $\y=(1,0,\ldots,0)^\top$ to deduce $\|\cR^\star\|_{\infty,\R}=\int_{\x\in \mathbb{S}^{n-1}} x_1^{2d}d\lambda_n$.  It is well known that the value of this integral is $\frac{\Gamma(n/2)\Gamma(3/2)}{\Gamma((n/2)+d))\Gamma(1/2)}$, where $\Gamma(x)$ is the Gamma function.   See for example \cite{Cho09}, or \cite[Lemma A.1]{FL20}.  Use the identity $\Gamma(x+1)=x\Gamma(x)$ for $x>0$ to deduce the equality $\|\cR^{\star}\|_{\infty,\R}=\frac{(d-1)!}{2\prod_{k=0}^{d-1}\big(\frac{n}{2}+k\big)}$.  This proves the lower bound in the first inequality of \eqref{lubsndtsym}.
The proof that the minimum is achieved only for $\cR^{\star}$ is similar to the proof of Theorem \ref{charseps}.

We now discuss the second set of inequalities in \eqref{lubsndtsym}.  The definition of $\beta_{1,n^{\times (2d)},s,\R}$ \eqref{defalbeinf1s} yields that $\|\cC\|_{\infty,\R}\le \beta_{1,n^{\times (2d)},s,\R}\|\cC\|$ for $\cC\in\rS^{2d}\R^n$.   Recall that $\cC\in \rH_{n^{\times d}}$.    View $\cC$ as a linear transformation from $\rC^{n^{\times d}}$ to itself.  Then $\|\cC\|=\|\overrightarrow{\lambda}(\cC)\|_2\le \|\overrightarrow{\lambda}(\cC)\|_1$.  As $\cR\in \rH_{n^{\times d},+,1}$ we deduce$ \|\overrightarrow{\lambda}(\cR)\|_1=\tr \cR=1$.
This proves the last inequality in the second set of inequalities in \eqref{lubsndtsym}.  

In view of  \eqref{normeqsym} and \eqref{lubsndt} we deduce that $\|\cR\|_{1,\R}\ge 1$.   Recall $\|\cR\|_{1,\R}=1$ if and only if $\cR$ is a separable density tensor in $\rH_{n^{\times d},+,1}$.    In particular,  $\|\cR\|_{1,\R}=1$ if $\cR$  is a real strongly separable density tensor.  Assume now that $\|\cR\|_{1,\R}=1$.  As the extreme points of the ball $\|\cR\|_{1,\R}$ are  $\{\pm \x^{\otimes (2d)}: \x\in\R^n,\|\x\|=1\}$ it follows that 
\begin{equation*}
\begin{aligned}
&\cR=\sum_{k=1}^m a_k\x_k^{\otimes (2d)},  \quad \x_k\in\R^n, \|\x_k\|=1, k\in[m], \\
&1=\|\cR\|_{1,\R}=\sum_{k=1}^m |a_k|=\tr \cR=\sum_{k=1} a_k\Rightarrow a_k\ge 0, k\in[m].
\end{aligned}
\end{equation*}
Hence $\cR$ is a strongly separable density tensor.  This finishes the proof of part \emph{(1)}.

\noindent \emph{(2)}  The proof of this part i similar to the proof of part \emph{(2)} of Theorem \ref{charseps}.
\end{proof}
\subsection{Bi-skew-symmetric Hermitian tensors}\label{subsec:bisksymt}
\begin{definition}\label{defskewsymher}
Assume that $\bn=n^{\times d}$ and $2\le d \le n$.  
Let $\bi=(i_1,\ldots,i_d)\in[n^{\times d}]$.  For  a permutation $\omega:[d]\to[d]$ denote $\omega(\bi)=(i_{\omega(1)},\ldots,i_{\omega(d)})$.  
A Hermitian tensor $\cT\in \rH_{n^{\times d}}$ is called bi-skew-symmetric (bi-anti-symmetric) Hermitian if for any permutation $\omega\in\Omega_d$ on has the equality
\begin{equation}\label{bisymdef}
 t_{\omega(\bi), \bj}=\mathrm{sign}(\omega)t_{\bi,\bj} \textrm{ for all  }\bi,\bj\in[n^{\times d}], \omega\in\Omega_d\quad \cT=[t_{\bi,\bj}]\in \rH_{n^{\times d}}. 
\end{equation}
Denote by $\rA\rH_{n^{\times d}} \subset \rH_{n^{\times d}}$ the subspace of bi-skew-symmetric Hermitian tensors.
\end{definition}
Observe that  if $\cT$ is bi-skew-symmetric Hermitian it follows that for permutations $\omega$ and $\psi$ we have the equalities $t_{\omega(\bi)\psi(\bj)}=\textrm{sign}(\omega\psi)t_{\bi\bj}$. 

The following results are analous to the results of subsection \ref{subsec:bisymt}.    
The proofs of the claims that are similar to the proofs in subsection \ref{subsec:bisymt} are omitted. 
\begin{lemma}\label{spdcbah}  Let $2\le d\le n$ 
and $0\ne\cT\in\rA\rH_{n^{\times d}}$.  Denote by $\Lambda\subset [n^d]$ the set of nonzero eigenvalues of $\cT$ given by \eqref{specdec}.
Then,  $\cX_k\in\rA^d\C^n$ for $k\in\Lambda$.
If in addition $\cT$ has real entries then $\cX_k\in \rA^d\R^n$ for $k\in\Lambda$.
Thus, $\mathrm{range}\,\rT\subset \rA^d\C^n$.  Hence, 
\begin{equation}\label{dimAHnd}
\begin{aligned}
&\dim \rA\rH_{n^{\times d}} ={n\choose d}^2,\\
&\dim \Re\rA\rH_{n^{\times d}} ={n\choose d}\left({n\choose d}+1\right)/2,
\end{aligned}
\end{equation}
where $\Re\rA\rH_{n^{\times d}}\subset \rA\rH_{n^{\times d}} $ is the set of real valued bi-skew-symmetric Hermitian tensors.
\end{lemma}
\begin{corollary}\label{bodta}
The set of positive semidefinite bi-skew-symmetric Hermitian tensors with trace one, denoted by $\rA\rH_{n^{\times d},+,1}$, correspond to the convex set of bi-skew-symmetric density tensors generated by Fermions on $\rA^d\C^n$.
\end{corollary}
\begin{definition} \label{defstsepa}
A bi-skew-symmetric Hermitian density tensor  $\cR\in\rA\rH_{n^{\times d}}$ is called strongly separable if $\cR$ is a convex combinations of $\cX\otimes \overline{\cX},  \cX\in [\W]_1, \W\in \mathrm{Gr}(d,\C^n)$.   Denote by $\mathrm{Sepa}_{n^{\times d}}$ the set of strongly separable bi-skew-symmetric Hermitian density tensors in $\rA\rH_{n^{\times d}}$. 
\end{definition}
\begin{lemma}\label{bisksymspecnuclem} Let $2\le d< n\in\N$.   Denote by $\cI_{{n\choose d}}\in\rA\H_{n^{\times d}}$ the identity tensor acting on $\rA^d\C^n$, which has a spectral decomposition as in Lemma \ref{spdcbah},  whose nonzero eigenvalues are $1$.  Assume that $\cB\in\rA\rH_{n^{\times d}}$.   
\begin{enumerate}
\item 
The following equality holds:
\begin{equation}\label{defaspecnrm}
\|\cB\|_{spec}:=\max_{\cY\in[\W]_1, \W\in \mathrm{Gr}(d,\C^n)}\frac{1}{d!}|\langle\cY,\cB\cY\rangle|
\end{equation}
\item 
Let $\|\cB\|_{anuc}$ be the dual norm of the restriction  $\|\cdot\|_{spec}$ to $\rA\rH_{n^{\times d}}$, called an a-nuclear norm:
\begin{equation}\label{nucnrma}
\|\cB\|_{anuc}=\max_{\cC\in\rA\rH_{n^{\times d}}, \|\cC\|_{spec}=1} |\langle\cC,\cB\rangle|.
\end{equation}
The extreme points if $\rB_{\|\cdot\|_{anuc}}(0,1)$ are $\pm \frac{1}{d!}\cY\otimes\overline{\cY}, \cY\in [\W]_1, \W\in\Gr(d,\C^n)$.  Hence,
\begin{equation}\label{charnuca}
\|\cB\|_{anuc}=\min_{\cX=\sum_{i=1}^r \frac{a_i}{d!}  \cY_i\otimes \overline{\cY_i}, \cY_i\in[\W]_1,\W\in\rG(d,\C^n),a_i\in\R, i\in[r], r\in\N}\sum_{i=1}^r |a_i|.
\end{equation}
Furthermore, the minimum is achieved for some $r\le {n\choose d}^2+1$.
\item The following inequalities hold:
\begin{equation}\label{anucin}
\|\cB\|_{nuc}\ge\|\cB\|_{anuc}\ge \|\cB\|_{1,\C}.
\end{equation}
\item 
Denote
\begin{equation}\label{defalbeinah}
\begin{aligned}
&\alpha_{\infty,n^{\times d},ah}=\min_{\cB\in\rA\rH_{n^{\times d}},\|\cB\|=1} \|\cB\|_{spec}, \quad \beta_{\infty,n^{\times d},ah}=\max_{\cB\in\rA\rH_{n^{\times d}},\|\cB\|=1} \|\cB\|_{spec},\\
&\alpha_{1,n^{\times d},ah}=\min_{\cB\in\rA\rH_{n^{\times d}},\|\cB\|=1} \|\cB\|_{anuc}, \quad \beta_{1,n^{\times d},ah}=\max_{\cB\in\rA\rH_{n^{\times d}},\|\cB\|=1} \|\cB\|_{anuc}.
\end{aligned}
\end{equation}
Then 
\begin{equation}\label{alpbethinah}
\begin{aligned}
&\alpha_{\infty,n^{\times d},ah}< \beta_{\infty,n^{\times d},ah}=\frac{1}{d!}, \quad \alpha_{1,n^{\times d},ah}=d!<\beta_{1,n^{\times d},ah},\\
&\alpha_{\infty,n^{\times d},ah}\beta_{1,n^{\times d},ah}=1.
\end{aligned}
\end{equation}
The equalities $\beta_{\infty,n^{\times d},ah}=\|\cB\|_{\infty}, \|\cB\|=1$ and $\alpha_{1,n^{\times d},ah}=\|\cC\|_{\infty}, \|\cC\|=1$ if and only if $\cB$ and $\cC$ are of the form $\pm (\z_1\wedge\cdots\z_d)\otimes (\bar\z_1\wedge\cdots\bar\z_d)$ for some orthonormal set of $\z_1,\ldots,\z_d$.
\item The group $\rU_{n}$ acts on $\rA\rH_{n^{\times d}}$ via the diagonal action $\cR\mapsto (\otimes^d U)\cR \overline{(\otimes^d U)}$:
\begin{equation*}
\begin{aligned}
&\cX\otimes \overline{\cX}\mapsto((\otimes^d U)\cX)\otimes\overline{((\otimes^d U)\cX)},\\
&(\x_1\wedge\cdots\wedge\x_d)\otimes(\overline{\x_1\wedge\cdots\wedge\x_d})\mapsto ((U\x_1)\wedge\cdots\wedge(U\x_d))\otimes(\overline{(U\x_1)\wedge\cdots\wedge(U\x_d)}).
\end{aligned}
\end{equation*}
as a group of isometries of the norms $\|\cdot\|_{spec}$ and $\|\cdot\|_{anuc}$.
Let $\mu_{n}$ be the Haar measure on $\rU_{n}$.   Then for any $\cR\in\rA\rH_{n^{\times d},+,1}$ the following equality holds:
 \begin{equation}\label{Haareqsk}
 \int_{\rU_{n}} ((\otimes^d U)\cR \overline{(\otimes^d U)}d\mu_{n}=\frac{1}{{n\choose d}}\cI_{{n\choose d}}.
 \end{equation}
\end{enumerate}
\end{lemma}
\begin{proof}\emph{(1)}  Assume that $\x_j\in\C^n, j\in[d]$ and $\cX\in \rA^d\C^n$.   The first equality in \eqref{ferid} yields 
\begin{equation*}
\begin{aligned}
&\langle \otimes_{j=1}^d\x_i,\cX\rangle=
\frac{1}{\sqrt{d!}}\langle\x_1\wedge\cdots\wedge\x_d,\cX\rangle \Rightarrow\\
&\langle (\otimes_{j=1}^d\x_i)\otimes (\overline{ \otimes_{j=1}^d\x_i}), \cX\otimes\overline{\cX}\rangle=
\langle \otimes_{j=1}^d\x_i,\cX\rangle \langle \overline{\otimes_{j=1}^d\x_i},\overline{\cX}\rangle=\langle \otimes_{j=1}^d\x_i,\cX\rangle\langle \cX,  \otimes_{j=1}^d\x_i\rangle=\\
&\frac{1}{d!}\langle\x_1\wedge\cdots\wedge\x_d,\cX\rangle\langle \cX,\x_1\wedge\cdots\wedge\x_d\rangle=\frac{1}{d!}\langle\x_1\wedge\cdots\wedge\x_d,(\cX\otimes\overline{\cX})\x_1\wedge\cdots\wedge\x_d\rangle.
\end{aligned}
\end{equation*}
Assume that $\cB\in\rA\rH_{n^{\times d}}$ has a decomposition as in Lemma \ref{spdcbah}.  From the above equalities we deduce
\begin{equation}\label{Bskhid}
\langle \otimes_{j=1}^d \x_j, \cB(\otimes_{j=1}^d\x_j)\rangle=\frac{1}{d!}\langle \x_1\wedge\cdots\wedge\x_d,\cB(\x_1\wedge\cdots\wedge\x_d)\rangle.
\end{equation}
Assume now that $\|\x_j\|=1,j\in[d]$.  Then $\|\x_1\wedge\cdots\wedge\x_d\|\le 1$ and equality holds if and only if $\x_1,\ldots,\x_d$ is an orthonormal set in $\C^n$.
Hence,  \eqref{defaspecnrm} holds.

\noindent\emph{(2)}. As in the proof of part \emph{(3)} of Lemma \eqref{bisymspecnuclem}.

\noindent\emph{(3)}.  Assume that $\cB\in \rA\rH_{n^{\times d}}$.  Recall that $\|\cB\|_{nuc}=\max_{\cC\in \rH_{n^{\times d}},\|\cC\|_{nuc}=1} \langle \cC,\cB\rangle$.
Compare the characterization to \eqref{nucnrma} to deduce the first inequality in  \eqref{anucin}.  To deduce the second inequality in \eqref{anucin} we observe that 
\begin{equation*}
\begin{aligned}
&\frac{1}{d!}( \x_1\wedge\cdots\wedge\x_d)\otimes ( \overline{\x_1\wedge\cdots\wedge\x_d})=\\
&\frac{1}{(d!)^2} \sum_{\omega,\psi\in\Omega_d}(\textrm{sign}(\omega)\otimes_{j=1}^d \x_{\omega(j)})\otimes ((\textrm{sign}(\psi)\otimes_{j=1}^d \bar\x_{\omega(j)}).
\end{aligned}
\end{equation*}
Assume that $\|\x_1\|=\cdots=\|\x_d\|=1$.  Then, the above sum is a convex combination of $(d!)^2$ extreme points of the unit ball of the norm $\|\cdot\|_{1,\C}$.
Hence, the minimal characterizations $\|\cB\|_{anuc}$ and $\|\cB\|_{1,\C}$ implies the inequality $\|\cB\|_{anuc}\ge\|\cB\|_{1,\C}$.

\noindent\emph{(4)}  Let $\cB\in \rA\rH_{n^{\times d}}, \|\cB\|=1$.   Consider the spectral decomposition of $\cB$ as given in Lemma \ref{spdcbah}:
\begin{equation*}
\cB=\sum_{i=1}^{{n\choose d}}\lambda_i(\cB)\cX_i\otimes \overline{\cX_i}, \cX_i\in\rA^d\C^n, \langle\cX_i, \cX_j\rangle=\delta_{ij},i,j\in[{n\choose d}], \sum_{i=1}^{{n\choose d}}\lambda_i^2(\cB)=1.
\end{equation*}
Use the \eqref{Bskhid} and the fact that $|\lambda_i|\le 1$ to deduce that 
\begin{equation*}
\begin{aligned}
&|\langle \otimes_{j=1}^d \x_j, \cB(\otimes_{j=1}^d\x_j)\rangle|= \frac{1}{d!}|\sum_{i=1}^{{n\choose d}}\lambda_i|\langle \x_1\wedge\cdots\wedge\x_d,\cX_i\rangle|^2|\le\\ &\frac{1}{d!}|\sum_{i=1}^{{n\choose d}}|\langle \x_1\wedge\cdots\wedge\x_d,\cX_i\rangle^2|=\frac{1}{d!} \|\x_1\wedge\cdots\wedge\x_d\|^2\le \frac{1}{d!}.
\end{aligned}
\end{equation*}
Equality $|\langle \otimes_{j=1}^d \x_j, \cB(\otimes_{j=1}^d\x_j)\rangle|=\frac{1}{d!}$ holds if and only if $\x_1,\ldots,\x_d$ is an orthonormal set of vectors,  and $\cB=\pm (\x_1\wedge\cdots\wedge\x_d)\otimes(\bar\x_1\wedge\cdots\wedge\bar\x_d)$.
This proves the first inequality in \eqref{alpbethinah}. Other inequalities follow from Theorem \ref{albetm}.

\noindent\emph{(5)}.  Assume that $\x_1,\ldots,\x_n\in\C^n$ is an orthonormal basis $\C^n$.   The $U\x_1,\ldots,U\x_n$  is an orthonormal basis in $\C^n$.    Recall that an orthonormal basis $\x_1,\ldots,\x_n$ corresponds to a unitary matrix $U\in\rU_n$ whose columns are $\x_1,\ldots,\x_n$.  Hence, $\rU_n$ acts transitively on orthonormal bases in $\C^n$. 
Recall that
$\x_{i_1}\wedge\cdots\wedge\x_{i_d}, 1\le i_1<\cdots<i_d\le n$ is an orthonormal basis in $\rA^d\C^n$.    Therefore,
\begin{equation*}
\begin{aligned}
&I_{n\choose d}=\sum_{1\le i_1<\cdots <i_d\le n}(\x_{i_1}\wedge\cdots\wedge\x_{i_d})\otimes(\overline{\x_{i_1}\wedge\cdots\wedge\x_{i_d}})\Rightarrow (\otimes^d U)I_{n\choose d} (\overline{\otimes^d U})=I_{n\choose d}\Rightarrow\\
&\int_{\rU_{n}}(\x_1\wedge\cdots\wedge\x_n)\otimes (\bar\x_1\wedge\cdots\wedge\bar\x_n)d\mu_n=\frac{1}{{n\choose d}}I_{{n\choose d}}.
\end{aligned}
\end{equation*}
Let $0\ne\cB\in \rA\rH_{n^{\times d}}$.   Assume that $\cB=\sum_{i=1}^m \frac{a_i}{d!}\y_{1,i}\wedge\cdots\wedge\y_{d,i}$ is the minimal decomposition of $\cB$ with respect to a-nuclear norm.  Thus, $a_i\in \R\setminus\{0\}$, and $\y_{1,i},\ldots,\y_{d,i}$ 
is an orthonormal set of vectors in $\C^n$ for $i\in[m]$.  Observe that $\tr B=\frac{1}{d!}\sum_{i=1}^m a_i$.  Hence,  
$$\int_{U\in \rU_n} (\otimes^dU)\cB(\otimes^d \bar U)d\mu_n=\frac{\tr\cB}{{n\choose d}} I_{{n\choose d}}.$$
This proves \eqref{Haareqs}.
\end{proof}
 The following theorem is an analog of Theorem \ref{charseps}:
\begin{theorem}\label{charsepa}  Let $2\le d\le n$.  
Then
\begin{enumerate}
\item  A bi-skew-symmetric density tensor $\cR\in\rA\rH_{n^{\times d},+,1}$, with spectral decomposition as in Lemma \ref{spdcbah},  satisfies the inequalities
\begin{equation}\label{lubsndta}
\begin{aligned}
&\frac{1}{d!{n\choose d}}\le \|\cR\|_{spec}\le \frac{\lambda_1(\cR)}{d!}\le \frac{1}{d!},\\
&d!\le \|\cR\|_{anuc}\le \beta'_{1,n^{\times d},ah}:=\max_{\cX\in\rA^d\C^n,\|\cX\|=1}\|\cX\otimes \overline{\cX}\|_{anuc}.
\end{aligned}
\end{equation}
Equalities in some of the above inequalities hold under the following conditions:
\begin{equation}\label{lubsndteqa}
\begin{aligned}
&\|\cR\|_{spec}=\frac{\lambda_1(\cR)}{d!}\iff \exists \x_i\in\C^n, \langle \x_i,\x_j\rangle=\delta_{ij}, i,j\in[d]\\ 
&\textrm{ such that } \langle \x_1\wedge\cdots\wedge\x_d,\cX_k\rangle = 0 \textrm{ if } \lambda_k(\cR) <\lambda_1(\cR),\\
&\|\cR\|_{spec}=\frac{1}{d!}\iff \cR= (\x_1\wedge\cdots\wedge\x_d)\otimes (\bar\x_1\wedge\cdots\wedge\bar\x_d),\langle \x_i,\x_j\rangle=\delta_{ij}, i,j\in[d],\\
&\frac{1}{d!{n\choose d}}=\|\cR\|_{spec} \iff \cR=\frac{1}{{n\choose d}}\cI_{{n\choose d}},\\
&d!=\|\cR\|_{anuc} \iff \cR\in\mathrm{Sepa}_{n^{\times d}},\\
& \|\cR\|_{anuc}=\beta'_{1,n^{\times d},ah} \Leftarrow \cR=\cX\otimes \overline{\cX}, \cX\in\rA^d\C^n, \|\cX\|=1, \|\cX\otimes \overline{\cX}\|_{anuc}=\beta'_{1,n^{\times d},ah}.
\end{aligned}
\end{equation}
\item The set $\mathrm{Sepa}_{n^{\times d}}$ is a face of the ball of $\|\cdot\|_{anuc}$ of the maximal dimension ${n\choose d}^2-1$.  Its supporting hyperplane is
\begin{equation}\label{suphypseps}
 \tr \cB=\langle \cI_{n\choose d},\cB\rangle\le 1 \textrm{ for } \|\cB\|_{anuc}\le d!.
 \end{equation}
 Equality holds if and only if $\cB$ is strongly separable.
In particular, any strongly separable density tensor is a convex combinations of at most ${n\choose d}^2$  bi-skew-symmetric density tensors of the form $(\x_1\wedge\cdots\wedge\x_d)\otimes (\bar\x_1\wedge\cdots\wedge\bar\x_d),\langle \x_i,\x_j\rangle=\delta_{ij}, i,j\in[d]$.  
\item The inequality \eqref{alph2nin} holds.
\end{enumerate}
\end{theorem}
\begin{proof}\emph{(1)-(2)}.
In the proof of part \emph{(4)} of Lemma \ref{bisksymspecnuclem} replace $\cB$ by $\cR$, and the condition $\sum_{i=1}^{{n\choose d}} \lambda^2_i(\cB)=1$ by the conditions $\lambda_1(\cR)\ge \cdots\ge \lambda_{{n\choose d}}(\cR)\ge 0, \sum_{i=1}^{{n\choose d}} \lambda_i(\cR)=1$.
Observe the inequality $\lambda_i(\cR)\le 1$ to deduce the inequality $\|\cR\|_{spec}\le \frac{\lambda_1(\cR)}{d!}$.   The first two equality cases \eqref{lubsndta} are deduced straightforward.

To deduce the lower bound $\|\cR\|_{aspec}= \frac{1}{d!{n\choose d}}$ we combine the proof of part \emph{(5)} of  Lemma \ref{bisksymspecnuclem}  with the arguments of the proof of similar lower bound in Theorem \ref{lubsndta}.  
Other parts of \emph{(1)-(2)} of this theorem are deduced as in the proof of Theorem \ref{charseps}.

\noindent\emph{(3)}  Let $\cX\in \rA^d\C^n, \|\cX\|=1$.   Clearly,  $\|\cX\otimes \bar\cX\|_{spec}=\|\cX\|_{\infty,\C}^2$.  The first inequality in \eqref{lubsndta} yields that $\|\cX\|_{\infty,\C}^2\ge \frac{1}{d!{n\choose d}}$.  This proves \eqref{alph2nin}.

\end{proof}
\begin{corollary}\label{charsepsca}  Denote by $\mathrm{Sepa}_{n^{\times  d},int}$ the relative interior of $\mathrm{Sepa}_{n^{\times d}}$.  
\begin{enumerate}
\item The sets $\mathrm{Sepa}_{n^{\times d}}$ and $\mathrm{Sepa}_{n^{\times d},int}$ are invariant under the action of $\rU_{n}$.
\item Let $\cR\in \mathrm{Sepa}_{n^{\times d},int}$.  There exists $r(\cR)>0$ such that $\rB_{\|\cdot\|}(\cR,r(\cR))\cap\rA\rH_{n^{\times d},+,1}\subset \mathrm{Sepa}_{n^{\times d},int}$.
\item The density tensor $ \frac{1}{{n\choose d}}I_{{n\choose d}}$ is in $\mathrm{Sepa}_{n^{\times d},int}$, where 

\noindent
$r(\frac{1}{{n\choose d}}I_{{n\choose d}})\ge \frac{2}{{n\choose d}\beta_{1,n^{\times d},ah}}$.
\end{enumerate} 
\end{corollary}
The proof of this corollary is similar to the proof of the corresponding parts of  Corollary \ref{charsepscb}.
\subsection{The set of separable density tensors is semi-algebraic}\label{subsec:semal}
The main result of this subsection is:
\begin{theorem}\label{sasep} The separability sets $\mathrm{Sep}_{\bn}$,  $\mathrm{Seps}_{n^{\times d}}$ and their restrictions to $\rS^{2d}\R^n_{+,1}$ are semi-algebraic sets.  
\end{theorem}

We first explain the notion of a semi-algebraic set and its importance in quantum information theory.   Denote by $\F[\x]$ the ring of polynomial over the field $\F$ with $n$ coordinate variables of $\x\in \F^n$.  For $p\in\F[\x]$ denote by $\rZ(p)=\{\x\in\F^n: \,p(\x)=0\}$, the zero set of $p$.   A set $\rV\subset \F^n$ is called an algebraic variety if $\rV=\cap_{l=1}^m \rZ(p_l)$.  Algebraic geometry studies various properties of algebraic varieties over the complex numbers $\C$ \cite{Har92}.

For $p\in\R[\x]$ denote $\rZ_+(p)=\{\x\in\R^n: p(\x)\ge 0\}$.    Observe that $\rZ(p)=\rZ_+(p)\cap \rZ_+(-p)$.  A basic semi-algebraic set $W\subset\R^n$ is given by $W=\cap_{l=1}^m  \rZ_+(p_l)$.  A semi-algebraic set is a finite union of basic semi-algebraic sets. 
In quantum information one needs to consider semi-algebraic sets over $\C^n$.   We identify $\C^n$ with $\R^n\times\R^n\sim\R^{2n}$  where $\z=(\x,\y)=\x+\bi\y, \x,\y,\in\R^n$.
Observe that a complex algebraic variety is a real algebraic variety on $\R^{2n}$.

The first simple examples of semi-algebraic sets are $\rH_{\bn}, \rH_{\bn,+},\rH_{\bn,+,1}$.  It is enough to consider the sets $\rH_n, \rH_{n,+},\rH_{n,+,1}$.   Clearly $\rH_n=\rS^2\R^n\times\rA^2\R^n$.   Hence,  $\rH_n$ is a real algebraic set.  
Recall that a Hermitian matrix is positive semidefinite if and only its all the  principal minors are nonnegative.  Hence,  $\rH_{\bn,+}$ and $\rH_{\bn,+,1}$ are  basic semi-algebraic sets.  

We next recall the definition of positive partial transport density tensors introduced by Peres for $d=2$ \cite{Per96}:
\begin{definition}\label{PPT}
Let $2\le d, 2\1_d\le \bn\in\N^d$.   For $k\in[d]$ denote by $\rP\rT_k:\rH_{\bn}\to \rH_{\bn}$ the $k$-partial transpose:
\begin{equation}
\begin{aligned}
&\rP\rT_k(\cB)=\cC,  \cB=[b_{(i_1,\ldots,k_d)(j_1,\ldots,j_d)}], \cC=[c_{(i_1,\ldots,i_d)(j_1,\ldots,j_d)}], \\
& c_{(i_1,\ldots,i_{k-1}, i_k,i_{k+1},\cdots,i_d)(j_1,\ldots,j_{k-1}, j_k,j_{k+1},\cdots,j_d)}=
b_{(i_1,\ldots,i_{k-1}, j_k,i_{k+1},\cdots,i_d)(j_1,\ldots,j_{k-1}, i_k,j_{k+1},\cdots,j_d)},\\
&(i_1,\ldots,i_d),(j_1,\ldots,j_f)\in[\bn], \quad k\in[d].
\end{aligned}
\end{equation}
Denote by $\rH_{\bn,ppt,1}$ the biggest convex subset  of $\rH_{\bn,+,1}$ which is invariant under $\rP\rT_k$ for $k\in[d]$.  This convex subset is called the positive partial transpose set of density tensors.
\end{definition}
\begin{proposition}\label{ptprop} Let $2\le d,  2\1_d\le \bn\in\N^d$.  
\begin{enumerate}
\item The set $\rP\rT_1,\ldots,\rP\rT_d$ is a commutative group of involutions acting on $\rH_{\bn}$ which preserve trace.
\item 
$\rP\rT_1\cdots\rP\rT_d(\cB)=\overline{\cB}$.
\item
\begin{equation}
\rP\rT_k\left(\otimes_{j=1}^d \x_j)\otimes(\otimes_{j=1}^d \bar\x_j)\right)=(\otimes_{j=1}^{k-1}\x_j)\otimes\bar\x_k\otimes (\otimes_{j=k+1}^d\x_j)\otimes
(\otimes_{j=1}^{k-1}\bar\x_j)\otimes\x_k\otimes (\otimes_{j=k+1}^d\bar\x_j),
\end{equation}
for $k\in[d]$.
\item The set $\rH_{\bn,ppt,1}$ is basic semi-algebraic.
\end{enumerate}
\end{proposition}
\begin{proof} The proof of {(1)-(3)} is straightforward.  

\noindent \emph{(4)} The set $\rH_{\bn,ppt,1}$ is obtained from $\rH_{\bn}$ by imposing a number of polynomial inequalities.  As $\rH_{\bn,+,1}$ is basic semi-algebraic it follows that is basic semi-algebraic.
\end{proof}

An important observation of Peres \cite{Per96} (for $d=2$) is: 
\begin{corollary}  Let $2\le d, 2\1_d\le \bn\in\N^d$.   Then $\mathrm{Sep}_{\bn}\subset  \rH_{\bn,ppt,1}$. 
\end{corollary}
It was shown in \cite{Hor96} that for $\bn=(2,2), (2,3)$ $\mathrm{Sep}_{\bn}=\rH_{\bn,ppt,1}$, and in \cite{Hor97} that for $\bn=(2,4), (3,3)$ $\mathrm{Sep}_{\bn}\not\subset\rH_{\bn,ppt,1}$.

The fundamental Tarski-Seidenberg theorem states
\begin{theorem}\label{TSthm}  Let $X\subset \R^n$ be semi-algebraic.  Assume that $F:\R^n \to \R^m$ is a polynomial map.  The $F(X)$ is semi-algebraic.
\end{theorem}
\begin{definition}\label{convk}  Let $X\subset \R^n$ be a compact set.  For $k\in\N$ denote by $\mathrm{conv}_k(X)$ the convex hull spanned by $k$-points in $X$.  Let 
$\mathrm{conv}\,X$ be the convex hull of $X$.
\end{definition}
Assume that $\dim \mathrm{conv}\,X=m$.  Caratheodory's theorem yields that $\mathrm{conv}_{m+1}(X)=\mathrm{conv}\,X$.  The following result is known and we bring its proof for completeness:
\begin{proposition}\label{convprop}
Let $X\subset \R^n$ be a compact semi-algebraic set.  Assume that $m=\dim \mathrm{conv}\,X$.  Then for $k\in[m+1]$ the convex set $\mathrm{conv}_k X$ is  semi-algebraic.  In particular,  $\mathrm{conv}\, X$ is  semi-algebraic. 
\end{proposition}
\begin{proof}
Clearly,  the $k-1$ dimensional sphere $\mathbb{S}^{k-1}\subset \R^k$ is an algebraic set.  Recall that a product of two semi-algebraic sets is semi-algebraic \cite{BCR98}.
Hence, the set $Y=\mathbb{S}^{k-1}\times X^{\times k}$ is semi-algebraic.  Let $\y=(\bu,\x_1,\ldots,\x_k)\in Y$, where $\bu=(u_1,\ldots,u_k)^\top\in \R^k$ and $\x_i\in\X$ for $i\in[k]$.  Consider the map $p_k:\R^k\times \R^{kn}\to\R^n$ given by 
\begin{equation*}
p_k(\y)=\sum_{i=1}^k u_i^2\x_i.
\end{equation*}
Clearly, $p_k$ is a polynomial map.  Observe that $\mathrm{conv}_k(X)=p_k(Y)$.
Hence, the Tarski-Seidenberg Theorem states that  $\mathrm{conv}_k X$ is  semi-algebraic.    As $\mathrm{conv}_{m+1}(X)=\mathrm{conv}\,X$ we deduce that $\mathrm{conv}\,X$ is semi-algebraic.
\end{proof}
The folowing result is well known and we bring its proof for completeness:
\begin{lemma}\label{Segv}  The set of rank-one tensors plus the zero tensor, denoted as $\Sigma_{\bn}\cup\{0\}$,  is a complex algebraic variety in $\otimes_{j=1}^{d}\C^{n_j}$.
\end{lemma}
\begin{proof}
Let $0\ne\cT\in\otimes_{j=1}^d \C^{n_j}.$  
View $\cT$ as a matrix in $\C^{n_1}\otimes(\otimes_{j=2}^d \C^{n_j})$.  This is called an unfolding  with respect to $\C^{n_1}$.  Assume that all $2\times 2$ minors of this matrix are zero.  Then this matrix is a rank-one matrix of the form $\x_1\otimes \cT_1$, where $\cT_1\in\otimes_{j=2}^d\C^{n_j}$.    If $d=2$ it follows that $\cT$ has rank-one.  Suppose that $d>2$.  Unfold $\cT$ the $j=2$ to view $\cT$ is a matrix on $\C^{n_2}\otimes (\otimes_{j\ne 2}\C^{n_j})$.  Assume now that all $2\times 2$ minors of this matrix are zero.  It now follows that $\cT=\x_1\otimes\x_2\otimes \cT_2$.  Continue to unfold in each $\ge 3$ and assume that all  $2\times 2$ minors of this matrix are zero.  We conclude that $\cT$ has rank one.  Observe finally that a zero tensor satisfies all the above quadratic equations.
\end{proof}
\emph{Proof of Theorem \ref{sasep}.}  Recall that the extreme points of $\mathrm{Sep}_{\bn}$,  denoted as $\rE_{\bn}$, are rank-one tensors in $\rH_{\bn,+,1}$.
Clearly, $\rE_{\bn}=\rH_{\bn,+,1}\cap (\Sigma_{(\bn,\bn)}\cup\{0\})$. Hence,  $\rE_{\bn}$ is a basic semi-algebraic set.   Proposition \ref{convprop} yields that $\mathrm{Sep}_{\bn}$ is semi-algebraic in $\rH_{\bn}$.

Similar arguments apply to strongly separable and real strongly separable density tensors in $\rB\rH_{n^{\times d},+,1}$ and $\rS^{2d}\R^n_{+,1}$ respectively.\qed\\ 

Recall that Tarski-Seidenberg Theorem is proved by elimination of quantifiers.  The proof is constructive,  but the resulting algorithm has a  high computational complexity. Collins  \cite{Col75} introduced the algorithm of cylindrical algebraic decomposition, which allows quantifier elimination over the reals in double exponential time. This complexity is optimal.

Observe that the $\dim\rH_{(2,4)}=64$.  It seems to be impossible to use a brute force algorithm to find semi-algebraic characterization of $(2,4)$ separable states.  The dimension of $\rS^4\R^3$ is ${2+4-1\choose 4}=5$.   
Hence, it seems possible to find a semi-algebraic characterization of real strongly separated density tensors in $\rS^4\R^3_{+,1}$ with computer help.

 \section{Complexity results on tensors, entanglement and separablity}\label{sec:comp}
The main result of this section is:
\begin{theorem}\label{polthm} Let $1\le n\in\N$ be fixed.  The following 
quantities are poly-time computable in $d\in\N$:
\begin{enumerate}
\item The norms $\|\cdot\|_{\infty,\F}$ and $\|\cdot\|_{1,\F}$ on $\rS^d\F^n$.
\item The norms  $\|\cdot\|_{\infty,\F}$ and $\|\cdot\|_{1,\F}$ for $\rS^d\F^n\otimes\rS^d\F^n$.
\end{enumerate}
\end{theorem}

 \subsection{Poly-time,  NP and NP-hard problems}\label{subsec:NP}
 Poly-time problems are computational problems that can be solved in time  bounded by a polynomial function of the input size.
These problems include multiplication and inverse of matrices,  finding the roots of a polynomial in one variable,  eigenvalues and eigenvectors of square matrices, Gram-Schmidt orthogonalization process,  singular value decomposition (Schmidt decomposition), and other known algorithms.

The class NP (nondeterministic poly-time) encompasses decision problems where a solution can verified in poly-time.  An NP-complete problem is the hardest problem in the class  NP: a poly-time algorithm for any NP-complete problem would imply that all problems in the class NP are solvable in poly-time.
An NP-hard problem is a problem that is at least as hard as NP-complete problem.
If a poly-time algorithm exists for any NP-hard problem, then all NP problems are solvable in poly-time.

Since we are interested in computing the various spectral and nuclear norms of tensors the following complexity result is very useful \cite{FL16}:
\begin{theorem}\label{dualnrmthm} Let $\F=\R,\C$ and $\nu,\nu^*:\F^n \to [0,\infty)$ be dual norms.  Then $\nu(\x)$ is poly-time computable for $\x$ with rational (Gaussian rational) entries if and only if  $\nu^*(\x)$ is poly-time computable for $\x$ with rational (Gaussian rational) entries.
\end{theorem}
 
Let us start with the largest clique problem as in \cite{FL18}.
An undirected simple graph is denoted by $G=(V,E)$, where $V$ is the set of vertices, and $E$ is the set of undirected edges $\{u,v\}\in E, u\ne v\in V$.  Assume that the cardinality of $V$ is $n$: $|V|=n$.   We identify $V$ with $[n]$. Then $A(G)=[a_{ij}]\in \{0,1\}^{n\times n}$ is a symmetric matrix with zero diagonal whose nonzero entries are equal to $1$.   Thus, $a_{ij}=1\iff \{i,j\}\in E$.   A  clique is a set $C\subset[n]$ such for  all $i\ne j\in C$ we have $\{i,j\}\in E$.   Denote by $\kappa(G)$ the maximum cardinality of all cliques in $G$.  It is known that finding $\kappa(G)$ is a NP-complete problem.  In \cite{MS65} the authors give a characterization for $\kappa(G)$ which is equivalent to the following characterization:
\begin{equation}\label{MSchar}
1-\frac{1}{\kappa(G)}=\max_{\|\x\|=1,\x\in\R^n}\sum_{i=j=1}^n a_{ij}x_i^2 x^2_j=\max_{\|\x\|=1,\x\in\C^n}|\sum_{i=j=1}^n a_{ij}x_i^2 \bar x^2_j|.
\end{equation}
\begin{lemma} The following quantities are NP-hard to compute for a fixed $m\in\N$:
\begin{enumerate}
\item The norms $\|\cdot\|_{\infty,\F}$ and $\|\cdot\|_{1,\F}$ for tensors in $\rS^{4m}\F^n$ for $\F=\R,\C$.
\item The norms $\|\cdot\|_{\infty,\C}$ and $\|\cdot\|_{1,\C}$ for  tensors in $\rH_{n^{\times 2m},+,1}$ and  $\rH_{n^{\times 2m}}$.
\item The norms $\|\cdot\|_{\infty,\R}$ and $\|\cdot\|_{1,\R}$ for  tensors in $\Re\rH_{n^{\times 2m},+,1}$ and  $\Re\rH_{n^{\times 2m}}$.
\item The norms $\|\cdot\|_{spec}$ and $\|\cdot\|_{nuc}$ for  tensors in $\rH_{n^{\times 2m},+,1}$ and  $\rH_{n^{\times 2m}}$.
\item The norms $\|\cdot\|_{\infty,\F}$ and  $\|\cdot\|_{1,\F}$ are NP-hard to compute on $\F^{n^{\times (4m)}}$.
\end{enumerate}
\end{lemma}
\begin{proof} \emph{(1)}  
Let $f(\x)=\sum_{(i,j)\in }a_{i,j}x_i^2x_j^2$, where $A=[a_{i,j}]=A(G)$.   Recall that
$f^m(\x)=\langle\cS_m,\otimes^d\x\rangle$,  $m\in\N$ for a corresponding $\cS_m\in\rS^{4m}\R^n$ \eqref{defpolfx}.   In view of \eqref{MSchar} it is NP-hard to find $\|\cS_m\|_{\infty,\R}$ for each $m\in\N$.   As all entries of $\cS_m$ are nonnegative, it follows that $\|\cS_m\|_{\infty,\C}\|=\|\cS_m\|_{\infty,\R}\|$.  Hence, it NP-hard to compute $\|\cS_m\|_{\infty,\C}$.
Therefore,  it is NP-hard to compute the dual norm $\|\cdot\|_{1,\F}$ for $\F=\R,\C$ \cite{FL16}.

\noindent\emph{(2-4)}
Observe that the maximum given in \eqref{MSchar} can be associated with the b-spectral norm of  following bi-symmetric Hermitian tensor $\cG=[g_{k,l,p,q}]\in \rB\rH_{n^{\times 2}}$:
\begin{equation}\label{defTG}
g_{k,l,p,q}=\begin{cases}
1 \textrm{ if }(k,l)=(i,i), (p,q)=(j,j),  \{i,j\}\in E,\\
0 \textrm{ otherwise}
\end{cases}
\end{equation}
Clearly, $\tr \cG=0$. 
Let $X=[x_{i,j}]\in\C^{n\times n}$ be a symmetric matrix.  ( $X\in\rS^2\C^n$).  Consider the sesquilinear form in entries of $X$:
$Q(X,\bar X):=\sum_{(i,j)\in E} x_{i,i}\bar x_{j,j}$.  As $|\x_{i,i}\x_{j,j}|\le \frac{1}{2}(|x_{i,i}|^2+|x_{j,j}|^2)$ it follows that 
$|Q(\x,\bar\x)|\le d\sum_{i\in [n]}\|x_{i,i}|^2\le d\|X\|^2$, where $d$ is the maximal degree of vertices in $G$.  Recall that $d\le n(n-1)/2$.  Let $\cI_{n(n+1)/2}$ be the identity tensor as defined in Lemma \ref{bisymspecnuclem}.

Define
\begin{equation}\label{def}
\cT=\frac{4}{n^2(n^2-1)}\left(\frac{n(n-1)}{2}\cI_{n(n+1)/2}+\cG\right).
\end{equation}
It is straightforward to show that $\cT\in \rB\rH_{n^{\times 2},+,1}$.
We claim that
\begin{equation}\label{nrmcT}
\|\cT\|_{bspec}=\frac{4}{n^2(n^2+1)}\left(n(n-1)+1-\kappa(G)\right).
\end{equation}
Clearly,  $\langle X,\cI_{n(n+1)/2}X\rangle=\|X\|^2$.   Next observe that 
\begin{equation*}
\begin{aligned}
&|\langle \x^{\otimes 2}\otimes\bar\x^{\otimes 2}, \cG\rangle|=|2\sum_{(i,j)\in E} \bar x_i^2 \x_j^2|\le 2\sum_{(i,j)\in E} |x_i|^2|x_j|^2\Rightarrow\\
&\|\cG\|_{bspec}=\max_{\|X\|=1,  X=\diag(\x), \|\x\|=1, \x=(x_1,\ldots,x_n)^\top\ge 0} \sum_{(i,j)\in E} x_i^2x_j^2=1-\kappa(G).
\end{aligned}
\end{equation*}
This proves \eqref{nrmcT}.  Hence,  it  is NP-hard to compute the $\|\cT\|_{bspec}$ for some tensor in $\rB\rH_{n^{\times 2},+,1}$.  In particular, it is NP-hard to compute $\|\cdot\|_{bspec}$.  Hence,  it is NP-hard to compute the dual norm $\|\cdot\|_{bnuc}$ \cite{FL16}.  This proves  \emph{(2)-(4)} for $m=1$.   
For $m>1$ use the tensor $\cT^{\otimes m}$.

\noindent \emph{(5)}  Follows from \emph{(1)},  as $\rS^{4m}\R^n\subset \rH_{n^{\times (2m)}}$. 
\end{proof}
\begin{corollary} The GME is NP-hard to compute for states in $\C^{n^{\times (4m)}}$ for $m\in\N$.
\end{corollary}

A remarkable result of Gurvits \cite{Gur03} states deciding if a density tensor $\cR\in \rH_{n^{\times 2},+,1}$ with Gaussian entries is separable is an NP-hard problem.  In view of Theorem \ref{charsep} this is equivalent to statement that deciding if $\|\cR\|_{nuc}=1$ is NP-hard.  This result was extended in \cite{Gha10}.
\subsection{Coding symmetric and bi-symmetric  Hermitian tensors as polynomials}
\label{subsec:dpolt}
\begin{definition}\label{defpolcS}
Denote by $\rP(d,n,\F)$ the linear space of homogeneous polynomials of degree $d$ in $n$-variables over $\F$.  Let 
\begin{equation*}
J(d,n)=\{\bj=(j_1,\ldots,j_n)\in\Z_+^d,  j_1+\cdots+j_d=d\}, \quad \Z_+:=\{0\}\cup \N.
\end{equation*}
Assume that $\bj\in J(d,n)$.  Denote by $\x^{\bj}=x_1^{j_1}\cdots x_n^{j_n}$.  Let  $\bj'=(j_{k_1},\ldots,j_{k_s}), s\in[d], 1\le k_1<\cdots<k_s\le n $ be obtained from $\bj$ by deleting  all $j_l=0$ in $\bj$.
A symmetric tensor $\cS=[s_{i_1,\ldots,i_d}]\in\rS^d\F^n$  corresponds to $f(\x)=\langle \bar \cS,\x^{\otimes d}\rangle$:
\begin{equation}\label{defpolfx}
\begin{aligned}
&f(\x)=\sum_{\bj\in J(d,n)} c(\bj) f_{j_1,\ldots,j_n} x_1^{j_1}\cdots x_n^{j_n}, \quad c(\bj)=\frac{d!}{j_1!\cdots j_n!},\bj\in J(d,n),\\
&f_{j_1,\ldots,j_n}=s_{i_1,\ldots,i_d} \textrm{ where } k_l \textrm{ appears } j_{k_l} \textrm{ times in } (i_1,\ldots,i_d),  l\in[s].
\end{aligned}
\end{equation}
The space $\rS^d\F^n\otimes \rS^d\F^n$ is called the space of bi-symmetric tensors $\cT=[t_{\bi,\bj}], \bi,\bk\in[\bn]$ in $\F^{n^{\times (2d)}}$:  $t_{\omega(\bi),\psi(\bk)}=t_{\bi,\bk}$ for $\omega,\psi\in\Omega_d$.  It corresponds to bi-homogeneous polynomials $\rP_{bi}(d,n,\F)$
\begin{equation}\label{defpolfxy}
f(\x,\y)=\sum_{\bj,\bl\in J(d,n)} c(\bj)c(\bk) f_{\bj,\bl}\x^{\bj}\y^{\bl},
\end{equation}
where the correspondence between $\cT\in \rS^d\F^n\otimes \rS^d\F^n$ and $f(\x,\y)$ is given as in \eqref{defpolfx}.   

The space $\rB\rH_{n^{\times d}}$ corresponds to Hermitian polynomials $f(\x,\bar\x)$:
\begin{equation*}
f(\x,\bar\x)=\sum_{\bj,\bl\in J(d,n)} c(\bj)c(\bk) f_{\bj,\bl}\x^{\bj}\bar \x^{\bl},  \quad 
f_{\bl,\bj}=\overline{f_{\bj,\bl}}, \bj,\bl\in[\bn],
\end{equation*}
denoted $\rP_{her}(d,n,\F)$, which is viewed as a real subspace of $\rP_{bi}(d,n,\F)$. 
(Here $\rP_{her}(d,n,\R)$ corresponds to $\Re \rS\rH_{n^{\times d}}$.) 
All correspondences are one-to-one and onto.
\end{definition}

\noindent
Hence, it is advantageous to replace $\rS^d\F^n, \rS^d\F^n\otimes\rS^d\F^n, \rH_{n^{\times d}}$ by $\rP(d,n,\F), \rP_{bi}(d,n,\F),\rP_{her}(d,n,\F)$ respectively.
The following lemma summarizes the properties of the above  isomorphisms,  and recalls Banach's characterization of the spectral norm of $\cS\in\rS^d\F^n$  \cite{Ban38}: 
\begin{lemma}\label{isolem}
\begin{enumerate}
\item Define  on $\rP(d,n,\F)$ the following inner product 
and the Hilbert norm:
\begin{equation}\label{inprodFJ}
\langle \f,\bg\rangle=\sum_{\bj\in J(d,n)} c(\bj)\overline{ f_{\bj}}g_{\bj},\quad \|\f\|=\sqrt{\langle \f,\f\rangle}.
\end{equation}
Then
\begin{enumerate}
\item Let $\e_{\bj}=(\delta_{\bj,\bk})_{\bk\in J(d,n)}, \bj\in J(d,n)$, where $\delta_{\bj,\bk}$ is Kronecker's delta function, be the standard basis in $\rP(d,n,\F)$.  Then $\frac{1}{\sqrt{c(\bj})}\e_{\bj}, \bj\in J(d,n)$ is an orthonormal basis in $\rP(d,n,\F)$ of cardinality $n+d-1\choose d$. 
\item The map $L:\rS^d\F^n\to\rP(d,n,\F)$ which is given by $L(\cS)=f$, where $f(\x)=\langle \bar \cS,\x^{\otimes d}\rangle$, is an isomorphism and an isometry. 
\item The spectral norm of a symmetric tensor is given by Banach's characterization \cite{Ban38}:
\begin{equation}\label{Banchar}
\|\cS\|_{\infty,\F}=\max_{\|\x\|=1, \x\in\F^n}|\langle\x^{\otimes d},\cS\rangle|=\max_{\|\x\|=1, \x\in\F^n}|\Re\langle\x^{\otimes d},\cS\rangle|
=\max_{\|\x\|=1, \x\in\F^n} |f(\x)|
\end{equation}
\end{enumerate} 
\item Define  on $\rP_{bi}(d,n,\F)$ the following inner product 
and the Hilbert norm:
\begin{eqnarray}\label{inprodbs}
\langle \f,\bg\rangle=\sum_{\bj,\bl\in J(d,n)} c(\bj)c(\bl)\overline{ f_{\bj,\bl}}g_{\bj,\bl}.
\end{eqnarray}
Then
\begin{enumerate}
\item\label{defejk} Let $\e_{\bj,\bl}=(\delta_{(\bj,\bl),(\bk,\bm)})_{(\bk,\bm)\in J(d,n)\times J(d,n)}, (\bj,\bl)\in J(d,n)\times J(d,n)$, where $\delta_{(\bj,\bl),(\bk,\bm)}$ is Kronecker's delta function, be the standard basis in $\rP_{bi}(d,n,\F)$.  Then $\frac{1}{\sqrt{c(\bj)c(\bl)})}\e_{(\bj,\bl)}, (\bj,\bl)\in J(d,n)\times J(d,n)$ is an orthonormal basis in $\rP_{bi}(d,n,\F)$ of cardinality ${n+d-1\choose d}^2$. 
\item The map $L:\rS^d\F^n\otimes\rS^d\F^n\to\rP_{bi}(d,n,\F)$ which is given by $L(\cT)=f$, where $f(\x,\y)=\langle \bar \cT,\x^{\otimes d}\otimes \y^{\otimes d}\rangle$, is an isomorphism and an isometry. 
\item The spectral norm of a bi-symmetric tensor is given by Banach's characterization \cite{Ban38}:
\begin{equation}\label{Bancharbi}
\begin{aligned}
&\|\cT\|_{\infty,\F}=\max_{\|\x\|=\|\y\|=1, \x,\y\in\F^n}|\langle\x^{\otimes d}\otimes\y^{\otimes d},\cT\rangle|=\max_{\|\x\|=\|\y\|=1, \x,\y\in\F^n}\Re\langle\x^{\otimes d}\otimes\y^{\otimes d},\cT\rangle=\\
&\max_{\|\x\|=\|\y\|=1, \x,\y\in\F^n} |f(\x,\y)|.
\end{aligned}
\end{equation}
\end{enumerate}
\item Define  on $\rP_{her}(d,n,\F)$ the inner product 
and the Hilbert norm as in \eqref{inprodbs}
\begin{enumerate}
\item  Let $\e_{\bj,\bl},\bj,\bl$ be defined as in (\ref{defejk}).  Then an orthonormal basis in $\rP_{her}(d,n,\F)$ is as follows.  First, to $\bj\in J(d,n)$  corresponds $\frac{1}{c(\bj)}\be_{\bj,\bj}$.  Second, to a pair $(\bj,\bl), (\bl,\bj),\bj\ne \bl$ corresponds $\frac{1}{\sqrt{2c(\bj)c(\bl)}}( \be_{\bj,\bl}+\be_{\bl,\bj})$.
The cardinality of this basis is ${n+d-1\choose d}^2$. 
\item The map $L:\rB\rH_{n^{\times d}}\to\rP_{bi}(d,n,\F)$ which is given by $L(\cT)=f$, where $f(\x,\bar\x)=\langle \bar \cT,\x^{\otimes d}\otimes \bar\x^{\otimes d}\rangle$, is an isomorphism and an isometry. 
\item The b-spectral norm of a bi-symmetric Hermitian tensor is: 
\begin{equation}\label{Bancharbih}
\|\cT\|_{bspec}=\max_{\|\x\|=1, \x\in\C^n} |f(\x,\bar\x)|=\max_{\|\x\|=1, \x\in\C^n} \max( f(\x,\bar\x), -f(\x,\bar\x)).
\end{equation}
Furthermore, if $\cT\in \rH_{n^{\times d},+}$ then
\begin{equation}\label{bspecinfeq}
\|\cT\|_{bspec}=\|\cT\|_{\infty,\C}
\end{equation}
\end{enumerate}
\end{enumerate}
\end{lemma}
\begin{proof}
\emph{(1)} Straightforward by repeating the arguments in \cite{FL18,FW20}.  (Observe that in \eqref{Banchar} for $\F=\C$, by changing the phase of $\x\in\C^n$ we can replace the $\pm$ sign to $+$.)

\noindent \emph{(2)}: \emph{(a)-(b)} straightforward. \emph{(c)}: We claim that for a bi-symmetric tensor $\cT\in \rS^d\F^n\otimes \rS^d\F^n$ the following equality holds:
\begin{equation}\label{cTinfeq}
 \|\cT\|_{\infty,\F}=\max_{\x,\y\in\F^n, \|\x\|=\|\y\|=1}|\langle \x^{\otimes d}\otimes  \y^{\otimes d},\cT\rangle|=\max_{\x,\y\in\F^n, \|\x\|=\|\y\|=1}\Re\langle \x^{\otimes d}\otimes  \y^{\otimes d},\cT\rangle.   
 \end{equation}
 Let 
\begin{equation*}
\|\cT\|_{\infty,\F}=\max_{\cX,\cY\in \Pi_{n^{\times d}}}|\langle \cX\otimes \cY,\cT\rangle|=|\langle \cX^{\star}\otimes\cY^{\star},\cT\rangle|,  \cX^\star=\otimes_{j=1}^d \x_j^\star, \cY^\star=\otimes_{j=1}^d\y_j^\star,\|\x_j^\star\|=\|\y_j^\star\|=1, j\in[d].
\end{equation*}
Let $\cS=[\sum_{bj=(j_1,\ldots,j_d)\in[n^{\times d}]} t_{\bi,\bj}\y_{j_1}^{\star}\cdots\y_{j_d}^{\star}]$.
Clearly,  $\cS\in\rS^d\F^n$.  Hence $\|\cT\|_{\infty,\F}=\|\cS\|_{\infty,\F}=|\langle \cX^{\star},\cS\rangle|$.  Use Banach's theorem to replace $\cX^\star$ by $(\x^{\star})^{\otimes d}$.   Set $\cS'=[\sum_{\bi=(i_1,\ldots,i_d)\in [n^{\times d}]}(\bar x^\star)_{i_1}\cdots (\bar x^{\star})_{i_d}t_{\bi,\bj}]$.    Use the arguments for $\cS'$ to deduce the first equality in \eqref{cTinfeq}.   The second equality in \eqref{cTinfeq}. is obtained by  changing the phase of $\x\in\C^n$.  Part \emph{(2b)} yields \eqref{Bancharbi}.

\noindent \emph{(3)}  The equality \eqref{bspecinfeq} follows from part \emph{(3)} of Lemma \ref{bisymspecnuclem}.  Other claims are straightforward.  (Recall that that $f(\x,\bar\x)=\langle \x^{\otimes d},\cT\x^{\otimes d}\rangle$ is real valued.)
\end{proof}

 Note the orthonormal basis in \emph{(1)} is corresponds to the Dicke basis \cite{Dic54}.  See \cite[Appendix 1]{FW20} for the computation of GME for the Dicke states.

The following lemma is well known and we bring its proof for completeness:
\begin{lemma}\label{prsde}  Let $2\1_d\le \bn\in\N^d$.  Assume that $\cX=\otimes_{j=1}^d\x_j, \cY=\otimes_{j=1}^d\y_j\in \Pi_{n^{\times d}}$.  Then
\begin{equation}\label{prsde1}
\|\otimes_{j=1}^d\x_j-\otimes_{j=1}^d\y_j\|\le \sum_{j=1}^d\|\x_j-\y_j\|.
\end{equation}
\end{lemma}
\begin{proof} Use the triangle inequality on the following identity.
\begin{equation}\label{telid}
\begin{aligned}
&\otimes_{j=1}^d\x_j-\otimes_{j=1}^d\y_j=
&\sum_{k=0}^{d-1}(\otimes_{l=1}^{d-k-1}\x_l)\otimes(\x_{d-k}-\y_{d-k})\otimes(\otimes_{d-k+1\le m\le d}\y_m).
\end{aligned}
\end{equation}
\end{proof}
\subsection{Approximations of spectral and nuclear norms}\label{subsec:aprsn}
Recall that a finite set $\rC\subset \mathbb{S}^{n-1}$ is an $\varepsilon$-covering set of $\mathbb{S}^{n-1}\subset \R^n$ if $\cup_{\x\in\rC} \rB(\x,\varepsilon)\supset  \mathbb{S}^{n-1}$.  Lemma 2.1 in \cite{Per04} gives an explicit simple $\frac{1}{m}$-covering set of $\mathbb{S}^{n-1}\subset \R^n$:
\begin{lemma}\label{PGcov}  Let $m,n\in\N$.  Define
\begin{equation}\label{PGcov1} 
\begin{aligned}
&\rA=\{\ba=(a_1,\ldots,a_n)^\top\in\R^n: a_i=\frac{h}{nm},  h=0,\pm 1,\ldots,\pm nm\},\\
&\rC(m,n)=\{\frac{1}{\|\ba\|}\ba: \ba\in A\setminus\{\0\}\}.
\end{aligned}
\end{equation}
Then $\rC(m,n)$ is an $\frac{1}{m}$-covering of $\mathbb{S}^{n-1}$ of cardinality $(2nm+1)^n-1$. 
\end{lemma}
Observe that $-\rC(m,n)=\rC(m,n)$.
For $\F=\C$ we view $\C^n\sim\R^n\times\R^n\sim\R^{2n}$, where $\z=\x+\bi\y=(\x,\y)\in \R^n\times \R^n$, and $\|\z\|=\sqrt{\|\x\|^2+\|\y\|^2}$.  Then $\rC_{\C}(m,n)\sim \rC(m,2n)$ corresponds to $\frac{1}{m}$-covering set of $\{\z\in\C^n: \|\z\|=1\}$,
We treat in this subsection all the norms that we discuss as real norms.  Our approximation of the norms we discuss are as real norms.

Define
\begin{equation}\label{defspecmn}
\begin{aligned}
&\|\cT\|_{\infty,\R,m,n}=\max_{\x\in\rC(m,n)}|\langle \x^{\otimes d},\cT\rangle|, \quad \cT\in \rS^{d}\R^n,\\
&\|\cT\|_{\infty,\C,m,n}=\max_{\z\in\rC_{\C}(m,n)}|\Re\langle \z^{\otimes d},\cT\rangle|, \quad \cT\in \rS^{d}\C^n,\\
&\|\cT\|_{\infty,\R,m,n}=\max_{\x,\y\in\rC(m,n)}|\langle \x^{\otimes d}\otimes\y^{\otimes d},\cT\rangle|, \quad \cT\in \rS^{d}\R^n\otimes\rS^{d}\R^n,\\
&\|\cT\|_{\infty,\C,m,n}=\max_{\z,\w\in\rC_{\C}(m,n)}|\Re\langle \z^{\otimes d}
\otimes\w^{\otimes d},\cT\rangle|, \quad \cT\in \rS^{d}\C^n\otimes\rS^{d}\C^n.
\end{aligned}
\end{equation}
\begin{lemma}\label{mnnormlem} Let $2\le n,d\in\N$ be fixed.  
\begin{enumerate}
\item
The quantities in \eqref{defspecmn} are real norms that satisfy the following inequalities for specified values of $m$:
\begin{equation}\label{specmnin}
\begin{aligned}
&\|\cT\|_{\infty,\R,m,n}\le \|\cT\|_{\infty,\R}\le \frac{1}{1-\frac{d}{m}}\|\cT\|_{\infty,\R,m,n} , \quad \cT\in \rS^{d}\R^n \textrm{ and } d<m,\\
&\|\cT\|_{\infty,\C,m,n}\le \|\cT\|_{\infty,\C}\le  \frac{1}{1-\frac{d}{m}}\|\cT\|_{\infty,\C,m,n},\quad \cT\in \rS^{d}\C^n \textrm{ and } d<m,\\
&\|\cT\|_{\infty,\R,m,n}\le \|\cT\|_{\infty,\R}\le  \frac{1}{1-\frac{2d}{m}} \|\cT\|_{\infty,\R,m,n}, \quad \cT\in \rS^{d}\R^n\otimes\rS^{d}\R^n  \textrm{ and } 2d<m,\\
&\|\cT\|_{\infty,\C,m,n}\le \|\cT\|_{\infty,\C}\le  \frac{1}{1-\frac{2d}{m}} \|\cT\|_{\infty,\C,m,n}, \quad \cT\in \rS^{d}\C^n\otimes\rS^{d}\C^n  \textrm{ and } 2d<m.
\end{aligned}
\end{equation}
\item  Assume that $d$ satisfy the above inequalities for the corresponding norms.
Denote by $\|\cT\|_{1,\R,m,n},  \|\cT\|_{1,\C,m,n}$ the corresponding real dual norms.  Then the following inequalities hold:
\begin{equation}\label{nucmnin}
\begin{aligned}
&(1-\frac{d}{m})\|\cT\|_{1,\R,m,n} \le\|\cT\|_{1,\R}\le\|\cT\|_{1,\R,m,n} , \quad \cT\in \rS^{d}\R^n \textrm{ and } d<m,\\
&(1-\frac{d}{m})\|\cT\|_{1,\C,m,n} \le\|\cT\|_{1,\C}\le\|\cT\|_{1,\C,m,n} \quad \cT\in \rS^{d}\C^n \textrm{ and } d<m,\\
&(1-\frac{2d}{m})\|\cT\|_{1,\R,m,n} \le\|\cT\|_{1,\R}\le\|\cT\|_{1,\R,m,n}, \quad \cT\in \rS^{d}\R^n\otimes\rS^{d}\R^n  \textrm{ and } 2d<m,\\
&(1-\frac{2d}{m})\|\cT\|_{1,\C,m,n} \le\|\cT\|_{1,\C}\le\|\cT\|_{1,\C,m,n}, \quad \cT\in \rS^{d}\C^n\otimes\rS^{d}\C^n  \textrm{ and } 2d<m.
\end{aligned}
\end{equation}
\end{enumerate}
\end{lemma}
\begin{proof}\emph{(1)}  As $\rC(m,n)\subset \mathbb{S}^{n-1}$ and $\rC_{\C}(m,n)\subset \{\z\in\C^n: \|\z\|=1\}$, use characterizations \eqref{defspecnuc},  \eqref{Banthm},  \eqref{cTinfeq} and \eqref{bspeceq} to deduce the first inequality in the five inequalities of \eqref{specmnin}.  We now show how to deduce the second inequality in the five inequalities of \eqref{specmnin}.    We start with the proof of the inequality $\|\cT\|_{\infty,\R}\le \frac{1}{1-\frac{d}{m}}\|\cT\|_{\infty,\R,m,n}$ for $\cT\in \rS^d\R^n$.
Assume that $\|\cT\|_{\infty,\R}=|\langle \x^{\otimes d},\cT\rangle|$ for $\x\in \mathbb{S}^{n-1}$.  Let $\y\in \rC(m,n)$ such that $\|\x-\y\|\le \frac{1}{m}$. Clearly,
\begin{equation*}
\begin{aligned}
&\|\cT\|_{\infty,\R}=|\langle \y^{\otimes d} +(\x^{\otimes d}-\y^{\otimes d}),\cT\rangle|\le |\langle \y^{\otimes d},\cT\rangle|+|\langle \x^{\otimes d}-\y^{\otimes d},\cT\rangle|\le\\
&\|\cT\|_{\infty,\R,m,n}+|\langle \x^{\otimes d}-\y^{\otimes d},\cT\rangle|.
\end{aligned}
\end{equation*}
Assume that $\x\ne \y$,  and let $(\x-\y)=\|\x-\y\| \z$, where $\z=\frac{1}{\|\x-\y\|}(\x-\y)$.
Use Banach's characterization of $\|\cT\|_{\infty,\R}$ and the equality \eqref{telid} to deduce $ |\langle \x^{\otimes d}-\y^{\otimes d}),\cT\rangle\le d\|\x-\y\|\|\cT\|_{\infty,\R}$.    That is,  $\|\cT\|_{\infty,\R}\le \|\cT\|_{\infty,\R,m,n}+\frac{d}{m}\|\cT\|_{\infty,\R}$.
Assume that  $d<m$ and deduce the second inequality of the first set of inequalities 
in \eqref{specmnin}.   
To deduce the inequality
 $\|\cT\|_{\infty,\C}\le  \frac{1}{1-\frac{d}{m}}\|\cT\|_{\infty,\C,m,n}$ for $ \cT\in \rS^{d}\C^n$ and $d<m$ we  assume that $\|cT\|_{\infty,\C}=\langle \x^{\otimes d},\cT\rangle$ for some $\x\in\C^n, \|\x\|=1$.  Suppose that $\y\in \rC_{\C}(m,n)$, and $\|\x-\y\|\le 1/m$.  Then 
\begin{equation*}
 \begin{aligned}
 &\|\cT\|_{\infty,\C}=\Re \langle \y^{\otimes d},\cT\rangle +\Re\langle  \x^{\otimes d}- \y^{\otimes d},\cT\rangle\le\|\cT\|_{\C,m,n}+|\langle  \x^{\otimes d}- \y^{\otimes d}|,\cT\rangle|\le\\
&\|\cT\|_{\C,m,n}+\frac{d}{m}\|\cT\|_{\infty,\C}\Rightarrow \|\cT\|_{\infty,\C}\le  \frac{1}{1-\frac{d}{m}}\|\cT\|_{\infty,\C,m,n} \textrm{ if } d<m.
\end{aligned}
 \end{equation*}

 Next, we point out briefly how to deduce the inequality $\|\cT\|_{\infty,\R}\le  \frac{1}{1-\frac{2d}{m}} \|\cT\|_{\infty,\R,m,n}$ for $\cT\in \rS^{d}\R^n\otimes\rS^{d}\R^n $  and  $2d<m$.  First, use Banach's theorem to deduce that $\|\cT\|_{\infty,\R}=|\langle \x^{\otimes d}\otimes\tilde\x^{\otimes d},\cT\rangle$ for some $\x,\tilde\x\in \mathbb{S}^{n-1}$.   Assume that $\y,\tilde\y\in \rC(m,n)$ such that $\|\x-\y\|,\|\tilde\x-\tilde \y\|\le \frac{1}{m}$.   Hence, 
 $$\|\cT\|_{\infty,\R}\le \|\cT\|_{\infty,\R,m,n}+|\langle \x^{\otimes d}\otimes\tilde\x^{\otimes d}-\y^{\otimes d}\otimes\tilde\y^{\otimes d},\cT\rangle|.$$
Use the equality \eqref{telid} and the maximum characterization of $\cT$ to deduce the inequality $\|\cT\|_{\infty,\R}\le \|\cT\|_{\infty,\R,m,n}+\frac{2d}{m}\|\cT\|_{\infty,\R}$.  Assuming that $2d<m$ we deduce the inequality 

\noindent
$\|\cT\|_{\infty,\R}\le  \frac{1}{1-\frac{2d}{m}}\|\cT\|_{\infty,\R,m,n}$.   The inequality for $\F=\C$ deduced similarly to the case $\cT\in\rS^{d}\C^n$.
Thus, under the assumptions that $m$ satisfies specified inequalities we that the quantities in \eqref{defspecmn} are real norms. 

\noindent\emph{(2)} We first prove the first set of inequalities in \eqref{nucmnin}.
Since  $\rS^d\R^n$ is a vector space over $\R$,  the norms $\|\cdot\|_{\infty,\R}$ and 
$\|\cdot\|_{1,\R}$ are real dual norms.  The inequality$\|\cT\|_{\infty,\R,m,n}\le \|\cT\|_{\infty,\R}$,  and the maximum characterization of the dual norms \eqref{defnustar} yields the inequality $\|\cT\|_{1,\R,m,n}\ge \|\cT\|_{1,\R}$.   Assume that $\nu$ and $\nu^*$ are dual norm.  Let $t>$ and define a new norm $\nu_1=t\nu$.  Then $\nu_1^*=t^{-1}\nu$.  Hence, the second inequality of the first set of inequalities in \eqref{nucmnin} yield that $(1-d/m)\|\cT\|_{1,\R,m,n}\le \|\cT\|_{1,\R}$.

Assume now that $\cT\in\rS^d\C^n$.  The definition \eqref{defnustar} of $\nu^*$ implies that $\|\cT\|_{1,\C}$ is a real dual norm of $\|\cT\|_{\infty,\C}$.  Therefore,  the previous arguments yield the second set of the inequalities in \eqref{nucmnin}.  Similar arguments yield the rest of the inequalities in \eqref{nucmnin}. 
\end{proof}
\subsection{Proof of Theorem \ref{polthm}}\label{subsec:prfpt}
We first show that the norm $\|\cT\|_{\infty,\F}$ is 
$\varepsilon$-approximable in poly-time in  the input $\cT\in \rS^d\Q[\sqrt{-1}]^n$, $\frac{1}{\varepsilon}$ and $d$, where $n$ is fixed and $\varepsilon\in\Q\cap (0,1)$.
We start with the case $\F=\R$ and $\cT\in \rS^d\R^n$.  Choose 
\begin{equation}\label{mchoic}
m=\lceil\frac{d(1+\varepsilon)}{\varepsilon}\rceil, \quad \varepsilon\in\Q\cap (0,1).
\end{equation}
Recall that the cardinality $\rC(m,n)$ is of order $(2mn+1)^n=O\left(\big(\frac{d}{\varepsilon}\big)^n\right)$.
We need to estimate the number of flops and the size of storage to compute $\langle \bu^{\otimes d},\cT\rangle$ for $\bu=(u_1,\ldots,u_d)^\top\in \rC(m,n)$.   We code the tensors $\bu^{\otimes d},\cT\in\rS^d\Q^n$ by the coefficients of polynomials $f(\x)$ and $g(\x)$ given by \eqref{defpolfx}.   We assume that the  entries of the tensor $\cT$ is given by $g_{\bj},\in J(d,n)$.
The coefficient $f_{\bi}, \bi=(i_1,\ldots,i_n)\in J(d,n)$ is $u_1^{i_1}\cdots u_n^{i_n}$.
Hence $\langle{\bu^{\otimes d},\cT}=\langle \bbf,\bg\rangle$ as in \eqref{inprodFJ}, and $\langle\bu^{\otimes d},\cT\rangle=\bg(\bu)$.  Observe that $|{n+d-1\choose d}|=|{n+d-1\choose n-1}|=O(d^{n-1})$ for a fixed $n$ and $d\gg 1$.
Recall that the computation of $c(\bj)=\frac{d!}{j_1!\cdots j_n!}$ needs $2d$ multiplications and one division.  Furthermore $c(\bj)\le d!=O(2^{d(\log_2 d}-\log_2e)$.  Hence, computing $\langle \bu^{\otimes d},\cT\rangle, \bu\in\rC(m,n)$ is polynomial in the storage entries of $f_{\bj}\in\Q$ and $\log m$.  However, to compute the first maximum in \eqref{defspecmn} one needs $O\left(\big(\frac{d}{\varepsilon}\big)^n\right)$ computations of $\langle \bu^{\otimes d},\cT\rangle$.   This shows that the computation of $\|cT\|_{\infty,\R,m,n}$ is polynomial in  the input $\cT\in \rS^d\Q^n$, $\frac{1}{\varepsilon}$ and $d$.  Polynomiality in $\varepsilon$ is $\varepsilon^{-n}$.
Our assumption \eqref{mchoic} combined with the first set of inequalites in \eqref{specmnin} yields
\begin{equation}\label{cTinfap}
\|\cT\|_{\infty,\R,m,n}\le\|\cT\|_{\infty,\R}\le (1+\varepsilon)\|\cT\|_{\infty,\R,m,n}.
\end{equation}
Similar results hold for other spectral norms in Theorem  \ref{polthm}.
It was shown in \cite{FW20} that  the GME of symmetric states is poly-time computable in $d$ for a fixed value of $n\ge 2$,  using poly-time solvability of certain systems of polynomial equations.

We now discuss poly-time approximation of the nuclear norms in Theorem  \ref{polthm}.   Let us first discuss $\varepsilon$-approximation of $\|\cT\|_{1,\R}$ by $\|\cT\|_{1,\R,m,n}$ for $\cT\in\rS^d\R^n$.  The condition \eqref{mchoic} combined with the first set of inequalities in \eqref{nucmnin} yields:
\begin{equation}\label{cT1ap}
(1-\varepsilon)\|\cT\|_{1,\R,m,n}\le\|\cT\|_{1,\R}\le \|\cT\|_{1,\R,m,n}.
\end{equation}
Since we showed that $\|\cT\|_{\infty,\R,m,n}$ is poly-time computable,  Theorem 3 in \cite{FL16} yields that $\|\cT\|_{1,\R,m,n}$ is poly-time computable.
Hence, $\varepsilon$-approximation  of $\|\cT\|_{1,\R}$ is poly-time computable.
Similar results hold for other nuclear norms in Theorem  \ref{polthm}.\\

In what follows, we show directly how to compute the nuclear norms in Theorem  \ref{polthm} using linear programming.  Our approach is similar to \cite{Per04}.
We restrict our discussion to the norm $\|\cdot\|_{1,\C,m,n}$.  Recall that we view  $\|\cdot\|_{\infty,\C,m,n}$ and  $\|\cdot\|_{1,\C,m,n}$ as norms over $\R$.
The characterization of $\|\cdot\|_{\infty,\C,m,n}\|$ in \eqref{defspecmn} for $2d<m$ yields that the extreme points the norm  $\|\cdot\|_{\infty,\C,m,n}$ on $\rS^d\C^n\otimes\rS^d\C^n$ are $\pm \z^{\otimes d}\otimes \w^{\otimes d}$ for $\bz,\bw\in\rC_{\C}(m,n)$.    Note, that from the definition of $\rC_{\C}(m,n)$ it follows that $-\rC_{\C}(m,n)=\overline{\rC_{\C}(m,n)}=\rC_{\C}(m,n)$.
Assume that $\cT\in(\rS^d\C^n\otimes \rS^d\C^n)\cap (\rS^d\Q[\sqrt{-1}]^n\otimes   \rS^d\Q[\sqrt{-1}]^n)$.  That is,  $\cT$ has Gaussian rational entries.   To $\cT$ and $\bz,\bw\in\rC_{\C}(m,n)$ we correspond polynomials $\bbf(\z,\bw)$ given in \eqref{defpolfxy}. 
\begin{lemma}\label{lpBH} Let $\cT\in(\rS^d\C^n\otimes \rS^d\C^n)\cap (\rS^d\Q[\sqrt{-1}]^n\otimes   \rS^d\Q[\sqrt{-1}]^n)$.   Assume that to tensors $\cT$ and $\bz^{\otimes d}\otimes \bw^{\otimes d}$ correspond vectors $\bbf$ and $\bg_{\z,\bw}$ as in \eqref{defpolfxy}. 
Suppose furthermore that $2d<m$. Then $\|\cT\|_{1,m,n}$ is the solution of the following linear programming problem:
\begin{equation}\label{lpBH1}
\|\cT\|_{1,m,n}=\min_{\bbf=\sum_{\z,\bw\in\rC_{\C}(m,n)} a_{\z,\bw}\bg_{\z,\bw}-b_{\z,\bw}\bg_{\z,\bw}, a_{\z,\bw},a_{\z,\bw}\ge 0} a_{\z,\bw}+b_{\z,\bw}
\end{equation}
\end{lemma} 
\begin{proof}  It suffices to prove the above characterization where $\|\cT\|_{1,m,n}=1$.  As the extreme points of the unit ball $\|\cX\|_{1,m,n}\le 1$ are $\pm \z^{\otimes d}\otimes \w^{\otimes d},\z,\bw\in\rC_{\C}(m,n)$,  $\cT$ is a convex combination of $\pm \z^{\otimes d}\otimes \w^{\otimes d},\bw\in\rC_{\C}(m,n)$.  That is 
$$\bbf=\sum_{\z,\bw\in\rC_{\C}(m,n)} a'_{\z,\bw}\bg_{\z,\bw}-b'_{\z,\bw} \bg_{\z,\bw}, a'_{\z,\bw},b'_{\z,\bw}\ge 0,  \sum_{\z,\bw\in\rC_{\C}(m,n)} a'_{\z,\bw}+b'_{\z,\bw}=1.$$
Assume that $\bbf=\sum_{\z,\bw\in\rC_{\C}(m,n)} a_{\z,\bw}\bg_{\z,\bw}-b_{\z,\bw} \bg_{\z,\bw}, a_{\z,\bw},b_{\z,\bw}\ge 0$.
As $\|\pm \bg_{\z,\bw}\|_{1,m,n}=1$ the triangle inequality yields 
$1=\|\bbf\|_{1,m,n}\le\sum_{\z,\bw\in \rC_{\C}(m,n)}a_{\z,\bw}+b_{\z,\bw}=1$.
\end{proof}

We now estimate the number of flops to find $\cT\|_{1,\C,m,n}$.  Let us recall the complexity result in \cite{LS15}.  Assume that we have
a linear programming problem of the form $min_{A\x=\bb, \x\ge 0}\bc^\top \x$, where $A\in \Q^{M\times N},\bc\in \Q^N $.   Then the arithmetic complexity, (the number of flops),  is of order $O(\sqrt{M}(MN+M^2))$.
For a given $\varepsilon\in (0,1)$ we want to have the inequality
\begin{equation}\label{cT1ap1}
(1-\varepsilon)\|\cT\|_{1,\C,m,n}\le\|\cT\|_{1,\C}\le \|\cT\|_{1,\C,m,n}.
\end{equation}
Use the last set of inequalities in \eqref{nucmnin} to deduce that
\begin{equation*}
m=\lceil\frac{2d}{\varepsilon}\rceil=O(\frac{d}{\varepsilon}).
\end{equation*}
Observe next that for the minimum problem \eqref{lpBH1} the values of $M$ and $N$ are as follows:
\begin{equation*}
M=2{n+d-1\choose d}^2=O(d^{2(n-1)}), \quad N=\big((2m(2n)+1)^{2n}-1\big)^2=O(\frac{d^{2(3n-1)n}}{\varepsilon^{4n}}).
\end{equation*}
Hence, the arithmetic complexity of finding $\|\cT\|_{1,\C,m,n}$ that satisfies \eqref{cT1ap1} is $O(\frac{d^{3(n-1)+2(3n-1)n}}{\varepsilon^{4n}})$.

Use the equality \eqref{bspecinfeq} to deduce
\begin{corollary} Assume that $\cR\in \rB\rH_{n^{\times d},+,1}\cap\otimes^{2d}(\Q[\sqrt{-1}])^n$, where $n\ge 2$ is a fixed integer.  For a given $\varepsilon\in\Q_+\cap (0,1)$ one can decide in poly-time in $d$ if the distance of $\cT$ to strongly separable density tensors less than $\varepsilon$.  More precisely,  the arithmetic complexity of this problem is of order $O(\frac{d^{3(n-1)+2(3n-1)n}}{\varepsilon^{4n}})$.
\end{corollary}
\section{Open problems}\label{sec:oprb}
\emph{Open problems}
For the following problems assume that $2\le d,n$.
\begin{enumerate}
\item Is $\|\cB\|_{\infty,\C}=\|\cB\|_{spec}$ for all $\cB\in\rH_{\bn}$?
\item Is $\|\cB\|_{1,\C}=\|\cB\|_{aspec}$ for all $\cB\in \rA\rH_{n^{\times d}}$ for $d<n$?
\item 
Assume that $d=\lceil\frac{n}{2}\rceil$.  Is the computation of $\|\cB\|_{\infty,\C}$ is NP-hard in $n$ for $\cB\in \rA\rH_{n^{\times d}}$.   Is the assertion that $\cB\in\rA\rH_{n^{\times d},+,1}$ is within $\varepsilon$-distance from $\mathrm{Sepa}_{n^{\times d}}$ NP-hard?
\item Is $\|\cT\|_{bspec},\|\cT\|_{bnuc}$ are poly-time computable for $\cT\in \rH_{n^{\times d}}\cap \otimes^{2d}(\Q[\sqrt{-1}])^n$?
\item Can one give better upper bounds for the nuclear ranks in\eqref{defmaxnrank} than those given by Caratheodory's bound? 
\end{enumerate}
\section*{Acknowledgment} 
\noindent
The author thanks Otfried G\"{u}hne for his useful comments, and the two referees for their remarks and comments.

\noindent
Part of this work was supported by the Simons Institute for the Theory of Computing, and conducted when the author was visiting the Institute in September 2025.

\section*{Data availability} Data sharing is not applicable to this article as no datasets were generated or analyzed during the current study.

\section*{Declarations Conflict of interest} The author declares that there is no conflict of interest.

\bibliographystyle{plain}

\begin{thebibliography}{MMM}
\bibitem{AMM10} M. Aulbach, D. Markham and Mi. Murao,The maximally entangled symmetric state in terms of the geometric measure, \emph{New Journal of Physics}, 12 (2010), pp. 073025. 	

\bibitem{Ban38} S. Banach, \"Uber homogene Polynome in ($L^2$), \emph{Studia Math.}, \textbf{7} (1938), pp.~36--44.

\bibitem{BL01} H. Barnum, N. Linden,Monotones and invariants for multi-partite quantum states, J. Phys. A 34 (2001), no.35, 6787--6805.

\bibitem{Bel64} J. S.  Bell, On the Einstein-Podolsky-Rosen paradox, Physics Vol. 1, No. 3, pp. 195—200, 1964.

 \bibitem{BCR98} J. Bochnak, M. Coste and M.F. Roy,
 \emph{Real algebraic geometry}, Ergebnisse der Mathematik und ihrer Grenzgebiete (3), 36, Springer-Verlag, Berlin, 1998.

\bibitem{BFZ} W. Bruzda, S.Friedland, K. {\.Z}yczkowski, Tensor rank and entanglement of pure quantum states,  {\it Linear and Multilinear Algebra} 72 (2024), no. 11, 1796-1859. 

\bibitem{CHLZ12} B. Chen, S. He, Z. Li, and S. Zhang, Maximum block improvement and polynomial optimization, \emph{SIAM J. OPTIM.} 22 (2012),  87--107.
\bibitem{CW02} K. Chen and L.-A. Wu,A matrix realignment method for recognizing entanglement, arXiv:quant-ph/0205017, 2002.

\bibitem{CCQS} L. Chen,  D. Chu, L. Qian,and Y. Shen,  Separability of completely symmetric states in a multi-partite system,   Phys. Rev.  A 99, 032312 (2019).

\bibitem{CXZ10} L. Chen,  A. Xu and H. Zhu,  Computation of the geometric measure of entanglement for pure multiqubit states, \emph{Physical Review A},  82, 032301, 2010.
 
\bibitem{Cho09} E.Cho. Inner product of random vectors,{\em Int. J. Pure Appl. Math.}, 56(2):217--221, 2009.

\bibitem{Col75}G E. Collins: Quantifier elimination for the elementary theory of real closed fields by cylindrical algebraic decomposition, Second GI Conf. Automata Theory and Formal Languages, Springer LNCS 33, 1975.

\bibitem{DIG25} S. Denker, S.  Imai and O.  G\"uhne, Chiral symmetries and multiparticle entanglement, arXiv:2506.15609, 2025.

\bibitem{Der16} H. Derksen, On the nuclear Norm and the singular value
decomposition of tensors,  Found. Comput. Math. (2016) 16:779–811.

\bibitem{DFLW17}  H. Derksen,  S. Friedland, L.-H. Lim, and L. Wang, Theoretical and computational aspects of entanglement,  arXiv:1705.07160.

\bibitem{Dic54}  R. H. Dicke, Coherence in Spontaneous Radiation Processes, \emph{Physical Review.} 93 (1) (1954), 99--110.
\bibitem{DPS02}  A.C.  Doherty, P.A.  Parrilo and F.M. Spedalieri, Distinguishing separable and entangled states,Phys. Rev. Lett. 88, 187904, 2002.
\bibitem{DPS04} A.C.  Doherty, P.A.  Parrilo and F.M. Spedalieri,  A complete family of separability criteria,  Phys. Rev. A 69, 022308,  2004.
\bibitem{DPS05} A.C.  Doherty, P.A.  Parrilo and F.M. Spedalieri,  Detecting multi-partite entanglement,  Phys. Rev. A 71, 032333, 2005.

\bibitem{DVC00} W. D\"ur, G. Vidal and J.I. Cirac, Three qubits can be entangled in two inequivalent ways, \emph{Phys. Rev. A.} 62 (2000), 062314.

\bibitem{Frib} S. Friedland, \emph{Matrices: Algebra, Analysis and Applications}, World Scientific, 596 pp., 2016, Singapore.ree qubits can be entangled in two inequivalent ways, \emph{Phys. Rev. A.} 62 (2000), 062314.

\bibitem{EW01} T. Eggeling and R.F. Werner, Separability properties of tripartite states with $U\otimes U\otimes U$ symmetry,  \emph{Phys. Rev. A} 63, (2001), 042111.

\bibitem{EPR35} A. Einstein, B. Podolsky, and N. Rosen N, Can Quantum-Mechanical Description of Physical Reality Be Considered Complete?,  \emph{Phys. Rev.} 47, is. 10, (1935),  777--780.
\bibitem{EB01} J. Eisert and H.-J. Briegel,The Schmidt measure as a tool for quantifying multi-particle entanglement,  Phys. Rev. A 64, 022306 , 2001.


\bibitem{Fri13} S. Friedland, Best rank one approximation of real symmetric tensors can be chosen symmetric,  \textit{Front. Math. China}, 8 (1) (2013), 19--40.

\bibitem{FK18} S. Friedland and T. Kemp,  Most Boson quantum states are almost maximally entangled,  \emph{Proceedings of Amer. Math. Soc.} 146, No.12, (2018), 5035--5049.

\bibitem{FL16} S. Friedland and L.-H. Lim,  The computational complexity of duality, \emph{SIAM Journal on Optimization}, 26, No. 4 (2016), 2378--2393.

\bibitem{FL18}  S. Friedland and L.-H. Lim, Nuclear norm of higher-order tensors,  \emph{Mathematics of Computation}, 87 (2018), 1255--1281.

\bibitem{FL20} S. Friedland and L.-H. Lim,  Symmetric Grothendieck inequality,  arXiv:2003.07345.

\bibitem{FW20} S. Friedland and L. Wang, Spectral norm of a symmetric tensor and its computation,  \emph{Mathematics of Computation}, 89 (2020), 2175--2215.

\bibitem{Gha10} S. Gharibian, Strong NP-hardness of the Quantum separability problem, \textit{Quantum Inf.\ Comput.}, \textbf{10} (2010), no.~3--4,  pp.~343--360.

\bibitem{GV13} G.H. Golub and C.F.  Van Loan, Johns Hopkins,  Matrix Computations,
 4-th edition, 2013.

\bibitem{GFE09} D. Gross, S. T. Flammia, and J. Eisert, Most Quantum States Are Too Entangled To Be Useful As Computational Resources, \emph{Phys. Rev. Lett.} 102, 190501, 2009.

\bibitem{Gro55} A. Grothendieck, Produits tensoriels topologiques et espaces nucl\'eaires, Mem. Amer. Math. Soc. No. 16 (1955), 140.

\bibitem{GB02} L.Gurvits and H.Barnum, Largest separable balls
around the maximally mixed bi-partite quantum
state," Phys. Rev. A 66 062311 (2002).

\bibitem{Gur03}  L.~Gurvits, Classical deterministic complexity of Edmonds problem and quantum entanglement,\textit{Proc.\ ACM Symp.\ Theory Comput.} (STOC), \textbf{35}, pp.~10--19, New York, NY, ACM Press, 2003.

\bibitem{Ha21} K.-C.  Ha,  Comment on “Separability of completely symmetric states in a multi-partite system'',   Phys. Rev.  A 104, 016401 (2021).

\bibitem{Has90} J. H\aa stad, Tensor rank is NP-complete, \emph{J. Algorithms}, 11 (1990),  644--654.

\bibitem{Har92} J. Harris, \emph{Algebraic Geometry: A First Course}, Springer, 1992.

\bibitem{HS00} A.Higuchi, A. Sudbery, How entangled can two couples get? \emph{Physics Letters A}, 273(4) (2000), pp. 213-217. 

\bibitem{HL13} C.J. Hillar and L.-H. Lim, Most tensor problems are NP-hard, \emph{J.\ Assoc.\ Comput.\ Mach.}, 60 (2013), no.~6, p.~45.

\bibitem{HJ13} R.A.  Horn and C.R. Johnson,  \emph{Matrix Analysis},  Cambridge University Press, 2013.

\bibitem{Hor96} M. Horodecki, P. Horodecki, and R. Horodecki, Separability of mixed states: necessary and sufficient conditions, \emph{J. Physics Letters A.} 223 (1996), 1--8.
\bibitem{Hor97} P. Horodecki,  Separability criterion and inseparable mixed states
with positive partial transposition,  Phys.  Lett. A 232 ( 1997), 333--339. 

\bibitem{Hubetall09} R. H\"ubener, M. Kleinmann, T.-C. Wei, C. Gonz\'alez-Guill\'en, and O. G\"uhne, The geometric measure of entanglement for symmetric states, \emph{Phys. Rev. A} 80 (2009), 032324. 

\bibitem{Lan12} J.M. Landsberg,  \emph{Tensors: geometry and applications}, Graduate Studies in Mathematics, 128, American Mathematical Society, Providence, RI, 2012. xx+439 pp.

\bibitem{LS15} Y. T.  Lee and A. Sidford, Efficient inverse maintenance
and faster algorithms for linear programming.
In FOCS, pages 230-249, IEEE Computer Society,2015.

\bibitem{M-S25} C. Marconi, G. M\"uller-Rigat, J. Romero-Pallej\`a, J. Tura and A. Sanpera, Symmetric quantum states: a review of recent progress,  Rep. Prog. Phys. 89 024001,  2026.

\bibitem{MSE21} Mathematics stack exchange:  Set of unit ball’s extreme points is not always closed, 

\noindent
https://math.stackexchange.com/questions/3970669/set-of-unit-ball-s-extreme-points-is-not-always-closed

\bibitem{MS65} T. S. Motzkin and E. G. Straus, Maxima for graphs and a new proof of a theorem of Tur\'an,Canad. J. Math. 17 (1965), 533–540. 

\bibitem{Per96} A. Peres, Separability Criterion for Density Matrices, \emph{Phys. Rev. Lett.} 77 (1996), 1413--1415.

\bibitem{Per04} D. P\'erez-Garcia,  Deciding separability with a fixed error,  Physics Letters A 330 (2004), 149-154.

\bibitem{Roc70} R. T. Rockafeller, Convex Analysis, Princeton Univ. Press 1970.

\bibitem{Rud00}  O. Rudolph,  A separability criterion for density operators
Journal of Physics A,  33 (2000),  no. 21, 3951–3955.

\bibitem{QC19} L. Qian and D. Chu,  Decomposition of completely symmetric states, Quantum Information Processing 18.7 (2019) \#208,  46 pp.

\bibitem{STV98} A. Sanpera, R. Tarrach and G. Vidal,  Quantum inseparability as local pseudomixture,  arXiv:quant-ph/9801024, 1998.
\bibitem{Sch50} R. Schatten, A Theory of Cross-Spaces,
Ann. of Math. Stud., No. 26,
Princeton University Press, Princeton, NJ, 1950.

\bibitem{Sch35} E. Schr\"{o}dinger, Discussion of probability relations between separated systems, \emph{Mathematical Proceedings of the Cambridge Philosophical Society}, 31 is. 4, (1935),  555--563. 

\bibitem{Sch36} E. Schr\"{o}dinger, Probability relations between separated systems,  \emph{Mathematical Proceedings of the Cambridge Philosophical Society}, 32,  is. 3, (1936), 446--452.

\bibitem{Shi95} A. Shimony,  Degree of entanglement, Ann. New York Acad. Sci., 755 
(1995), 675–679.
\bibitem{SG22} J. Steinberg and O. G\"uhne,  Finding maximal quantum resources, arXiv:2210.13475 [quant-ph].

\bibitem{TWP09} S. Tamaryan, T.-C. Wei, D. Park, Maximally entangled three-qubit states via geometric measure of entanglement, \emph{Phys. Rev.} A 80 (2009), 052315.

\bibitem{TG09} G. Toth and O.  G\"{u}hne,  Entanglement and permutational symmetry,
\emph{Phys.Rev.Lett.}  102 (2009),  170503.

 
 \bibitem{WG03} T.-C. Wei and P.M. Goldbart, Geometric measure of entanglement and applications to bi-partite and multi-partite quantum states, \emph{Phys. Rev.} A 68 (2003), 042307.
\bibitem{WG25} L.T. Weinbrenner and O. G\"uhne,  Quantifying entanglement from the geometric perspective,  2025 EPL 151 68001.

\bibitem{Wer89} R. F.  Werner, Quantum states with Einstein-Podolsky-Rosen correlations admitting a hidden-variable model,  Physical Review A. 40 (8): 4277-4281, 1989.


\bibitem{ZCH10} H. Zhu, L. Chen and M. Hayashi,  Additivity and non-additivity of multi-partite entanglement measures,  \emph{ New J. Phys.}12 (2010), 083002 (41pp.)

\end{thebibliography}

 \appendix 
 \section{Comparison of general and Euclidean norms}\label{sec:eucgen}
 Let $\F$ be either the field of the real numbers $\R$ or complex numbers $\C$.  For $\x=(x_1,\ldots,x_n)^\top\in \F^n$ we define $\bx^*=(\bar x_1,\ldots,\bar x_n)$.   (Note that for $\x\in\R^n$ we have $\x^*=\x^\top=(x_1,\ldots,x_n)$.)  The inner product in $\F^n$ is given by $\langle \x,\y\rangle =\x^*\y$.  The Euclidean norm is given by $\|\x\|=\sqrt{\langle\x,\x\rangle}$.   A set $\rK\subset \F^n$ is convex,
 if for any $\x,\y\in\rK$ the interval $[\x,\y]:=\{\z=t\x+(1-t)\y: 0\le t\le1\}$ is in $\rK$. 
 Let $\Delta^m\subset \R^m$ be the simplex of probability vectors.   Clearly,  $\Delta^m$ is a convex set.
 For a given set $\rS\subset \F^n$,  the set cospan$(\rS):=\{\x=\sum_{i=1}^m a_i\x_i, (a_1,\ldots,a_m)\in\Delta^m, \x_i\in \rS, i\in[m], m\in\N\}$ is a convex set spanned by $\rS$.  We say that $\rS$ \emph{cospans} the set cospan$(\rS)$.
 A point $\z\in \rK$ is called an extreme point of $\rK$ if for any $\x,\y\in\rK$ such that $\z\in[\x,\y]$ the equality $\x=\y=z$ holds.  
 
 Assume that $\rK\subset \R^n$ is a compact convex set.   Minkowski's theorem states that the set of extreme points of $K$, denoted as $\rE(\rK)$, is the minimal set of $\rK$ that cospans $\rK$ \cite[page 427]{Roc70}.   
 Assume that $\F=\R$.   Caratheodory's theorem \cite{Roc70} states that every $\x\in\rK$ is a convex combination of at most $n+1$ extreme points.  Assume that $\F=\C$.   By viewing $\C^n$ as $\R^{2n}$ we deduce that every $\x\in\rK\subset \C^n $ is cospanned by at most $2n+1$ extreme points.
 
 Assume that $\nu:\F^n\to \R_+:=[0,\infty)$ is a norm on $\F^n$.  That is
\begin{equation*}
\begin{aligned}
&\nu(\x)=0 \iff \x=\0,\\
&\nu(a\x)=|a|\nu(\x) \textrm{ for }a\in\F, \x\in\F^n,
&\nu(\x+\y)\le \nu(\x)+\nu(\y).
\end{aligned}
\end{equation*}
Denote $B_{\nu}(\ba,r)=\{\x\in\F^n: \nu(\x-a)\le r\}$.  Let $\rB(\ba,r):=\{\x\in\F^n:\, \|\x-\ba\|\le r\}$.   
 
The dual norm $\nu^*$ on $\F^n$  is given by 
\begin{equation}\label{defnustar}
\nu^*(\x)=\max_{\nu(\y)\le 1} \Re \langle \y,\x\rangle.  
\end{equation}
Recall that $\nu^{**}=\nu$,  and $\|\cdot\|^*=\|\cdot\|$.   

Denote by $\rE(\nu)$ the extreme points of the ball $\rB_{\nu}(\0,1)$.  Observe that $\rE(\nu)$ is a balanced set: $\zeta \rE(\nu)=\rE(\nu)$ for $\zeta\in\F,\|\zeta|=1$.  Hence,  \eqref{defnustar} is equivalent to
\begin{equation}\label{defnustar1}
\nu^*(\x)=\sup_{\y\in\rE(\nu)} \Re \langle \y,\x\rangle.  
\end{equation}
(Recall that the set $\rE(\nu)$ may not be compact \cite{MSE21}.)
Vice versa,  assume $\rE\subset \F^n$ is a bounded balanced set that spans $\F^n$.  Let $\nu$ be a norm whose unit ball is the closure of cospan$(\rE)$.   Then $\nu^*$ is given by  
\begin{equation}\label{defnustar2}
\begin{aligned}
&\nu^*(\x)=\sup_{\y\in\rE} \Re \langle \y,\x\rangle,\\
&\rE(\nu)\subseteq \textrm{closure}(\rE).
\end{aligned}
\end{equation}

The following lemma is well known and we bring its proof for completeness:
\begin{lemma}\label{dualem}
Let $\nu$ be a norm on $\F^n$.  Then
\begin{equation}\label{dualem1}
\|\y\|^2\le \nu(\y)\nu^*(\y).
\end{equation}
Assume that the linear functional $\varphi_{\y}: \F^n\to \R$ is given by $\varphi_{\y}(\x)=\Re\langle \y,\x\rangle$. 
Equality holds in \eqref{dualem1} holds  for $\y\ne\0$  if and only one of the following two equivalent conditions hold.   
\begin{enumerate} 
\item  The inequality $\varphi_{\y}(\x)\le \nu(\y)\nu^*(\y)$ is a supporting hyperplane of $\rB_{\nu^*}(\0,\nu^*(\y))$ at $\y$.
\item  The inequality $\varphi_{\y}(\x)\le \nu(\y)\nu^*(\y)$ is a supporting hyperplane of $\rB_{\nu}(\0,\nu(\y))$ at $\y$.
\end{enumerate}
\end{lemma}
\begin{proof}
Assume that $\y\ne \0$, and set $\y_1=\frac{1}{\nu(\y)}\y$.   The characterization \eqref{defnustar} yields that $\nu^*(\y)\ge \Re\langle \y_1,\y\rangle=\frac{\|\y\|^2}{\nu(\y)}$. This establishes the inequality \eqref{dualem1}.

We now consider the equality case in \eqref{dualem1}.   $(1)$. Suppose first that the equality case holds in \eqref{dualem1}.   Then $\varphi_{\y}(\y)=\nu(\y)\nu^*(\y)$.   Assume that $\x\in \rB_{\nu^*}(\0,\nu^*(\y))$.   We claim that 
\begin{equation}\label{phiine}
\varphi_{\y}(\x)\le \nu(\y)\nu^*(\y).  
\end{equation}
Suppose first that $\varphi_{\y}(\x)\le 0$.  Then \eqref{phiine} trivially holds.
Assume that $\varphi_{\y}(\x)>0$.
By considering $t\x, t\ge 1$ such that $\nu^*(t\x)=\nu^*(\y)$ we can assume without loss of generality that $\nu^*(\x)=\nu*(y)$.
Use \eqref{defnustar} to deduce that 
$$\varphi_\y(\frac{1}{\nu(\y)}\x)=\Re\langle \y_1,\x\rangle \le \nu^*(\x)=\nu^*(\y),$$
which proves \eqref{phiine}.

Assume now that $\varphi_\y(\x)\le \nu(\y)\nu^*(\y)$ is a supporting hyperplane of $\rB_{\nu^*}(\0,\nu^*(\y))$ at $\y$.  Then $\nu(\y)\nu^*(\y)=\varphi(\y)$.

As $\nu^{**}=\nu$ we deduce $(2)$ from $(1)$ from by replacing $\nu$ with $\nu^*$.
\end{proof}
\begin{theorem}\label{albetm}  Let
\begin{equation}\label{albedef} 
\alpha(\nu)=\min_{\|\x\|=1}\nu(\x), \quad \beta(\nu)=\max_{\|\x\|=1}\nu(\x).
\end{equation}
Then
\begin{equation}\label{albetm}  
\alpha(\nu^*)=\frac{1}{\beta(\nu)}, \quad \beta(\nu^*)=\frac{1}{\alpha(\nu)}.
\end{equation}
Furthermore, 
\begin{equation}\label{necalbeteq}
\begin{aligned}
&(\|\x\|=1)\wedge(\nu(\x)=\alpha(\nu))\Rightarrow \nu(\x)\nu^*(\x)=1,\\
&(\|\y\|=1)\wedge(\nu(\y)=\beta(\nu))\Rightarrow \nu(\y)\nu^*(\y)=1.
\end{aligned}
\end{equation}
In particular, the following implications hold for $\|\x\|=\|\y\|=1$:
\begin{equation}\label{maxmineq}
\begin{aligned}
&\alpha(\nu)=\nu(\x)\iff \beta(\nu^*)=\nu^*(\x),\\
&\beta(\nu)=\nu(\y)\iff\alpha(\nu^*)=\nu(\y).
\end{aligned}
\end{equation}
\end{theorem}
\begin{proof} Observe that \eqref{albedef}  is equivalent to the sharp inequalities
\begin{equation}\label{nueucest}
\alpha(\nu)\|\x\|\le \nu(\x)\le \beta(\nu)\|\x\| \textrm{ for all }\x\in\F^n.
\end{equation}
Combine the  above inequality with  the equality $\|\x\|^*=\|\x\|$ to deduce 
$$\nu^*(\x) =\max_{\y\ne \0}\frac{|\langle \y,\x\rangle|}{\nu(\y)}\Rightarrow \frac{1}{\beta(\nu)}\|\x\|\le \nu^*(\x)\le \frac{1}{\alpha(\nu)} \|\x\|.$$
We claim that the above inequalities are sharp.  If not, apply the above inequalities to $\nu=\nu^{**}$ to obtain a contradiction to the sharp inequalities in \eqref{nueucest}.  This establishes \eqref{albetm}. 

We first prove the first implication in \eqref{necalbeteq}.   The inequality $\nu(\x)\ge \alpha(\nu)\|\x\|$ yields $\rB_{\nu}(0,\alpha(\nu))\subset \rB(0,1)$.
Assume that $\|\x\|=1$ and $\nu(\x)=\alpha(\nu)$.  That is,  and $\x\in\partial\rB_{\nu}(0,\alpha(\nu))\cap  \partial\rB(0,1)$.   Let $\varphi_\x=\Re \langle \x,\y\rangle$.  Then the supporting hyperplane of $\rB(0,1)$ at $\x$ is $\varphi_x(\y)\le 1$.  Hence,   $\varphi_x(\y)\le 1$ is also a supporting hyperplane of $\rB_{\nu}(0,\alpha(\nu))$ at $\x$.   As $\nu=\nu^{**}$, any supporting hyperplane of $\rB_{\nu}(0,\alpha(\nu))=\rB_{\nu}(0,\nu(x))$ at $\x$ is of the form $\varphi_\z\le \nu(\x)$ for some $z, \nu^*(\z)=1$.  That is,  $\z=\frac{1}{\nu^*(\x)}\x$.  Hence, 
$$\nu(\x)=\Re\langle \z,\x\rangle\Rightarrow \nu(\x)\nu^*(\x)=\langle \x,\x\rangle.$$
We now prove the second  implication in \eqref{necalbeteq}.    Assume that $\|\y\|=1$ and $\nu(\y)=\beta(\nu)$.  Then, $\rB_{\nu}(0,\beta(\nu))\supset \rB(0,1)$, and $\y\in\partial\rB_{\nu}(0,\beta(\nu))\cap  \partial\rB(0,1)$.   Let $\varphi_\z$ be the supporting hyperplane of 
$\rB_{\nu}(0,\nu(\y))$ at $\y$.  That is, $\varphi_\z(\x)=\Re\langle \z,\x\rangle$, where $\nu^*(\z)=1$, and $\varphi_\z(\x)\le \varphi_\z(\y)=\nu(\y)$ for $\x\in \rB_{\nu}(0,\nu(\y))$.   Hence,  $\varphi_{\frac{1}{\nu(\y)}\z}(\x)\le 1$ is a supporting hyperplane of $\rB(0,1)$ at $\y$.  Recall that the unique supporting hyperplane of $\rB(0,1)$ at $\y$ is $\varphi_\y(\x)=\Re\langle \y,\x\rangle \le 1$.  Therefore $\z=\nu(\y)\y$, and 
$$\langle \y,\y\rangle=1=\nu^*(\z)=\nu^*(\nu(\y)\y)=\nu(\y)\nu^*(\y).$$
Implications \eqref{maxmineq} follows straightforward from \eqref{necalbeteq}.
\end{proof}
\subsection{The $\nu$$\rank$ of a vector}\label{subsec:nurank}
\begin{definition}  Let $\nu$ be a norm on $\F^n$.   For a vector $\0\ne\x\in\F^n$ the $\nu$$\rank$ of $\x$ is the minimal $r$ such that the unit vector $\frac{1}{\nu(\x)}\x$ is a convex combination of $r$ extreme points of $\rB_{\nu}(0,1)$, denoted as $\nu$$\rank \x$.   Let $\nu$$\rank\0=0$.
The maximal $\nu$$\rank\x$ is denoted by $\max$$\nu\mathrm{rank}$.
\end{definition}
Caratheodory's theorem yields that
\begin{equation}\label{maxrnkin}
\mathrm{max}\nu\mathrm{rank}\le \begin{cases}
n+1 \textrm{ if }\F=\R,\\
2n+1 \textrm{ if }\F=\C.
\end{cases}
\end{equation}

Recall that a face $\bF$ of $\rB_{\nu}(\0,1)$ is a closed convex set $\bF\subset\partial\rB_{\nu}(\0,1)$  such that for any $\x\ne \y\in \rB_{\nu}(\0,1)$, such that if $[\x,\y]\cap \bF$ contains an interior point $z\in [\x,\y]$ then $[\x,\y]\subset \bF$.  A face $\bF$ is called exposed if there exists $\z\in\F^n, \nu^*(\z)=1$ such that 
\begin{equation}\label{defexpfce}
\bF=\{\x\in \F^n: \nu(\x)=1, \Re\langle\z,\x\rangle =1\}.
\end{equation}
It is straightforward to show that any maximal face of $\rB_{\nu}(\0,1)$ is exposed.
\begin{lemma}\label{charmaxnurank}
Let $\nu$ be a norm on $\F^n$.  Then
\begin{enumerate}
\item The set of the extreme points $\rE$ of $\rB_{\nu}(\0,1)$ is closed if and only if $\nu\mathrm{rank}$ is lower semicontinuous.
\item The maximum dimension of an exposed face of $\rB_{\nu}(\0,1)$ is an upper bound for $\mathrm{max}\nu\mathrm{rank}-1$.
\end{enumerate}
\end{lemma}
\begin{proof}  \emph{(1)}  Assume that $\rE$ is not closed.  Then there exists a sequence of $x_i,i\in\rE, i\in\N$ such that $\lim_{i\to\infty}\x_i=x, \nu(\x)=1$,  such that 
$\x\not\in\rE$.  Clearly, $\nu\rank \x_i=1, i\in\N$ and $\nu\rank\x\ge 2$.  Hence $\nu$rank is not lower semicontinuous. 

Assume that $\rE$ is a closed set.  Let $\x\ne 0$, and assume that $\lim_{i\to\infty}\x_i=\x$, and $\lim_{i\to\infty}\nu\textrm{rank}\,\x_i=k\in\N$.   Hence,  $\nu\textrm{rank}\,\x_i=k$ for $i\ge N$.  Thus $\frac{1}{\nu(\x_i)}\x_i=\sum_{j=1}^k a_{i,j}y_{j,i}$, where $a_{i,j}>0,y_{j,i}\in\rE, j\in[k],\sum_{j=1}^k a_{i,j}=1$ and $i\ge N$.  By taking a subsequence of $i\ge N$ and using the compactness of $\rE$ and the interval $[0,1]$ we deduce that $\frac{1}{\nu(\x)}\x=\sum_{j=1}^k a_j\y_j$, where $a_j\ge 0, \y_j\in\rE, \sum_{j=1}^k a_j=1$.

\noindent\emph{(2)}
Assume that $\nu(\x)=1$, and max$\nu$rank$=\nu\rank\x=r$.  Thus $\x=\sum_{i=1}^r a_i\x_i$, where $a_i>0,\x_i\in\rE(\nu), i\in[r], \sum_{i=1}^r a_i=1$.   Observe that cospan$(\x_1,\ldots,\x_r)$ is contained in a maximal face $\cF$ of $\rB_{\nu}(\0,1)$, whose extreme points are $\cF\cap \rB(\0,1)$.  As $\cF$ is compact, Caratheodory's theorem yields that $\x$ is a convex combination of at most $\dim\cF+1$ extreme points of $\cF$.
\end{proof}

\end{document}